\let\oldvec\vec
\documentclass{llncs}
\let\vec\oldvec

\usepackage[utf8]{inputenc}
\usepackage{mdframed, float}
\usepackage{microtype,lmodern,amsmath,amsthm,amsfonts,longtable,etoolbox,tikz, colortbl, mathtools}
\usetikzlibrary{graphs,quotes,decorations.pathreplacing,petri}
\usepackage{csquotes}
\usepackage{multirow}
\usepackage{array}
\usepackage{pbox}
\usepackage{caption}
\newlength{\problemoffset}
\newcommand{\escale}[1]{\ensuremath{\textbf{\scalebox{0.8}{#1}}}}
\newcommand{\nscale}[1]{\ensuremath{\textbf{\scalebox{0.8}{#1}}}}
\newcommand{\myEdge}[2]{ \tikz[baseline=-1pt]{
\draw[#2,line width=0.3pt] (0,0) -- ++(0.6,0) node[anchor=base, yshift=3pt, pos=0.5] {\escale{$#1$}};
}}
\newcommand{\Edge}[1]{ \tikz[baseline=-1pt]{
\draw[->,line width=0.3pt] (0,0) -- ++(0.6,0) node[anchor=base, yshift=5pt, pos=0.5] {\escale{$#1$}};
}}
\newcommand{\lEdge}[1]{ \tikz[baseline=-1pt]{
\draw[->,line width=0.3pt] (0,0) -- ++(1,0) node[anchor=base, yshift=5pt, pos=0.5] {\escale{$#1$}};
}}

\newcommand{\edge}[1]{\myEdge{#1}{->}}

\newcommand{\nop}{\ensuremath{\textsf{nop}}}
\newcommand{\inp}{\ensuremath{\textsf{inp}}}
\newcommand{\out}{\ensuremath{\textsf{out}}}
\newcommand{\set}{\ensuremath{\textsf{set}}}
\newcommand{\res}{\ensuremath{\textsf{res}}}
\newcommand{\swap}{\ensuremath{\textsf{swap}}}
\newcommand{\free}{\ensuremath{\textsf{free}}}
\newcommand{\used}{\ensuremath{\textsf{used}}}

\newcommand{\exit}{\ensuremath{\textsf{exit}}}
\newcommand{\enter}{\ensuremath{\textsf{enter}}}
\newcommand{\keepo}{\ensuremath{\textsf{keep}^+}}
\newcommand{\keepze}{\ensuremath{\textsf{keep}^-}}

\begin{document}
\title{The Complexity of Synthesizing \nop-Equipped Boolean Nets from $g$-Bounded Inputs \\ (Technical Report)
}
\author{Ronny Tredup}
\institute{Universit\"at Rostock, Institut f\"ur Informatik, Theoretische Informatik, Albert-Einstein-Stra\ss e 22, 18059, Rostock}
\maketitle

\begin{abstract}
Boolean Petri nets equipped with \nop\ allow places and transitions to be independent by being related by $\nop$.
We characterize for any fixed $g\in \mathbb{N}$ the computational complexity of synthesizing \nop-equipped Boolean Petri nets from labeled directed graphs whose states have at most $g$ incoming and at most $g$ outgoing arcs.
\end{abstract}

\section{Introduction}%

Boolean Petri nets are a basic model for the description of distributed and concurrent systems.
These nets allow at most one token on each place $p$ in every reachable marking.
Thus, $p$ is considered a Boolean condition that is \emph{true} if $p$ is marked and \emph{false} otherwise.
A place $p$ and a transition $t$ of a Boolean Petri net $N$ are related by one of the following Boolean \emph{interactions}: 
\emph{no operation} (\nop), \emph{input} (\inp), \emph{output} (\out), \emph{unconditionally set to true} (\set), \emph{unconditionally
reset to false} (\res), \emph{inverting} (\swap), \emph{test if true} (\used), and \emph{test if false} (\free).
The relation between $p$ and $t$ determines which conditions $p$ must satisfy to allow $t$'s firing and which impact has the firing of $t$ on $p$:
The interaction $\inp$ ($\out$) defines that $p$ must be \emph{true} (\emph{false}) first and \emph{false} (\emph{true}) after $t$ has fired.
If $p$ and $t$ are related by $\free$ ($\used$) then $t$'s firing proves that $p$ is \emph{false} (\emph{true}).
The interaction $\nop$ says that $p$ and $t$ are independent, that is, neither need $p$ to fulfill any condition nor has the firing of $t$ any impact on $p$.
If $p$ and $t$ are related by $\res$ ($\set$) then $p$ can be both \emph{false} or \emph{true} but after $t$'s firing it is \emph{false} (\emph{true}).
Also, the interaction $\swap$ does not require that $p$ satisfies any particular condition to enable $t$.
Here, the firing of $t$ inverts $p$'s Boolean value.

Boolean Petri nets are classified by the interactions of $I$ that they use to relate places and transitions.
More exactly, a subset $\tau\subseteq I$ is called a \emph{type of net} and a net $N$ is of type $\tau$ (a $\tau$-net, for short) if it applies at most the interactions of $\tau$.
So far, research has explicitly discussed seven Boolean Petri net types:
\emph{Elementary net systems} $\{\nop, \inp, \out\}$~\cite{DBLP:conf/ac/RozenbergE96},  \emph{Contextual nets} $\{\nop, \inp, \out, \used, \free\}$~\cite{DBLP:journals/acta/MontanariR95}, \emph{event/condition nets} $\{\nop, \inp, \out, \used\}$~\cite{DBLP:series/txtcs/BadouelBD15}, \emph{inhibitor nets} $\{\nop, \inp, \out, \free\}$~\cite{DBLP:conf/apn/Pietkiewicz-Koutny97}, \emph{set nets} $\{\nop, \inp, \set, \used\}$~\cite{DBLP:journals/acta/KleijnKPR13}, \emph{trace nets} $\{\nop, \inp, \out, \set, \res, \used, \free\}$~\cite{DBLP:journals/acta/BadouelD95}, and \emph{flip flop nets} $\{\nop, \inp, \out, \swap\}$~\cite{DBLP:conf/stacs/Schmitt96}.
However, since we have eight interactions to choose from, there are actually a total of 256 different types.

This paper addresses the computational complexity of the $\tau$-\emph{synthesis} problem. 
It consists in deciding whether a given directed labeled graph $A$, also called \emph{transition system}, is isomorphic to the reachability graph of a $\tau$-net $N$ and in constructing $N$ if it exists.
It has been shown that $\tau$-\emph{synthesis} is NP-complete if $\tau=\{\nop,\inp,\out\}$ ~\cite{DBLP:journals/tcs/BadouelBD97}, even if the inputs are strongly  restricted~\cite{DBLP:conf/concur/TredupR18,DBLP:conf/apn/TredupRW18}. 
On the contrary, in~\cite{DBLP:conf/stacs/Schmitt96}, it has been shown that it becomes polynomial if $\tau=\{\nop, \inp, \out, \swap\}$.
These opposing results motivate the question which interactions of $I$ make the synthesis problem hard and which make it tractable.
In our previous work of \cite{DBLP:journals/corr/abs-1909-05968,DBLP:conf/tamc/TredupR19,DBLP:conf/apn/TredupR19}, we answer this question partly and reveal the computational complexity of 120 of the 128 types that allow \nop.

In this paper, we investigate for fixed $g\in \mathbb{N}$ the computational complexity of $\tau$-synthesis restricted to $g$-bounded inputs, that is, every state of $A$ has at most $g$ incoming and at most $g$ outgoing arcs.
On the one hand, inputs of practical applications tend to have a low bound $g$ such as benchmarks of digital hardware design~\cite{jordiCortadella2017}.
On the other hand, considering restricted inputs hopefully gives us a better understanding of the problem's hardness.
Thus, $g$-bounded inputs are interesting from both the practical and the theoretical point of view.
In this paper, we completely characterize the complexity of $\tau$-synthesis restricted to $g$-bounded inputs for all types that allow places and transitions to be independent, that is, which contain $\nop$.
Figure~\ref{fig:summary} summarizes our findings:
For the types of \S 1 and \S 2, we showed hardness of synthesis without restriction in \cite{DBLP:conf/tamc/TredupR19}.
In this paper, we strengthen these results to $2$- and $3$-bounded inputs, respectively, and show that these bounds are tight.
The hardness result of the types of \S 3 origins from \cite{DBLP:conf/apn/TredupR19}.
This paper shows that a bound less than $2$ makes synthesis tractable.
Hardness for the types of \S 4 to \S 8 has been shown for $2$-bounded inputs in \cite{DBLP:conf/apn/TredupR19}.
In this paper, we strengthen this results to $1$-bounded inputs.
The hardness part for the types of \S 9 origin from \cite{DBLP:journals/corr/abs-1909-05968}.
In this paper, we argue that the bound $2$ is tight.
Finally, while the results of \S 10 are new, the ones of \S 11 have been found in \cite{DBLP:conf/tamc/TredupR19}.

For all considered types $\tau$, the corresponding hardness results are based on a reduction of the so called \emph{cubic monotone one-in-three 3SAT} problem~\cite{DBLP:journals/dcg/MooreR01}.
All reductions follow a common approach that represents clauses by directed labelled paths.
Thus, this paper also contributes a very general way to prove NP-completeness of synthesis of Boolean types of nets.

\begin{figure}[t!]\centering
\begin{tabular}{ c | p{8cm}| c| c| c}
$\S$ &Type of net $\tau$ & $g$ & Complexity & \#
\\ \hline
 %
 1& \multirow{2}*{\pbox{20cm}{$\{\nop, \inp, \free\}$, $\{\nop, \inp, \used, \free\}$, \\ $\{\nop, \out, \used\}$,  $\{\nop, \out, \used, \free\}$}} & $\geq 2$ & NP-complete & 4 \\   &  & $ < 2$ & polynomial &  \\ \hline
2 & \multirow{2}*{$\{\nop, \set, \res \} \cup \omega$ and $\emptyset \not=\omega \subseteq \{\used, \free\}$} & $\geq 3$ & NP-complete  & 3 \\  & & $ < 3$ & polynomial &  \\   \hline
\multirow{4}*{3} & \multirow{2}*{\pbox{20cm}{$\{\nop,\inp,\set\}$, $\{\nop,\inp,\set,\used\}$, \\ $\{\nop,\inp,\res,\set\}\cup\omega$ and $\omega\subseteq \{\out,\used,\free\}$, \\ $\{\nop,\out,\res\}$, $\{\nop,\out,\res,\free\}$, \\ $\{\nop,\out,\res,\set\}\cup\omega$ and $\omega\subseteq \{\inp,\used,\free\}$}} & $\geq 2$ & NP-complete & 16 \\  && $ < 2$ & polynomial&  \\ &&&&  \\ &&&&\\\hline
4 & $\{\nop, \inp, \out,\set \} \cup \omega$ or $\{\nop, \inp, \out,\res \} \cup \omega$ and \newline $\omega\subseteq \{\used,\free\}$ & $\geq 1$ & NP-complete & 8\\ \hline
 5 & $\{\nop, \inp,\set, \free \} $, $\{\nop, \inp,\set, \used, \free \} $, \newline $\{\nop, \out, \res, \used\} $, $\{\nop, \out, \res, \used,\free \} $  & $\geq 1$ &  NP-complete &  4 \\ \hline
6 & $\{\nop, \inp, \res,\swap \} \cup \omega$ or $\{\nop, \out, \set,\swap \} \cup \omega$ and \newline $\omega \subseteq \{\used, \free\}$ & $\geq 1$ & NP-complete  &  8  \\ \hline 
7 & $\{\nop, \inp, \set,\swap \} \cup \omega$ and $\omega \subseteq \{\out, \res,\used, \free\}$,  \newline $\{\nop, \out, \res,\swap \} \cup \omega$ and $\omega \subseteq \{\inp, \set,\used, \free\}$ & $\geq 1$ & NP-complete  & 28  \\ \hline 
8 & $\{\nop, \inp, \out \} \cup \omega$ and $\omega \subseteq \{\used, \free\}$& $\geq 1$& NP-complete  & 4  \\ \hline
9&\multirow{2}*{\pbox{20cm}{$\{\nop,\set,\swap\}\cup\omega$, $\{\nop,\res,\swap\}\cup\omega$,\\ $\{\nop,\res,\set, \swap\}\cup\omega$ and $\emptyset\not=\omega\subseteq\{\used,\free\}$}} & $\geq 2$ & NP-complete  & 9  \\ && $ < 2$ & polynomial &   \\ \hline
10& $\{\nop, \inp\}$, $\{\nop,\inp,\used\}$, $\{\nop,\out\}$, $\{\nop,\out,\free\}$ \newline
$\{\nop,\set,\swap\}$, $\{\nop,\res,\swap\}$, $\{\nop,\set,\res\}$,  \newline
$\{\nop,\set,\res,\swap\}$  & $\geq0$ & polynomial & 8\\\hline
11 & $\{\nop, \res\} \cup \omega$ and $\omega \subseteq \{\inp, \used, \free\}$, \newline 
 $\{\nop, \set\} \cup \omega$ and  $\omega \subseteq \{\out, \used, \free\}$, \newline 
 $\{\nop, \swap\} \cup \omega$ and  $\omega \subseteq \{\inp, \out, \used, \free\}$, \newline 
 $\{\nop\} \cup \omega$ and  $\omega \subseteq \{\used, \free\}$&  $\geq0$   & polynomial & 36
\end{tabular}
\caption{The computational complexity of Boolean net synthesis from $g$-bounded TS for all types that contain $\nop$.}
\label{fig:summary}
\end{figure}

\section{Preliminaries}%


\textbf{Transition Systems.}
A \emph{transition system} (TS, for short) $A=(S,E, \delta)$ is a directed labeled graph with states $S$, events $E$ and partial \emph{transition function} $\delta: S\times E \longrightarrow S$, where $\delta(s,e)=s'$ is interpreted as $s\edge{e}s'$.
For $s\edge{e}s'$ we say $s$ is a source and $s'$ is a sink of $e$, respectively.
An event $e$ \emph{occurs} at a state $s$, denoted by $s\edge{e}$, if $\delta(s,e)$ is defined.
An \emph{initialized} TS $A=(S,E,\delta, s_0)$ is a TS with a distinct state $s_0\in S$ where every state $s\in S$ is \emph{reachable} from $s_0$ by a directed labeled path.
TSs in this paper are \emph{deterministic} by design as their state transition behavior is given by a (partial) function.
Let $g\in \mathbb{N}$. 
An initialized TS $A$ is called $g$-bounded if for all $s\in S(A)$ the number of incoming and outgoing arcs at $s$ is restricted by $g$: $\vert \{e\in E(A) \mid  \edge{e}s\}\vert \leq g$ and $\vert \{e\in E(A) \mid  s\edge{e}\}\vert \leq g$. 

\begin{figure}\centering
\begin{tabular}{c|c|c|c|c|c|c|c|c}
$x$ & $\nop(x)$ & $\inp(x)$ & $\out(x)$ & $\set(x)$ & $\res(x)$ & $\swap(x)$ & $\used(x)$ & $\free(x)$\\ \hline
$0$ & $0$ & & $1$ & $1$ & $0$ & $1$ & & $0$\\
$1$ & $1$ & $0$ & & $1$ & $0$ & $0$ & $1$ & \\
\end{tabular}
\caption{
All interactions in $I$.
An empty cell means that the column's function is undefined on the respective $x$.
The entirely undefined function is missing in $I$.
}\label{fig:interactions}
\end{figure}

\textbf{Boolean Types of Nets~\cite{DBLP:series/txtcs/BadouelBD15}.}
The following notion of Boolean types of nets serves as vehicle to capture many Boolean Petri nets in a uniform way.
A \emph{Boolean type of net} $\tau=(\{0,1\},E_\tau,\delta_\tau)$ is a TS such that $E_\tau$ is a subset of the Boolean interactions:
$E_\tau \subseteq I = \{\nop, \inp, \out, \set, \res, \swap, \used, \free\}$. 
The interactions $i \in I$ are binary partial functions $i: \{0,1\} \rightarrow \{0,1\}$ as defined in Figure~\ref{fig:interactions}.
For all $x\in \{0,1\}$ and all $i\in E_\tau$ the transition function of $\tau$ is defined by $\delta_\tau(x,i)=i(x)$.
Notice that $I$ contains all binary partial functions $\{0,1\} \rightarrow \{0,1\}$ except for the entirely undefined function $\bot$. 
Even if a type $\tau$ includes $\bot$, this event can never occur, so it would be useless.
Thus, $I$ is complete for deterministic Boolean types of nets, and that means there are a total of 256 of them.
By definition, a Boolean type $\tau$ is completely determined by its event set $E_\tau$.
Hence, in the following we identify $\tau$ with $E_\tau$, cf. Figure~\ref{fig:type_of_nets}.
Moreover, for readability, we group interactions by $\enter = \{\out,\set,\swap\}$, $\exit = \{\inp,\res,\swap\}$, $\keepo = \{\nop,\set,\used\}$, and $\keepze = \{\nop,\res,\free\}$.

\begin{figure}[h!]

\centering
\begin{tikzpicture}[scale = 1.2]
\begin{scope}
\node (0) at (0,0) {\nscale{$0$}};
\node (1) at (2,0) {\nscale{$1$}};

\path (0) edge [->, out=-120,in=120,looseness=5] node[left, align =left] {\escale{$\nop$} \\ \escale{\free} \\ \escale{\res} } (0);
\path (1) edge [<-, out=60,in=-60,looseness=5] node[right, align=left] {\escale{$\nop$}  } (1);

\path (0) edge [<-, bend right= 30] node[below] {\escale{$\res, \swap$}} (1);
\path (0) edge [->, bend left= 30] node[above] {\escale{$\out,\swap$}} (1);
\end{scope}
\begin{scope}[xshift=4.5cm]
\node (0) at (0,0) {\nscale{$0$}};
\node (1) at (2,0) {\nscale{$1$}};

\path (0) edge [->, out=-120,in=120,looseness=5] node[left, align =left] {\escale{$\nop$} } (0);
\path (1) edge [<-, out=60,in=-60,looseness=5] node[right, align=left] {\escale{$\nop$} \\ \escale{\used} \\ \escale{\set}} (1);

\path (0) edge [<-, bend right= 30] node[below] {\escale{$\inp,\swap$}} (1);
\path (0) edge [->, bend left= 30] node[above] {\escale{$\set,\swap$}} (1);
\end{scope}
\end{tikzpicture}
\caption{
Left: $\tau=\{\nop, \out, \res, \swap, \free\}$.
Right: $\tilde{\tau}=\{\nop, \inp, \set, \swap, \used\}$.
$\tau$ and $\tilde{\tau}$ are isomorphic.
The isomorphism $\phi: \tau\rightarrow \tilde{\tau}$ is given by $\phi(s)=1-s$ for $s\in \{0,1\}$, $\phi(i)=i$ for $i\in \{\nop,\swap\}$, $\phi(\out)=\inp$, $\phi(\res)=\set$ and $\phi(\free)=\used$.}
\label{fig:type_of_nets}
\end{figure}
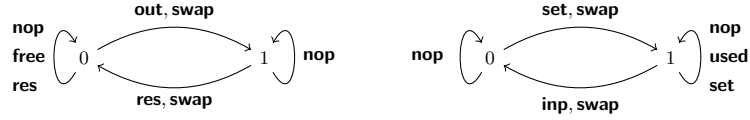

\textbf{$\tau$-Nets.}
Let $\tau\subseteq I$.
A Boolean Petri net $N = (P, T, H_0, f)$ of type $\tau$, ($\tau$-net, for short) is given by finite and disjoint sets $P$ of places and $T$ of transitions, an initial marking $H_0: P\longrightarrow  \{0,1\}$, and a (total) flow function $f: P \times T \rightarrow \tau$. 
A $\tau$-net realizes a certain behavior by firing sequences of transitions: 
A transition $t \in T$ can fire in a marking $M: P\longrightarrow  \{0,1\}$ if $\delta_\tau(M(p), f(p,t))$ is defined for all $p\in P$.
By firing, $t$ produces the next marking $M' : P\longrightarrow  \{0,1\}$ where $M'(p)=\delta_\tau(M(p), f(p,t))$ for all $p\in P$. 
This is denoted by $M \edge{t} M'$.
Given a $\tau$-net $N=(P, T, H_0, f)$, its behavior is captured by a transition system $A_N$, called the reachability graph of $N$.
The state set of $A_N$ consists of all markings that, starting from initial state $H_0$, can be reached by firing a sequence of transitions.
For every reachable marking $M$ and transition $t \in T$ with $M \edge{t} M'$ the state transition function $\delta$ of $A$ is defined as $\delta(M,t) = M'$.

\textbf{$\tau$-Regions.}
Let $\tau\subseteq I$.
If an input $A$ of $\tau$-synthesis allows a positive decision then we want to construct a corresponding $\tau$-net $N$ purely from $A$.
Since $A$ and $A_N$ are isomorphic, $N$'s transitions correspond to $A$'s events.
However, the notion of a place is unknown for TSs.
So called regions mimic places of nets:
A $\tau$-region of a given $A=(S, E, \delta, s_0)$ is a pair $(sup, sig)$ of \emph{support} $sup: S \rightarrow S_\tau = \{0,1\}$ and \emph{signature} $sig: E\rightarrow E_\tau = \tau$ where every transition $s \edge{e} s'$ of $A$ leads to a transition $sup(s) \edge{sig(e)} sup(s')$ of $\tau$.
While a region divides $S$ into the two sets $sup^{-1}(b) = \{s \in S \mid sup(s) = b\}$ for $b \in \{0,1\}$, the events are cumulated by $sig^{-1}(i) = \{e \in E \mid sig(e) = i\}$ for all available interactions $i \in \tau$.
We also use $sig^{-1}(\tau') = \{e \in E \mid sig(e) \in \tau'\}$ for $\tau' \subseteq \tau$.
A region $(sup, sig)$ models a place $p$ and the corresponding part of the flow function $f$.
In particular, $sig(e)$ models $f(p,e)$ and $sup(s)$ models $M(p)$ in the marking $M\in RS(N)$ corresponding to $s\in S(A)$. 
Every set $\mathcal{R} $ of $\tau$-regions of $A$ defines the \emph{synthesized $\tau$-net} $N^{\mathcal{R}}_A=(\mathcal{R}, E, f, H_0)$ with flow function $f((sup, sig),e)=sig(e)$ and initial marking $H_0((sup, sig))=sup(s_{0})$ for all $(sup, sig)\in \mathcal{R}, e\in E$.
It is well known that $A_{N^{\mathcal{R}}_A }$ and $A$ are isomorphic if and only if $\mathcal{R}$'s regions solve certain separation atoms \cite{DBLP:series/txtcs/BadouelBD15}, to be introduced next.
A pair $(s, s')$ of distinct states of $A$ define a \emph{state separation atom} (SSP atom, for short).
A $\tau$-region $R=(sup, sig)$ \emph{solves} $(s,s')$ if $sup(s)\not=sup(s')$.
The meaning of $R$ is to ensure that $N^{\mathcal{R}}_A$ contains at least one place $R$ such that $M(R)\not=M'(R)$ for the markings $M$ and $M'$ corresponding to $s$ and $s'$, respectively.
If there is a $\tau$-region that solves $(s,s')$ then $s$ and $s'$ are called \emph{$\tau$-solvable}.
If every SSP atom of $A$ is $\tau$-solvable then $A$ has the \emph{$\tau$-state separation property} ($\tau$-SSP, for short).
A pair $(e,s)$ of event $e\in E $ and state $s\in S$ where $e$ does not occur at $s$, that is $\neg s\edge{e}$, define an \emph{event state separation atom} (ESSP atom, for short).
A $\tau$-region $R=(sup, sig)$ \emph{solves} $(e,s)$ if $sig(e)$ is not defined on $sup(s)$ in $\tau$, that is, $\neg \delta_\tau(sup(s), sig(e))$.
The meaning of $R$ is to ensure that there is at least one place $R$ in $N^{\mathcal{R}}_A$ such that $\neg M\edge{e}$ for the marking $M$ corresponding to $s$.
If there is a $\tau$-region that solves $(e,s)$ then $e$ and $s$ are called \emph{$\tau$-solvable}.
If every ESSP atom of $A$ is $\tau$-solvable then $A$ has the \emph{$\tau$-event state separation property} ($\tau$-ESSP, for short).
A set $\mathcal{R}$ of $\tau$-regions of $A$ is called $\tau$-\emph{admissible} if for every of $A$'s (E)SSP atoms there is a $\tau$-region $R$ in $\mathcal{R}$ that solves it.
The following lemma, borrowed from \cite[p.163]{DBLP:series/txtcs/BadouelBD15}, summarizes the already implied connection between the existence of $\tau$-admissible sets of $A$ and (the solvability of) $\tau$-synthesis:
\begin{lemma}[\cite{DBLP:series/txtcs/BadouelBD15}]\label{lem:admissible} 
A TS $A$ is isomorphic to the reachability graph of a $\tau$-net $N$ if and only if there is a $\tau$-admissible set $\mathcal{R}$ of $A$ such that $N=N^{\mathcal{R}}_A$.
\end{lemma}
We say a $\tau$-net $N$ $\tau$-\emph{solves} $A$ if $A_N$ and $A$ are isomorphic.
By Lemma~\ref{lem:admissible}, deciding if $A$ is $\tau$-solvable reduces to deciding whether it has the $\tau$-(E)SSP.
Moreover, it is easy to see that if $\tau$ and $\tilde{\tau}$ are isomorphic then deciding the $\tau$-(E)SSP reduces to deciding the $\tilde{\tau}$-(E)SSP:
\begin{lemma}[Without proof]
\label{lem:isomorphy}
If $\tau$ and $\tilde{\tau}$ are isomorphic types of nets then a TS $A$ has the $\tau$-(E)SSP if and only if $A$ has the $\tilde{\tau}$-(E)SSP.
\end{lemma}
In particular, we benefit from the isomorphisms that map \nop\ to \nop, \swap\ to \swap, \inp\ to \out, \set\ to \res, \used\ to \free, and vice versa.
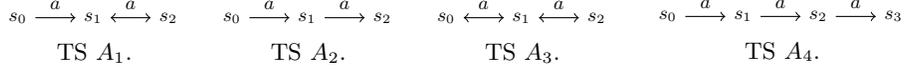
\begin{figure}[h!]
\begin{minipage}{0.2\textwidth}
\begin{center}
\begin{tikzpicture}[new set = import nodes]
\begin{scope}[nodes={set=import nodes}]

\node (0) at (0,0) {\nscale{$s_0$}};
\node (1) at (1,0) {\nscale{$s_1$}};
\node (2) at (2,0) {\nscale{$s_2$}};
\graph {
 (0)->[ "\escale{$a$}"] (1)<->["\escale{$a$}"] (2);%
};
\end{scope}
\end{tikzpicture}
TS $A_1$.
\end{center}
\end{minipage}
\hspace{0.2cm}
\begin{minipage}{0.2\textwidth}
\begin{center}
\begin{tikzpicture}[new set = import nodes]
\begin{scope}[nodes={set=import nodes}]

\node (0) at (0,0) {\nscale{$s_0$}};
\node (1) at (1,0) {\nscale{$s_1$}};
\node (2) at (2,0) {\nscale{$s_2$}};
\graph {
 (0)->[ "\escale{$a$}"] (1)->["\escale{$a$}"] (2);%
};
\end{scope}
\end{tikzpicture}
TS $A_2$.
\end{center}
\end{minipage}
\hspace{0.2cm}
\begin{minipage}{0.2\textwidth}
\begin{center}
\begin{tikzpicture}[new set = import nodes]
\begin{scope}[nodes={set=import nodes}]

\node (0) at (0,0) {\nscale{$s_0$}};
\node (1) at (1,0) {\nscale{$s_1$}};
\node (2) at (2,0) {\nscale{$s_2$}};
\graph {
 (0)<->[ "\escale{$a$}"] (1)<->["\escale{$a$}"] (2);%
};
\end{scope}
\end{tikzpicture}
TS $A_3$.
\end{center}
\end{minipage}
\hspace{0.2cm}
\begin{minipage}{0.3\textwidth}
\begin{center}
\begin{tikzpicture}[new set = import nodes]
\begin{scope}[nodes={set=import nodes}]

\node (0) at (0,0) {\nscale{$s_0$}};
\node (1) at (1,0) {\nscale{$s_1$}};
\node (2) at (2,0) {\nscale{$s_2$}};
\node (3) at (3,0) {\nscale{$s_3$}};
\graph {
 (0)->[ "\escale{$a$}"] (1)->["\escale{$a$}"] (2)->["\escale{$a$}"] (3);%
};
\end{scope}
\end{tikzpicture}
TS $A_4$.
\end{center}
\end{minipage}
\caption{Let $\tau=\{\nop,\set, \swap, \free\}$.
The TSs $A_1,\dots, A_4$ give examples for the presence and absence of the $\tau$-(E)SSP:
TS $A_1$ has the $\tau$-ESSP as $a$ occurs at every state.
It has also the $\tau$-SSP: 
The region $R=(sup, sig)$ where $sup(s_0)=sup(s_2)=1$, $sup(s_1)=0$ and $sig(a)=\swap$ separates the pairs $s_0,s_1$ and $s_2, s_1$.
Moreover, the region $R'=(sup', sig')$ where $sup'(s_0)=0$ and $sup'(s_1)=sup'(s_2)=1$ and $sig'(a)=\set$ separates $s_0$ and $s_1$. 
Notice that $R$ and $R'$ can be translated into $\tilde{\tau}$-regions, where $\tilde{\tau}=\{\nop, \res, \swap, \used\}$, via the isomorphism of Figure~\ref{fig:type_of_nets}.
For example, if $s\in S(A_1)$ and $e\in E(A_1)$ and $sup''(s)=\phi(sup(s))$ and $sig''(e)=\phi(sig(e))$ then the resulting $\tilde{\tau}$-region $R''=(sup'', sig'')$ separates $s_0, s_1$ and $s_2,s_1$.
Thus, $A_1$ has also $\tilde{\tau}$-(E)SSP. 
TS $A_2$ has the $\tau$-SSP but not the $\tau$-ESSP as event $a$ is not inhibitable at the state $s_2$.
TS $A_3$ has the $\tau$-ESSP ($a$ occurs at every state) but not the $\tau$-SSP as $s_1$ and $s_2$ are not separable.
TS $A_4$ has neither the $\tau$-ESSP nor the $\tau$-SSP.
}
\label{fig:regions}
\end{figure}

\section{Hardness Results}%

In this section, we show in accordance to the summary of Figure~\ref{fig:summary} for several types of nets $\tau\subseteq I$ and fixed $g\in \mathbb{N}$ that $\tau$-synthesis NP-complete even if the input TS $A$ is $g$-bounded.
Since $\tau$-synthesis is known to be in NP for all $\tau\subseteq I$ \cite{DBLP:conf/apn/TredupR19}, we restrict ourselves to the hardness part. 
All proofs base on a reduction of the problem \emph{cubic monotone one-in-three $3$-SAT} which has been shown to be NP-complete in~\cite{DBLP:journals/dcg/MooreR01}.
The input for this problem is a Boolean expression $\varphi=\{\zeta_0,\dots, \zeta_{m-1}\}$ of $m$ negation-free three-clauses $\zeta_i=\{X_{i_0}, X_{i_1}, X_{i_2}\}$ such that every variable $X\in V(\varphi)=\bigcup_{i=0}^{m-1}\zeta_i$ occurs in exactly three clauses.
Notice that the latter implies $\vert V(\varphi)\vert =m$.
Moreover, we assume without loss of generality that if $\zeta_i=\{X_{i_0}, X_{i_1}, X_{i_2}\}$ then $i_0 < i_1 < i_2$.
The question to answer is whether there is a subset $M\subseteq V(\varphi)$ with $\vert M\cap \zeta_i\vert =1$ for all $i\in \{0,\dots, m-1\}$.
For all considered types of nets $\tau$ and corresponding bounds $g$, we reduce a given instance $\varphi$ to a $g$-bounded TS $A^\tau_\varphi$ such that the following two conditions are true:
Firstly, the TS $A^\tau_\varphi$ has an ESSP atom $\alpha$ which is $\tau$-solvable if and only if there is a one-in-three model $M$ of $\varphi$.
Secondly, if the ESSP atom $\alpha$ is $\tau$-solvable then all ESSP and SSP atoms of $A^\tau_\varphi$ are also $\tau$-solvable. 
A reduction that satisfies these conditions proves the NP-hardness of $\tau$-synthesis as follows:
If $\varphi$ has a one-three-model then the conditions ensure that the TS $A^\tau_\varphi$ has the $\tau$-(E)SSP and thus is $\tau$-solvable.
Reversely, if $A^\tau_\varphi$ is $\tau$-solvable then, by definition, it has the $\tau$-ESSP.
In particular, there is a $\tau$-region that solves $\alpha$ which, by the first condition, implies that $\varphi$ has a one-in-three model.
Consequently, $A^\tau_\varphi$ is $\tau$-solvable if and only if  $\varphi$ has a one-in-three model.
For all types considered, the proof that the solvability of $\alpha$ implies the (E)SSP has been transferred to the appendix (section A), so as not to affect the readability of the paper. 

A key idea, applied by all reductions in one way or another, is the representation of every clause $\zeta_i=\{X_{i_0}, X_{i_1}, X_{i_2}\}$, $i\in \{0,\dots, m-1\}$, by a directed labeled path of $A^\tau_\varphi$ on which the variables of $\zeta_i$ occur as events:
\[ 
s_{i,0}\dots s_{i,j}\edge{X_{i_0}}s_{i,j+1}\dots s_{i,j'}\edge{X_{i_1}}s_{i,j'+1}\dots s_{i,j''}\edge{X_{i_2}}s_{i,j''+1}\dots s_{i,n
}\] 
The reductions ensure that if a $\tau$-region $(sup, sig)$ solves the atom $\alpha$ then $sup(s_{i,0})\not=sup(s_{i,n})$.
This makes the image of this path under $(sup, sig)$ a directed path from $0$ to $1$ or from $1$ to $0$ in $\tau$.
Thus, there has to be an event $e$ on the path that causes the state change from $sup(s_{i,0})$ to $sup(s_{i,n})$ by $sig(e)$.
We exploit this property and ensure that this state change is unambiguously done by (the signature of) exactly one variable event per clause.
As a result, the corresponding variable events define a searched model of $\varphi$ via their signature.
The proof of the following theorem gives a first example of this approach and Figure~\ref{fig:full_example} shows a full example reduction.
\begin{theorem}\label{the:nop_inp_free}
For any fixed $g\geq 2$, deciding if a $g$-bounded TS $A$ is $\tau$-solvable is NP-complete if $\tau=\{\nop,\inp,\free\}$, $\tau=\{\nop,\inp,\used, \free\}$, $\tau=\{\nop,\out,\used\}$ and $\tau=\{\nop,\out,\used, \free\}$.
\end{theorem}
\begin{figure}
\begin{center}
\begin{tikzpicture}[new set = import nodes]
\begin{scope}[nodes={set=import nodes}]%
			
		\foreach \i in {0,...,4} { \coordinate (0\i) at (0,\i*-0.9cm) ;}
		\foreach \i in {0,...,4} { \coordinate (1\i) at (1.8,\i*-0.9cm) ;}
		\foreach \i in {0,...,4} { \coordinate (2\i) at (3.6,\i*-0.9cm) ;}
		\foreach \i in {0,...,4} { \coordinate (3\i) at (5.4,\i*-0.9cm) ;}
		\foreach \i in {0,...,4} { \coordinate (4\i) at (7.2,\i*-0.9cm) ;}
		\foreach \i in {0,...,4} { \coordinate (5\i) at (9,\i*-0.9cm) ;}
		\foreach \i in {0,...,2} { \coordinate (m\i) at (10.8,\i*-0.9cm) ;}
		\foreach \i in {0,...,6} { \coordinate (b\i) at (\i*1.8, 0.9cm) ;}
		\foreach \i in {00,10,20,30,40,50,31,41,51,b0, b1,b2, b3, b4, b5, b6,m0} {\fill[red!15] (\i) circle (0.35cm);}
		
		\foreach \i in {0,...,4} { \node (0\i) at (0\i) {\nscale{$t_{0,\i}$}};}
		\node (05) at (0.9,0) {\nscale{$t_{0,5}$}};
		\foreach \i in {0,...,4} { \node (1\i) at (1\i) {\nscale{$t_{1,\i}$}};}
		\node (15) at (2.7,0) {\nscale{$t_{1,5}$}};
		\foreach \i in {0,...,4} { \node (2\i) at (2\i) {\nscale{$t_{2,\i}$}};}
		\node (25) at (4.5,0) {\nscale{$t_{2,5}$}};
		\foreach \i in {0,...,4} { \node (3\i) at (3\i) {\nscale{$t_{3,\i}$}};}
		\node (35) at (6.3,0) {\nscale{$t_{3,5}$}};
		\foreach \i in {0,...,4} { \node (4\i) at (4\i) {\nscale{$t_{4,\i}$}};}
		\node (45) at (8.1,0) {\nscale{$t_{4,5}$}};
		\foreach \i in {0,...,4} { \node (5\i) at (5\i) {\nscale{$t_{5,\i}$}};}
		\node (55) at (9.9,0) {\nscale{$t_{5,5}$}};
		\foreach \i in {0,...,2} { \node (m\i) at (m\i) {\nscale{$h_{\i}$}};}
		\foreach \i in {0,...,6} { \node (b\i) at (b\i) {\nscale{$\bot_\i$}};  }
\graph {
	(import nodes);
			00 ->[swap,"\escale{$X_0$}"]01->[swap,"\escale{$X_1$}"]02->[swap,"\escale{$X_2$}"]03->[swap,"\escale{$k_1$}"]04;  
			00 ->[swap,"\escale{$k_0$}"]05;
			10 ->[swap,"\escale{$X_0$}"]11->[swap,"\escale{$X_2$}"]12->[swap,"\escale{$X_3$}"]13->[swap,"\escale{$k_1$}"]14;  
			10 ->[swap,"\escale{$k_0$}"]15;
			20 ->[swap,"\escale{$X_0$}"]21->[swap,"\escale{$X_1$}"]22->[swap,"\escale{$X_3$}"]23->[swap,"\escale{$k_1$}"]24;  
			20 ->[swap,"\escale{$k_0$}"]25;
			30 ->[swap,"\escale{$X_2$}"]31->[swap,"\escale{$X_4$}"]32->[swap,"\escale{$X_5$}"]33->[swap,"\escale{$k_1$}"]34;  
			30 ->[swap,"\escale{$k_0$}"]35;
			40 ->[swap,"\escale{$X_1$}"]41->[swap,"\escale{$X_4$}"]42->[swap,"\escale{$X_5$}"]43->[swap,"\escale{$k_1$}"]44;  
			40 ->[swap,"\escale{$k_0$}"]45;
			50 ->[swap,"\escale{$X_3$}"]51->[swap,"\escale{$X_4$}"]52->[swap,"\escale{$X_5$}"]53->[swap,"\escale{$k_1$}"]54;  
			50 ->[swap,"\escale{$k_0$}"]55;
			m0 ->[swap,"\escale{$k_0$}"]m1->[swap,"\escale{$k_1$}"]m2;
			b0 ->["\escale{$\ominus_1$}"]b1->["\escale{$\ominus_2$}"]b2->["\escale{$\ominus_3$}"]b3->["\escale{$\ominus_4$}"]b4->["\escale{$\ominus_5$}"]b5->["\escale{$\ominus_6$}"]b6;
			b0 ->[swap,"\escale{$\oplus_0$}"]00; b1 ->[swap,"\escale{$\oplus_1$}"]10; b2 ->[swap,"\escale{$\oplus_2$}"]20;
			b3 ->[swap,"\escale{$\oplus_3$}"]30; b4 ->[swap,"\escale{$\oplus_4$}"]40; b5 ->[swap,"\escale{$\oplus_5$}"]50;
			b6 ->[swap,"\escale{$\oplus_6$}"]m0; 
			
			};
\end{scope}
\end{tikzpicture}
\end{center}
\caption{
The TS $A^\tau_\varphi$ for $\varphi=\{\zeta_0,\dots, \zeta_{5}\}$ with clauses $\zeta_0=\{X_0,X_1,X_2\},\ \zeta_1= \{X_0,X_2,X_3\},\ \zeta_2= \{X_0,X_1,X_3\},\ \zeta_3= \{X_2,X_4,X_5\},\ \zeta_4=\{X_1,X_4,X_5\},\ \zeta_5= \{X_3,X_4,X_5\}$ .
The red colored area sketches the $\tau$-region $(sup, sig)$ that solves $(k_1,m_0)$ and corresponds to the one-in-three model $M=\{X_0,X_4\}$.
}\label{fig:full_example}
\end{figure}

\begin{proof}
We argue for $\tau\in \{\{\nop,\inp,\free\}, \{\nop,\inp,\used, \free\} \}$, which by Lemma~\ref{lem:isomorphy} proves the claim for the other types, too.

Firstly, the TS $A^\tau_\varphi$ has the following gadget $H$ (left hand side) that provides the events $k_0, k_1$ and the atom $\alpha=(k_1, m_0)$.
Secondly, it has for every clause $\zeta_i=\{X_{i_0}, X_{i_1}, X_{i_2}\}$ the following gadget $T_i$ (right hand side) that applies $k_0$, $k_1$ and $\zeta_i's$ variables as events.
\begin{center}
\begin{tikzpicture}[new set = import nodes]
\begin{scope}[nodes={set=import nodes}]%

		\foreach \i in {0,...,2} { \coordinate (\i) at (\i*1.5cm,0) ;}
		\foreach \i in {0,...,2} { \node (\i) at (\i) {\nscale{$m_{\i}$}};}
		\graph {
	(import nodes);
			0 ->["\escale{$k_0$}"]1->["\escale{$k_1$}"]2;
			};
\end{scope}
\begin{scope}[xshift=4.5cm, nodes={set=import nodes}]%

		\foreach \i in {0,...,4} { \coordinate (\i) at (\i*1.3cm,0) ;}
		\foreach \i in {0,...,4} { \node (\i) at (\i) {\nscale{$t_{i,\i}$}};}
		\node (5) at (0,-1.2) {\nscale{$t_{i,5}$}};
\graph {
	(import nodes);
			0 ->["\escale{$X_{i_0}$}"]1->["\escale{$X_{i_1}$}"]2->["\escale{$X_{i_2}$}"]3->["\escale{$k_1$}"]4;  
			0 ->["\escale{$k_0$}"]5;
			};
\end{scope}
\end{tikzpicture}
\end{center}
Finally, $A^\tau_\varphi$ uses the states $\bot_0,\dots, \bot_m$ and events $\ominus_1,\dots \ominus_m$ and $\oplus_0,\dots,\oplus_m$ to join the gadgets $T_0,\dots, T_{m-1}$ and $H$ by $\bot_{i}\edge{\ominus_{i+1}}\bot_{i+1}$ and $\bot_i\edge{\oplus_i}t_{i,0}$, for all $i\in \{0,\dots, m-1\}$, and $\bot_m\edge{\oplus_m}h_0$, cf. Figure~\ref{fig:full_example}.

The gadget $H$ ensures that if $(sup, sig)$ is a region that solves $\alpha$ then $sup(h_0)=1$ and $sig(k_1)=\free$ which implies $sup(h_1)=0$ and $sig(k_0)=\inp$.
This is because $sig(k_1)\in \{\inp,\used\}$ and $sup(h_0)=0$ implies $sig(k_0)\in \{\out,\set,\swap\}$, which is impossible.
Consequently, $s\edge{k_0}$ and $s'\edge{k_1}$ imply $sup(s)=1$ and $sup(s')=0$, respectively.
The TS $A^\tau_\varphi$ uses these properties to ensure via $T_0,\dots, T_{m-1}$ that the region $(sup, sig)$ implies a one-in-three model of $\varphi$.

More exactly, if $i\in \{0,\dots, m-1\}$ then for $T_i$ we have by $t_{i,0}\edge{k_0}$ and $t_{i,3}\edge{k_1}$  that $sup(t_{i,0})=1$ and $sup(t_{i,3})=0$.
Thus, there is an event $X_{i_j}$, where $j\in \{0,1,2\}$, such that $sig(X_{i_j})=\inp$.
Clearly, if $sig(X_{i_j})=\inp$ then $sig(X_{i_\ell})\not=\inp$ for all $j < \ell\in \{0,1,2\}$ as $X_{i_\ell}$'s sources have a $0$-support.
Consequently, there is \emph{exactly one} variable event $X\in \zeta_i$ such that $sig(X)=\inp$.
Since $i$ was arbitrary, this is simultaneously true for all clauses $\zeta_0,\dots, \zeta_{m-1}$.
Thus, the set $M=\{X\in V(\varphi) \mid sig(X)=\inp\}$ is a one in three model of $\varphi$.

In reverse, if $\varphi$ is one-in-three satisfiable then there is a $\tau$-region $(sup, sig)$ of $A^\tau_\varphi$ that solves $\alpha$.
In particular, if $M$ is a one-in-three model of $\varphi$ then we first define $sup(\bot_0)=1$.
Secondly, for all $e\in E(A^\tau_\varphi)$ we define $sig(e)=\free$ if $e=k_1$, $sig(e)=\inp$ if $e\in \{k_0\}\cup M$ and else $sig(e)=\nop$.
Since $A^\tau_\varphi$ is reachable, by inductively defining $sup(s_{i+1})=\delta_\tau(sup(s_i), sig(e_i))$ for all paths $\bot_0 \edge{e_1}s_1\dots s_{n-1}\edge{e_n}s_{n}$, this defines a fitting region $(sup, sig)$, cf. Figure~\ref{fig:full_example}.

This proves that $\alpha$ is $\tau$-solvable if and only if $\varphi$ is one-in-three satisfiable.
\end{proof}
In the remainder of this section, we present the remaining hardness results in accordance to Figure~\ref{fig:summary} and the corresponding reductions that proves them.
\begin{theorem}\label{lem:nop_set_res+used_free}
For any fixed $g\geq 3$, deciding if a $g$-bounded TS $A$ is $\tau$-solvable is NP-complete if $\tau=\{\nop,\set,\res\}\cup\omega$ and $\emptyset\not=\omega\subseteq \{\used,\free\}$.
\end{theorem}
\begin{proof}
The TS $A^\tau_\varphi$ has the following gadgets $H_0, H_1$ and $H_2$ (in this order): 
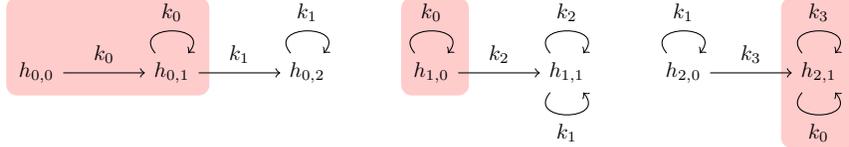
\begin{center}
\begin{tikzpicture}[new set = import nodes]
\begin{scope}[nodes={set=import nodes}]%
	\foreach \i in {0,...,2} {\coordinate (\i) at (\i*1.8cm,0);}
	\foreach \i in {0} {\fill[red!20, rounded corners] (\i) +(-0.4,-0.3) rectangle +(2.3,1);}
	\foreach \i in {0,...,2} {\node (t\i) at (\i) {\nscale{$h_{0,\i}$}};}
	
\path (t1) edge [->, out=130,in=50,looseness=3] node[above] {\nscale{$k_0$} } (t1);	
\path (t2) edge [->, out=130,in=50,looseness=3] node[above] {\nscale{$k_1$} } (t2);
\graph {
	(import nodes);
			t0 ->["\escale{$k_0$}"]t1;
			t1 ->["\escale{$k_1$}"]t2;
			};
\end{scope}
\begin{scope}[xshift=5.25cm,nodes={set=import nodes}]%
	\foreach \i in {0,...,1} {\coordinate (\i) at (\i*1.8cm,0);}
	\foreach \i in {0} {\fill[red!20, rounded corners] (\i) +(-0.4,-0.3) rectangle +(0.5,1);}
	\foreach \i in {0,...,1} {\node (t\i) at (\i) {\nscale{$h_{1,\i}$}};}
	
\path (t0) edge [->, out=130,in=50,looseness=3] node[above] {\nscale{$k_0$} } (t0);	
\path (t1) edge [->, out=130,in=50,looseness=3] node[above] {\nscale{$k_2$} } (t1);
\path (t1) edge [->, out=-130,in=-50,looseness=3] node[below] {\nscale{$k_1$} } (t1);
\graph {
	(import nodes);
			t0 ->["\escale{$k_2$}"]t1;
			};
\end{scope}
\begin{scope}[xshift=8.6cm,nodes={set=import nodes}]%
	\foreach \i in {0,...,1} {\coordinate (\i) at (\i*1.8cm,0);}
	\foreach \i in {1} {\fill[red!20, rounded corners] (\i) +(-0.5,-1) rectangle +(0.5,1);}
	\foreach \i in {0,...,1} {\node (t\i) at (\i) {\nscale{$h_{2,\i}$}};}
	
\path (t0) edge [->, out=130,in=50,looseness=3] node[above] {\nscale{$k_1$} } (t0);	
\path (t1) edge [->, out=130,in=50,looseness=3] node[above] {\nscale{$k_3$} } (t1);
\path (t1) edge [->, out=-130,in=-50,looseness=3] node[below] {\nscale{$k_0$} } (t1);
\graph {
	(import nodes);
			t0 ->["\escale{$k_3$}"]t1;
			};
\end{scope}
\end{tikzpicture}
\end{center}
The gadget $H_0$ provides the atom $\alpha=(k_0,h_{0,2})$.
By symmetry, a TS $A$ is $\{\nop,\set,\res,\used\}$-solvable if and only if it is $\{\nop,\set,\res,\free\}$-solvable or $\{\nop,\set,\res,\free, \used\}$-solvable.
Thus, in the following we assume $sig(k_0)=\used$ and $sup(h_{0,2})=0$ if $(sup, sig)$ $\tau$-solves $\alpha$.
As a result, by $sig(k_0)=\used$, implying $sup(h_{0,1})=1$, and $sup(h_{0,2})=0$ we have $sig(k_1)=\res$.
Especially, if $\edge{k_0}s$ then $sup(s)=1$ and if $\edge{k_1}s$ then $sup(s)=0$.
Thus, $sup(h_{1,0})=sup(h_{2,1})=1$ and $sup(h_{1,1})=sup(h_{2,0})=0$ which implies $sig(k_2)=\res$ and $sig(k_3)=\set$.

The construction uses $k_2$ and $k_3$ to produce some neutral events.
More exactly, the TS $A^\tau_\varphi$ implements for all $j\in \{0,\dots, 16m-1\}$  the following gadget $F_{j}$ that uses $k_2$ and $k_3$ to ensure that the events $z_j$ are neutral:
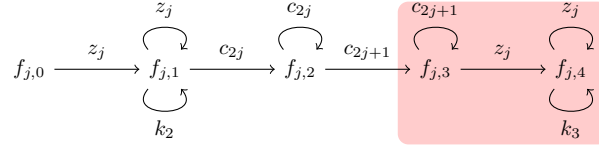
\begin{center}
\begin{tikzpicture}[new set = import nodes]
\begin{scope}[nodes={set=import nodes}]%
	\foreach \i in {0,...,4} {\coordinate (\i) at (\i*1.8cm,0);}
	\foreach \i in {3} {\fill[red!20, rounded corners] (\i) +(-0.5,-1) rectangle +(2.3,0.9);}
	\foreach \i in {0,...,4} {\node (t\i) at (\i) {\nscale{$f_{j,\i}$}};}
	
	\path (t1) edge [->, out=130,in=50,looseness=3] node[above] {\nscale{$z_j$} } (t1);
	\path (t1) edge [->, out=-130,in=-50,looseness=3] node[below] {\nscale{$k_2$} } (t1);
	\path (t2) edge [->, out=130,in=50,looseness=3] node[above] {\nscale{$c_{2j}$} } (t2);
	\path (t3) edge [->, out=130,in=50,looseness=3] node[above] {\nscale{$c_{2j+1}$} } (t3);
	\path (t4) edge [->, out=130,in=50,looseness=3] node[above] {\nscale{$z_j$} } (t4);
	\path (t4) edge [->, out=-130,in=-50,looseness=3] node[below] {\nscale{$k_3$} } (t4);
\graph {
	(import nodes);
			t0 ->["\escale{$z_j$}"]t1->["\escale{$c_{2j}$}"]t2->["\escale{$c_{2j+1}$}"]t3->["\escale{$z_j$}"]t4;
			};
\end{scope}
\end{tikzpicture}
\end{center}
By $sig(k_2)=\res$ and $sig(k_3)=\set$ we have $sup(f_{j,1})=0$ and $sup(f_{j,4})=1$.
This implies $\edge{sig(z_j)}0$ and $\edge{sig(z_j)}1$ and thus $sig(z_j)=\nop$.	

\noindent
Finally, for every $i\in \{0, \dots, m-1\}$ and clause $\zeta_i=\{X_{i_0}, X_{i_1}, X_{i_2}\}$, the TS $A^\tau_\varphi$ has the following four gadgets $T_{i,0}, T_{i,1}T_{i,2}$ and $T_{i,3}$ (in this order):
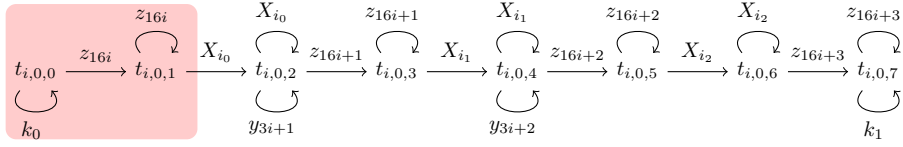
\begin{center}
\begin{tikzpicture}[new set = import nodes]
\begin{scope}[nodes={set=import nodes}]
	\foreach \i in {0,...,7} {\coordinate (\i) at (\i*1.6cm,0);}
	\foreach \i in {0} {\fill[red!20, rounded corners] (\i) +(-0.4,-0.9) rectangle +(2.15,0.9);}
	\foreach \i in {0,...,7} {\node (t\i) at (\i) {\nscale{$t_{i,0,\i}$}};}
	
\path (t0) edge [->, out=-130,in=-50,looseness=3] node[below, align =left] {\nscale{$k_0$} } (t0);	
\path (t1) edge [->, out=130,in=50,looseness=3] node[above, align =left] {\nscale{$z_{16i}$} } (t1);

\path (t2) edge [->, out=130,in=50,looseness=3] node[above, align =left] {\nscale{$X_{i_0}$} } (t2);
\path (t2) edge [->, out=-130,in=-50,looseness=3] node[below, align =left] {\nscale{$y_{3i+1}$} } (t2);

\path (t3) edge [->, out=130,in=50,looseness=3] node[above, align =left] {\nscale{$z_{16i+1}$} } (t3);

\path (t4) edge [->, out=130,in=50,looseness=3] node[above, align =left] {\nscale{$X_{i_1}$} } (t4);
\path (t4) edge [->, out=-130,in=-50,looseness=3] node[below, align =left] {\nscale{$y_{3i+2}$} } (t4);

\path (t5) edge [->, out=130,in=50,looseness=3] node[above, align =left] {\nscale{$z_{16i+2}$} } (t5);
\path (t6) edge [->, out=130,in=50,looseness=3] node[above, align =left] {\nscale{$X_{i_2}$} } (t6);

\path (t7) edge [->, out=130,in=50,looseness=3] node[above, align =left] {\nscale{$z_{16i+3} $} } (t7);
\path (t7) edge [->, out=-130,in=-50,looseness=3] node[below, align =left] {\nscale{$k_1$} } (t7);
\graph {
	(import nodes);
			t0 ->["\escale{$z_{16i}$}"]t1;
			t1 ->["\escale{$X_{i_0}$}"]t2;
			t2 ->["\escale{$z_{16i+1}$}"]t3;
			t3 ->["\escale{$X_{i_1}$}"]t4;
			t4 ->["\escale{$z_{16i+2}$}"]t5;
			t5 ->["\escale{$X_{i_2}$}"]t6;
			t6 ->["\escale{$z_{16i+3}$}"]t7;
			};
\end{scope}
\end{tikzpicture}
\end{center}
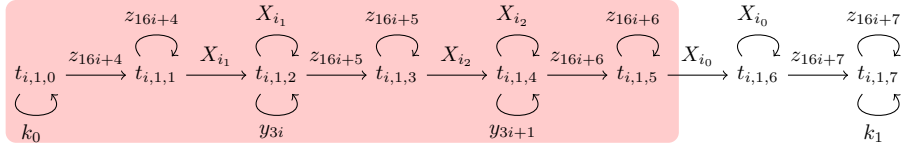
\begin{center}
\begin{tikzpicture}[new set = import nodes]
\begin{scope}[yshift=-2.25cm,nodes={set=import nodes}]
	\foreach \i in {0,...,7} {\coordinate (\i) at (\i*1.6cm,0);}
	\foreach \i in {0} {\fill[red!20, rounded corners] (\i) +(-0.4,-0.9) rectangle +(8.55,1);}

	\foreach \i in {0,...,7} {\node (t\i) at (\i) {\nscale{$t_{i,1,\i}$}};}
	
\path (t0) edge [->, out=-130,in=-50,looseness=3] node[below, align =left] {\nscale{$k_0$} } (t0);	
\path (t1) edge [->, out=130,in=50,looseness=3] node[above, align =left] {\nscale{$z_{16i+4}$} } (t1);

\path (t2) edge [->, out=130,in=50,looseness=3] node[above, align =left] {\nscale{$X_{i_1}$} } (t2);
\path (t2) edge [->, out=-130,in=-50,looseness=3] node[below, align =left] {\nscale{$y_{3i}$} } (t2);

\path (t3) edge [->, out=130,in=50,looseness=3] node[above, align =left] {\nscale{$z_{16i+5}$} } (t3);

\path (t4) edge [->, out=130,in=50,looseness=3] node[above, align =left] {\nscale{$X_{i_2}$} } (t4);
\path (t4) edge [->, out=-130,in=-50,looseness=3] node[below, align =left] {\nscale{$y_{3i+1}$} } (t4);

\path (t5) edge [->, out=130,in=50,looseness=3] node[above, align =left] {\nscale{$z_{16i+6}$} } (t5);
\path (t6) edge [->, out=130,in=50,looseness=3] node[above, align =left] {\nscale{$X_{i_0}$} } (t6);
\path (t7) edge [->, out=130,in=50,looseness=3] node[above, align =left] {\nscale{$z_{16i+7} $} } (t7);
\path (t7) edge [->, out=-130,in=-50,looseness=3] node[below, align =left] {\nscale{$k_1$} } (t7);
\graph {
	(import nodes);
			t0 ->["\escale{$z_{16i+4}$}"]t1;
			t1 ->["\escale{$X_{i_1}$}"]t2;
			t2 ->["\escale{$z_{16i+5}$}"]t3;
			t3 ->["\escale{$X_{i_2}$}"]t4;
			t4 ->["\escale{$z_{16i+6}$}"]t5;
			t5 ->["\escale{$X_{i_0}$}"]t6;
			t6 ->["\escale{$z_{16i+7}$}"]t7;
			};
\end{scope}
\end{tikzpicture}
\end{center}
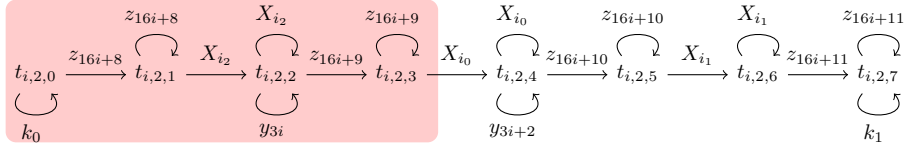
\begin{center}
\begin{tikzpicture}[new set = import nodes]
\begin{scope}[yshift=-5cm,nodes={set=import nodes}]
	\foreach \i in {0,...,7} {\coordinate (\i) at (\i*1.6cm,0);}
	\foreach \i in {0} {\fill[red!20, rounded corners] (\i) +(-0.4,-0.9) rectangle +(5.35,1);}
	\foreach \i in {0,...,7} {\node (t\i) at (\i) {\nscale{$t_{i,2,\i}$}};}
	
\path (t0) edge [->, out=-130,in=-50,looseness=3] node[below, align =left] {\nscale{$k_0$} } (t0);	
\path (t1) edge [->, out=130,in=50,looseness=3] node[above, align =left] {\nscale{$z_{16i+8}$} } (t1);

\path (t2) edge [->, out=130,in=50,looseness=3] node[above, align =left] {\nscale{$X_{i_2}$} } (t2);
\path (t2) edge [->, out=-130,in=-50,looseness=3] node[below, align =left] {\nscale{$y_{3i}$} } (t2);

\path (t3) edge [->, out=130,in=50,looseness=3] node[above, align =left] {\nscale{$z_{16i+9}$} } (t3);

\path (t4) edge [->, out=130,in=50,looseness=3] node[above, align =left] {\nscale{$X_{i_0}$} } (t4);
\path (t4) edge [->, out=-130,in=-50,looseness=3] node[below, align =left] {\nscale{$y_{3i+2}$} } (t4);

\path (t5) edge [->, out=130,in=50,looseness=3] node[above, align =left] {\nscale{$z_{16i+10}$} } (t5);
\path (t6) edge [->, out=130,in=50,looseness=3] node[above, align =left] {\nscale{$X_{i_1}$} } (t6);

\path (t7) edge [->, out=130,in=50,looseness=3] node[above, align =left] {\nscale{$z_{16i+11} $} } (t7);
\path (t7) edge [->, out=-130,in=-50,looseness=3] node[below, align =left] {\nscale{$k_1$} } (t7);
\graph {
	(import nodes);
			t0 ->["\escale{$z_{16i+8}$}"]t1;
			t1 ->["\escale{$X_{i_2}$}"]t2;
			t2 ->["\escale{$z_{16i+9}$}"]t3;
			t3 ->["\escale{$X_{i_0}$}"]t4;
			t4 ->["\escale{$z_{16i+10}$}"]t5;
			t5 ->["\escale{$X_{i_1}$}"]t6;
			t6 ->["\escale{$z_{16i+11}$}"]t7;
			};
\end{scope}
\end{tikzpicture}
\end{center}
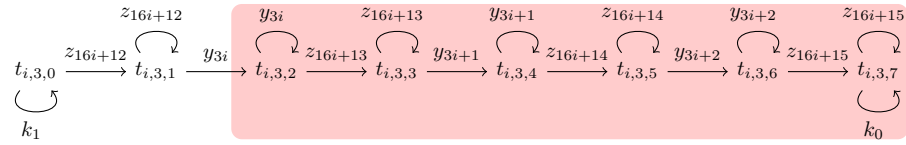
\begin{center}
\begin{tikzpicture}[new set = import nodes]
\begin{scope}[nodes={set=import nodes}]%
	\foreach \i in {0,...,7} {\coordinate (\i) at (\i*1.6cm,0);}
	\foreach \i in {2} {\fill[red!20, rounded corners] (\i) +(-0.6,-0.9) rectangle +(8.4,0.9);}
	\foreach \i in {0,...,7} {\node (t\i) at (\i) {\nscale{$t_{i,3,\i}$}};}
	
\path (t0) edge [->, out=-130,in=-50,looseness=3] node[below, align =left] {\nscale{$k_1$} } (t0);	
\path (t1) edge [->, out=130,in=50,looseness=3] node[above, align =left] {\nscale{$z_{16i+12} $} } (t1);

\path (t2) edge [->, out=130,in=50,looseness=3] node[above, align =left] {\nscale{$y_{3i}$} } (t2);

\path (t3) edge [->, out=130,in=50,looseness=3] node[above, align =left] {\nscale{$z_{16i+13} $} } (t3);

\path (t4) edge [->, out=130,in=50,looseness=3] node[above, align =left] {\nscale{$y_{3i+1}$} } (t4);

\path (t5) edge [->, out=130,in=50,looseness=3] node[above, align =left] {\nscale{$z_{16i+14} $} } (t5);
\path (t6) edge [->, out=130,in=50,looseness=3] node[above, align =left] {\nscale{$y_{3i+2}$} } (t6);

\path (t7) edge [->, out=130,in=50,looseness=3] node[above, align =left] {\nscale{$z_{16i+15} $} } (t7);
\path (t7) edge [->, out=-130,in=-50,looseness=3] node[below, align =left] {\nscale{$k_0$} } (t7);
\graph {
	(import nodes);
			t0 ->["\escale{$z_{16i+12} $}"]t1;
			t1 ->["\escale{$y_{3i}$}"]t2;
			t2 ->["\escale{$z_{16i+13} $}"]t3;
			t3 ->["\escale{$y_{3i+1}$}"]t4;
			t4 ->["\escale{$z_{16i+14} $}"]t5;
			t5 ->["\escale{$y_{3i+2}$}"]t6;
			t6 ->["\escale{$z_{16i+15} $}"]t7;
			};
\end{scope}
\end{tikzpicture}
\end{center}
The functionality of $T_{i,0},\dots, T_{i,4}$ is like this:
By $sig(k_0)=\used$ and $sig(k_1)=\res$ we have that $sup(t_{i,0,0})=sup(t_{i,1,0})=sup(t_{i,2,0})=sup(t_{i,3,7})=1$ and $sup(t_{i,0,7})=sup(t_{i,1,7})=sup(t_{i,2,7})=sup(t_{i,3,0})=0$.
Thus, by the neutrality of $z_{16i}, \dots, z_{16i+11}$, there has to be at least one variable event with a $\res$-signature.
Moreover, by $sup(t_{i,3,0})=0$ and $sup(t_{i,3,3})=1$ and the neutrality of $z_{16i+12}, \dots, z_{16i+15}$ there is an event of $y_{3i}, y_{3i+1}, y_{3i+2}$ with a $\set$-signature.
We argue that there is exactly one variable event with a \res-signature:
If $sig(X_{i_0})=\res$ then $sup(t_{i,0,2})=sup(t_{i,2,4})=0$ which implies $sig(y_{3i+1})\not=\set$ and $sig(y_{3i+2})\not=\set$ and thus $sig(y_{3i})=\set$.
By $sig(y_{3i})=\set$ we have $sup(t_{i,1,2})=sup(t_{i,2,2})=1$ which implies $sig(X_{i_1})\not=\res$ and $sig(X_{i_2})\not=\res$. 
Similarly, one argues that $sig(X_{i_1})=\res$ implies $sig(X_{i_0})\not=\res$ and $sig(X_{i_1})\not=\res$ and that $sig(X_{i_2})=\res$ requires $sig(X_{i_0})\not=\res$ and $sig(X_{i_1})\not=\res$.
Thus, the set $M=\{X\in V(\varphi) \mid sig(X)=\res\}$ is a one-in-three model of $\varphi$. 
This shows that the $\tau$-solvability of $A^\tau_\varphi$ implies the one-in-three satisfiability of $\varphi$. 

To join the gadgets finally build $A^\tau_\varphi$, we use the states $\bot=\{\bot_0,\dots, \bot_{20m+2}\}$ and the events $\oplus=\{\oplus_{0},\dots, \oplus_{20m+2}\}$ and $\ominus=\{\ominus_{1},\dots, \ominus_{20m+2}\}$.
The states of $\bot$ are connected by $\bot_j\edge{\ominus_{j+1}}\bot_{j+1}$ and $\bot_{j+1}\edge{\ominus_{j+1}}\bot_{j+1}$ for $j\in \{0,\dots, 20m+2\}$.
Let $x=16m+3$ and $y=19m+3$.
For all $i\in \{0,1,2\}$, for all $j\in \{3,\dots, 13m+2\}$, for all $\ell\in \{0,\dots, m-1\}$ and for all $n\in \{0,\dots, 2\}$ we add the following edges that connect the gadgets $H_0,H_1,H_2$ and $F_0,\dots, F_{16m-1}$ and $T_{0,0},T_{0,1},T_{0,2},\dots, T_{m-1,0},T_{m-1,1}T_{m-1,2}$ and $T_{0,3},\dots,T_{m-1,3}$ of $A^\tau_\varphi$:
\begin{center}
\begin{tikzpicture}[new set = import nodes]
\begin{scope}[nodes={set=import nodes}]%
	\node (0) at (0,0) {\nscale{$\bot_j$}};
	\node (1) at (1.5,0) {\nscale{$h_{j,0}$}};
		
	\path (1) edge [->, out=130,in=50,looseness=3] node[above] {\nscale{$\oplus_j$} } (1);

\graph {
	(import nodes);
			0 ->["\escale{$\oplus_j$}"]1;
			};
\end{scope}
\begin{scope}[xshift=3cm,nodes={set=import nodes}]%
	\node (0) at (0,0) {\nscale{$\bot_\ell$}};
	\node (1) at (1.5,0) {\nscale{$f_{\ell,0}$}};
		
	\path (1) edge [->, out=130,in=50,looseness=3] node[above] {\nscale{$\oplus_\ell$} } (1);

\graph {
	(import nodes);
			0 ->["\escale{$\oplus_\ell$}"]1;
			};
\end{scope}
\begin{scope}[xshift=6cm,nodes={set=import nodes}]%
	\node (0) at (0,0) {\nscale{$\bot_{x+3i+n}$}};
	\node (1) at (2,0) {\nscale{$t_{i,n,0}$}};
		
	\path (1) edge [->, out=130,in=50,looseness=3] node[above] {\nscale{$\oplus_{x+3i+n}$} } (1);

\graph {
	(import nodes);
			0 ->["\escale{$\oplus_{x+3i+n}$}"]1;
			};
\end{scope}
\begin{scope}[xshift=9.25cm,nodes={set=import nodes}]%
	\node (0) at (0,0) {\nscale{$\bot_{y+\ell}$}};
	\node (1) at (1.8,0) {\nscale{$t_{\ell,3,0}$}};
		
	\path (1) edge [->, out=130,in=50,looseness=3] node[above] {\nscale{$\oplus_{y+\ell}$} } (1);

\graph {
	(import nodes);
			0 ->["\escale{$\oplus_{y+\ell}$}"]1;
			};
\end{scope}
\end{tikzpicture}
\end{center}
If $M$ is a one-in-three model of $\varphi$ then $\alpha$ is $\tau$-solvable by a $\tau$-region $(sup, sig)$.
The red colored area above indicates already a positive support of some states.
In particular, if $s\in \{h_{0,0}, h_{1,0}, h_{2,1}\}$ or $\{f_{j,0} \mid j\in \{0,\dots, 16m-1\}\}$ then $sup(s)=1$.
The support values of the states of $T_{i,0},\dots, T_{i,3}$, where $i\in \{0,\dots, m-1\}$, are defined in accordance to which event of $X_{i_0}, X_{i_1}, X_{i_2}$ belongs to $M$.
The red colored area above sketches $X_{i_0}\in M$.
Moreover, we define $sup(s)=0$ for all $s\in \bot$.
Let $e\in E(A^\tau_\varphi)\setminus \oplus$.
We define $sig(e)=\used$ if $e=k_0$ and $sig(e)=\res$ if $e\in \{k_1\}\cup M$.
For all $i\in \{0,\dots, m-1\}$ and clauses $\{X_{i_0}, X_{i_1}, X_{i_2}\}$ and all $j\in \{0,1,2\}$ we set $sig(e)=\set$ if $e=y_{3i+j}$ and $X_{i_j}\in M$.
Otherwise, we define $sig(e)=\nop$.
Finally, for all events $e\in\oplus$ and edges $s\edge{e}s'$ of $A$ we define $sig(e)=\set$ if $sup(s')=1$ and, otherwise, $sig(e)=\nop$.
The resulting $\tau$-region $(sup, sig)$ of $A^\tau_\varphi$ solves $\alpha$.
\end{proof}

\begin{theorem}\label{the:nop_inp_set_2bounded}
For any fixed $g\geq 2$, deciding if a $g$-bounded TS $A$ is $\tau$-solvable is NP-complete if one of the following conditions is true:
\begin{enumerate}
\item
$\tau=\{\nop,\inp,\set\}$ or $\tau=\{\nop,\inp,\set,\used\}$ or $\tau=\{\nop,\inp,\res,\set\}\cup\omega$ and $\omega\subseteq\{\out,\used,\free\}$,
\item
$\tau=\{\nop,\out,\res\}$ or $\tau=\{\nop,\out,\res,\free\}$ or $\tau=\{\nop,\out,\res,\set\}\cup\omega$ and $\omega\subseteq \{\inp,\used,\free\}$.
\end{enumerate}
\end{theorem}
\begin{proof}
We present a reduction for Condition~1 which, by Lemma~\ref{lem:isomorphy}, proves the claim for the other types, too.
The TS $A^\tau_\varphi$ has the following gadget $H$:
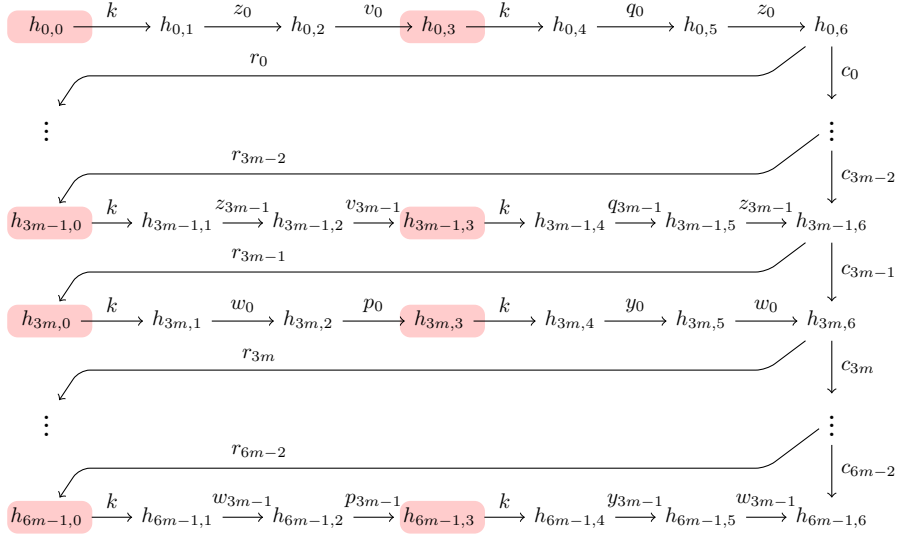
\begin{center}
\noindent
\begin{tikzpicture}[new set = import nodes, scale =0.87]
\begin{scope}[nodes={set=import nodes}]
\foreach \i in {0,...,6} {\coordinate (h0\i) at (\i*2cm,0);}
\foreach \i in {0,...,6} {\coordinate (h3m_1\i) at (\i*2cm,-3);}
\foreach \i in {0,...,6} {\coordinate (h3m\i) at (\i*2cm,-4.5);}
\foreach \i in {0,...,6} {\coordinate (h6m_1\i) at (\i*2cm,-7.5);}
\foreach \i in {h00,h03,h3m_10,h3m_13,h3m0,h3m3,h6m_10,h6m_13} {\fill[red!20, rounded corners] (\i) +(-0.6,-0.25) rectangle +(0.7,0.25);}
\foreach \i in {0,...,6} {\node (h0\i) at (h0\i) {\nscale{$h_{0,\i}$}};}
\foreach \i in {0,...,6} {\node (h3m_1\i) at (h3m_1\i) {\nscale{$h_{3m-1,\i}$}};}
\foreach \i in {0,...,6} {\node (h3m\i) at (h3m\i) {\nscale{$h_{3m,\i}$}};}
\foreach \i in {0,...,6} {\node (h6m_1\i) at (h6m_1\i) {\nscale{$h_{6m-1,\i}$}};}

\coordinate (d0) at(0,-1.5);
\coordinate (d1) at(12cm,-1.5cm);
\coordinate (r0_0) at(11,-0.75);
\coordinate (r0_1) at(6cm,-0.75cm);
\coordinate (r0_2) at(0.5cm,-0.75cm);
\node (d0) at (d0) {\nscale{$\vdots$}};
\node (d1) at (d1) {\nscale{$\vdots$}};
\draw[->, , rounded corners] (h06) -- (r0_0)--(r0_1)--node[above] { \escale{$r_{0}$}}(r0_2)-- (d0);
\coordinate (r3m_2_0) at(11cm,-2.25cm);
\coordinate (r3m_2_1) at(6cm,-2.25cm);
\coordinate (r3m_2_2) at(0.5cm,-2.25cm);
\draw[->, , rounded corners] (d1) -- (r3m_2_0)--(r3m_2_1)--node[above] { \escale{$r_{3m-2}$}}(r3m_2_2)-- (h3m_10);
\coordinate (r3m_1_0) at(11cm,-3.75cm);
\coordinate (r3m_1_1) at(6cm,-3.75cm);
\coordinate (r3m_1_2) at(0.5cm,-3.75cm);
\draw[->, , rounded corners] (h3m_16) -- (r3m_1_0)--(r3m_1_1)--node[above] { \escale{$r_{3m-1}$}}(r3m_1_2)-- (h3m0);
\coordinate (d2) at(0,-6);
\coordinate (d3) at(12cm,-6cm);
\coordinate (r3m_0) at(11,-5.25);
\coordinate (r3m_1) at(6cm,-5.25cm);
\coordinate (r3m_2) at(0.5cm,-5.25cm);
\node (d2) at (d2) {\nscale{$\vdots$}};
\node (d3) at (d3) {\nscale{$\vdots$}};
\draw[->, , rounded corners] (h3m6) -- (r3m_0)--(r3m_1)--node[above] { \escale{$r_{3m}$}}(r3m_2)-- (d2);
\coordinate (r6m_2_0) at(11cm,-6.75cm);
\coordinate (r6m_2_1) at(6cm,-6.75cm);
\coordinate (r6m_2_2) at(0.5cm,-6.75cm);
\draw[->, , rounded corners] (d3) -- (r6m_2_0)--(r6m_2_1)--node[above] { \escale{$r_{6m-2}$}}(r6m_2_2)-- (h6m_10);

\graph { 
%
(h00) ->["\escale{$k$}"] (h01) ->["\escale{$z_0$}"] (h02) ->["\escale{$v_0$}"] (h03) ->["\escale{$k$}"] (h04) ->["\escale{$q_0$}"] (h05)->["\escale{$z_0$}"](h06);
%
(h3m_10) ->["\escale{$k$}"] (h3m_11) ->["\escale{$z_{3m-1}$}"] (h3m_12) ->["\escale{$v_{3m-1}$}"] (h3m_13) ->["\escale{$k$}"] (h3m_14) ->["\escale{$q_{3m-1}$}"] (h3m_15)->["\escale{$z_{3m-1}$}"](h3m_16);
%
(h3m0) ->["\escale{$k$}"] (h3m1) ->["\escale{$w_0$}"] (h3m2) ->["\escale{$p_{0}$}"] (h3m3) ->["\escale{$k$}"] (h3m4) ->["\escale{$y_{0}$}"] (h3m5)->["\escale{$w_0$}"](h3m6);
%
(h6m_10) ->["\escale{$k$}"] (h6m_11) ->["\escale{$w_{3m-1}$}"] (h6m_12) ->["\escale{$p_{3m-1}$}"] (h6m_13) ->["\escale{$k$}"] (h6m_14) ->["\escale{$y_{3m-1}$}"] (h6m_15)->["\escale{$w_{3m-1}$}"](h6m_16);

(h06)->["\escale{$c_0$}"] (d1); 
(d1)->["\escale{$c_{3m-2}$}"] (h3m_16); 
(h3m_16)->["\escale{$c_{3m-1}$}"] (h3m6); 
(h3m6)->["\escale{$c_{3m}$}"] (d3); 
(d3)->["\escale{$c_{6m-2}$}"] (h6m_16); 
};
\end{scope}
\end{tikzpicture}
\end{center}
The intention of the gadget $H$ is to provide the atom $\alpha=(k, h_{0,6})$ and the events of $Z=\{z_0,\dots, z_{3m-1}\}$, $V=\{v_0,\dots, v_{3m-1}\}$ and $W=\{w_0,\dots, w_{3m-1}\}$.

Moreover, the TS $A^\tau_\varphi$ has the following two gadgets $F_0$ and $F_1$ and for all $j\in \{0,\dots, 6m-2\}$ the following gadget $G_j$ (in this order):
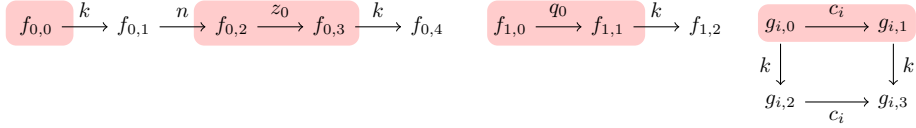
\begin{center}
\begin{tikzpicture}[new set = import nodes]
\begin{scope}[nodes={set=import nodes}]%
	\foreach \i in {0,...,4} {\coordinate (\i) at (\i*1.3cm,0);}
	\foreach \i in {0} {\fill[red!20, rounded corners] (\i) +(-0.4,-0.25) rectangle +(0.5,0.35);}
	\foreach \i in {3} {\fill[red!20, rounded corners] (\i) +(-1.8,-0.25) rectangle +(0.4,0.35);}
	\foreach \i in {0,...,4} {\node (n\i) at (\i) {\nscale{$f_{0,\i}$}};}
		
\end{scope}
\graph {
	(import nodes);
			n0 ->["\escale{$k$}"]n1;
			n1 ->["\escale{$n$}"]n2;
			n2 ->["\escale{$z_0$}"]n3;
			n3 ->["\escale{$k$}"]n4;
			};
\begin{scope}[nodes={set=import nodes}, xshift=6.3cm]%
		\foreach \i in {0,...,2} {\coordinate (\i) at (\i*1.3cm,0);}
		\foreach \i in {1} {\fill[red!20, rounded corners] (\i) +(-1.6,-0.25) rectangle +(0.5,0.35);}
		\foreach \i in {0,...,2} {\node (m\i) at (\i) {\nscale{$f_{1,\i}$}};}
\end{scope}
\graph {
	(import nodes);
			m0 ->["\escale{$q_0$}"]m1;
			m1 ->["\escale{$k$}"]m2;
		};
\begin{scope}[xshift=9.9cm,nodes={set=import nodes}]%

		\coordinate (1) at (1.5,0) ;	
		\foreach \i in {1} {\fill[red!20, rounded corners] (\i) +(-1.8,-0.2) rectangle +(0.3,0.3);}
		\node (0) at (0,0) {\nscale{$g_{i,0}$}};
		\node (1) at (1.5,0) {\nscale{$g_{i,1}$}};	
		\node (2) at (0,-1) {\nscale{$g_{i,2}$}};
		\node (3) at (1.5,-1) {\nscale{$g_{i,3}$}};
\graph {
	(import nodes);
			0 ->["\escale{$c_i$}"]1;
			0 ->[swap, "\escale{$k$}"]2;
			2 ->[swap,"\escale{$c_i$}"]3;
			1 ->["\escale{$k$}"]3;
			
			};
\end{scope}
\end{tikzpicture}
\end{center}
%
%
Finally, the TS $A^\tau_\varphi$ has for every clause $\zeta_i=\{X_{i_0}, X_{i_1}, X_{i_2}\}$, $i\in \{0,\dots, m-1\}$, the following gadgets $T_{i,0}, T_{i,1}$ and $T_{i,2}$ (in this order): 
%
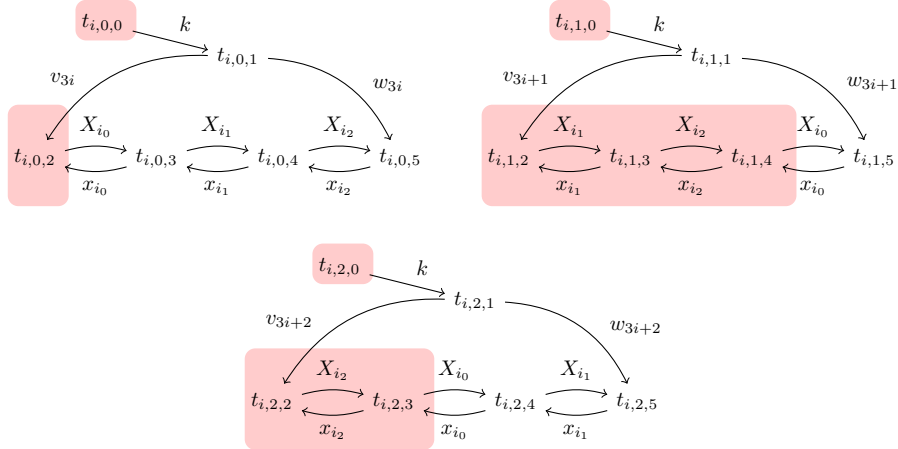
\begin{center}
\begin{tikzpicture}[new set = import nodes,scale=0.9]
\begin{scope}[nodes={set=import nodes}]
		\coordinate (0) at (1,2);
		\foreach \i in {0} {\fill[red!20, rounded corners] (\i) +(-0.4,-0.25) rectangle +(0.5,0.35);}
		\coordinate(1) at (3,1.5);
		\foreach \i in {2,...,5} { \pgfmathparse{\i-2} \coordinate (\i) at (\pgfmathresult*1.8cm,0) ;}
		\foreach \i in {2} {\fill[red!20, rounded corners] (\i) +(-0.4,-0.7) rectangle +(0.5,0.8);}
		\foreach \i in {0,...,5} { \node (\i) at (\i) {\nscale{$t_{i,0,\i}$}};}
\end{scope}
\graph {
	(import nodes);
			0 ->["\escale{$k$}"]1;
			1 ->[swap, bend right =30, pos=0.7,"\escale{$v_{3i}$}"]2;
			1 ->[bend left=30,pos=0.7,"\escale{$w_{3i}$}"]5;
			2->[bend left=15 , "\escale{$X_{i_0}$}"]3;
			3 ->[bend left=15 , "\escale{$X_{i_1}$}"]4;
			4 ->[ bend left=15 , "\escale{$X_{i_2}$}"]5;
			3 ->[bend left=15 , "\escale{$x_{i_0}$}"]2;
			4 ->[bend left=15 , "\escale{$x_{i_1}$}"]3;
			5 ->[bend left=15 ,  "\escale{$x_{i_2}$}"]4;
	
	};
\begin{scope}[nodes={set=import nodes}, xshift=7cm]
		\coordinate (0) at (1,2);
		\foreach \i in {0} {\fill[red!20, rounded corners] (\i) +(-0.4,-0.25) rectangle +(0.5,0.35);}
		\coordinate(1) at (3,1.5);
		\foreach \i in {2,...,5} { \pgfmathparse{\i-2} \coordinate (\i) at (\pgfmathresult*1.8cm,0) ;}
		\foreach \i in {2} {\fill[red!20, rounded corners] (\i) +(-0.4,-0.7) rectangle +(4.25,0.8);}
		\foreach \i in {0,...,5} { \node (\i) at (\i) {\nscale{$t_{i,1,\i}$}};}
\end{scope}
\graph {
	(import nodes);
			0 ->["\escale{$k$}"]1;
			1 ->[swap, bend right =30, pos=0.7,"\escale{$v_{3i+1}$}"]2;
			1 ->[bend left=30,pos=0.7,"\escale{$w_{3i+1}$}"]5;
			2->[bend left=15 , "\escale{$X_{i_1}$}"]3;
			3 ->[bend left=15 , "\escale{$X_{i_2}$}"]4;
			4 ->[ bend left=15 , "\escale{$X_{i_0}$}"]5;
			3 ->[bend left=15 , "\escale{$x_{i_1}$}"]2;
			4 ->[bend left=15 , "\escale{$x_{i_2}$}"]3;
			5 ->[bend left=15 ,  "\escale{$x_{i_0}$}"]4;
	
	};
\end{tikzpicture}
\end{center}
\begin{center}
\begin{tikzpicture}[new set = import nodes,scale=0.9]
\begin{scope}[nodes={set=import nodes},xshift=3.5cm, yshift=-3.25cm]
		\coordinate (0) at (1,2);
		\foreach \i in {0} {\fill[red!20, rounded corners] (\i) +(-0.4,-0.25) rectangle +(0.5,0.35);}
		\coordinate(1) at (3,1.5);
		\foreach \i in {2,...,5} { \pgfmathparse{\i-2} \coordinate (\i) at (\pgfmathresult*1.8cm,0) ;}
		\foreach \i in {2} {\fill[red!20, rounded corners] (\i) +(-0.4,-0.7) rectangle +(2.4,0.8);}
		\foreach \i in {0,...,5} { \node (\i) at (\i) {\nscale{$t_{i,2,\i}$}};}
\end{scope}
\graph {
	(import nodes);
			0 ->["\escale{$k$}"]1;
			1 ->[swap, bend right =30, pos=0.7,"\escale{$v_{3i+2}$}"]2;
			1 ->[bend left=30,pos=0.7,"\escale{$w_{3i+2}$}"]5;
			2->[bend left=15 , "\escale{$X_{i_2}$}"]3;
			3 ->[bend left=15 , "\escale{$X_{i_0}$}"]4;
			4 ->[ bend left=15 , "\escale{$X_{i_1}$}"]5;

			3 ->[bend left=15 , "\escale{$x_{i_2}$}"]2;
			4 ->[bend left=15 , "\escale{$x_{i_0}$}"]3;
			5 ->[bend left=15 ,  "\escale{$x_{i_1}$}"]4;
	
	};
\end{tikzpicture}
\end{center} 
In the following, we argue that $H,F_0,F_1$ and $G_0,\dots, G_{m-2}$ collaborate like this:
If $(sup, sig)$ is a $\tau$-region solving $\alpha$ then either $sig(k)=\inp$, $V\subseteq sig^{-1}(\enter)$ and $W\subseteq  sig^{-1}(\keepze)$ or $sig(k)=\out$ and $V\subseteq sig^{-1}(\exit)$ and $W\subseteq sig^{-1}(\keepo)$.
Moreover, we prove that this implies by the functionality of $T_{0,0},\dots, T_{m-1,2}$ that $M=\{X\in V(\varphi) \mid sig(X)\not=\nop\}$ is a one-in-three model of $\varphi$.

Let $(sup, sig)$ be a $\tau$-region that solves $\alpha$.
Recall that there are basically four interactions possibly useful for $sig(k)$, namely $\inp,\out,\used,\free$.
Since the interactions $\res,\set,\swap,\nop$ are defined on both $0$ and $1$, they do not fit to solve $\alpha$.
If $sig(k)=\used$ then $sup(s)=sup(s')=1$ for every transition $s\edge{k}s'$.
Hence, we have $sup(f_{0,3})=sup(f_{1,1})=sup(h_{0,4})=1$.
By definition of $\inp, \res$ we have that $\edge{e}s$ and $sig(e)\in \{\inp,\res\}$ implies $sup(s)=0$.
Consequently, by $\edge{z_0}f_{0,3}$ and $\edge{q_0}f_{1,1}$ we get $sig(z_0), sig(q_0)\in \keepo$ and thus $sup(h_{0,4})=sup(h_{0,5})=sup(h_{0,6})=1$ which contradicts $\neg sup(h_{0,6})\edge{sig(k)}$.
Hence, $sig(k)\not=\used$.
Similarly, $sig(k) =\free$ implies $sup(h_{0,6})=0$, which is a contradiction.
Thus, we have that $sig(k)=\inp$ and $sup(h_{0,6})=0$ or $sig(k)=\out$ and  $sup(h_{0,6})=1$. 

As a next step, we show that $sig(k)=\inp$ and $sup(h_{0,6})=0$ implies $sig(v_0)\in \enter$ and $sig(z_0)\in \keepze$.
By $sig(k)=\inp$ and $\edge{k}h_{0,1}$ and $h_{0,3}\edge{k}$ we get $sup(h_{0,1})=0$ and $sup(h_{0,3})=1$.
Moreover, by $\edge{z_0}h_{0,6}$ and $sup(h_{0,6})=0$ we obtain $sig(z_0)\in\keepze$, which by $sup(h_{0,1})=0$ implies  $sup(h_{0,2})=0$.
Finally, $sup(h_{0,2})=0$ and $sup(h_{0,3})=1$ imply $sig(v_0)\in \enter$.
Notice that this reasoning purely bases on $sig(k)=\inp$ and $sup(h_{0,6})=0$.
Moreover, $A^\tau_\varphi$ uses for every $j\in \{0,\dots, 6m-2\}$ the TS $G_j$ to transfer ensure $sup(h_{0,6})=sup(h_{1,6})=\dots=sup(h_{6m-1,6})$.
This transfers $z_0\in \keepze$ and $v_0\in \enter $ to $V\subseteq \enter$ and $W\subseteq \keepze$.
In particular, by $sig(k)=\inp$ we have $sup(g_{i,0})=sup(g_{i,0})=1$ and $sup(g_{i,2})=sup(g_{i,3})=0$, that is, $sig(c_i)=\nop$.
Hence, if $sig(k)=\inp$ and $sup(h_{0,6})=0$ then $sup(h_{i,6})=0$ for all $i\in \{0,\dots, 6m-1\}$.
Perfectly similar to the discussion for $z_0$ and $v_0$ we obtain that $V\subseteq sig^{-1}(\enter)$ and $W\subseteq  sig^{-1}(\keepze)$, respectively.
Similarly, $sig(k)=\out$ and $sup(h_{0,6})=1$ imply $V\subseteq sig^{-1}(\exit)$ and $W\subseteq sig^{-1}(\keepo)$.

We now argue that $T_{i,0},\dots, T_{m-1,2}$ ensure that $M=\{X\in V(\varphi)\mid sig(X)\not=\nop\}$ is a one-in-three model of $\varphi$.
Let $i\in \{0,\dots, m-1\}$ and $sig(k)=\inp$ and $sup(h_{0,6})=0$ implying $V\subseteq sig^{-1}(\enter)$ and $W\subseteq  sig^{-1}(\keepze)$.
By $sig(k)=\inp$ and $V\subseteq sig^{-1}(\enter)$ and $W\subseteq  sig^{-1}(\keepze)$ we have that $sup(t_{i,0,2})=sup(t_{i,1,2})=sup(t_{i,2,2})=1$ and $sup(t_{i,0,5})=sup(t_{i,1,5})=sup(t_{i,2,5})=0$.
As a result, every event $e\in \{X_{i_0}, X_{i_1}, X_{i_2}\}$ has a $0$-sink, which implies $sig(e)\in \{\nop,\inp,\res\}$, and every event $e\in \{x_{i_0}, x_{i_1}, x_{i_2}\}$ has a $1$-sink, which implies $sig(e)\in \{\nop,\out,\set\}$.
By $sup(t_{i,0,2})=1$ and $sup(t_{i,0,5})=0$ there is a $X\in \{X_{i_0}, X_{i_1}, X_{i_2}\}$ such that $sig(X)\in \{\inp,\res\}$.
We argue that $sig(Y)=\nop$ for $Y\in \{X_{i_0}, X_{i_1}, X_{i_2}\}\setminus \{X\}$.
If $sig(X_{i_0})\in \{\inp,\res\}$ then $sup(t_{i,0,3})=0$ which implies $sig(x_{i_0})\in \{\out,\set\}$ and, therefor, $sup(t_{i,1,4})=1$.
Since $sig(X_{i_1}), sig(X_{i_2})\not\in \{\out, \set\}$ and $sig(x_{i_1}), sig(x_{i_2})\not\in \{\inp, \res\}$ we obtain that $sup(t_{i,0,3})=sup(t_{i,0,4})=0$ and $sup(t_{i,1,3})=sup(t_{i,1,4})=1$, respectively.
Thus, for all $e\in \{X_{i_1}, X_{i_2}\}$ there are edges $\edge{e}s$ and $\edge{e}s'$ such that $sup(s)=0$ and $sup(s')=1$, which implies $sig(e)=\nop$.
Similarly, if $sig(X_{i_1})=\inp$ then $sig(X_{i_0})=sig(X_{i_1})=\nop$ and if $sig(X_{i_1})=\inp$ then $sig(X_{i_0})=sig(X_{i_2})=\nop$.
Consequently, there is for every clause $\zeta_i$ exactly one variable event with a signature different from \nop, which makes $M=\{X\in V(\varphi)\mid sig(X)\not=\nop\}$ a one-in-three model of $\varphi$.
By symmetry, if $sig(k)=\out$ and $sup(h_{0,6})=1$ then $M$ is a one-in-three model of $\varphi$, too.

To join the gadgets and finally build $A^\tau_\varphi$, we use the states $\bot=\{\bot_0,\dots, \bot_{9m+1}\}$ and the events $\oplus=\{\oplus_{0},\dots, \oplus_{9m+1}\}$ and $\ominus=\{\ominus_{1},\dots, \ominus_{9m+1}\}$.
The states of $\bot$ are connected by $\bot_j\edge{\ominus_{j+1}}\bot_{j+1}$ for $j\in \{0,\dots, 9m+1\}$.
Let $x=6m+2$.
For all $i\in \{0,\dots, 6m-2\}$, for all $j\in \{0,\dots, m-1\}$ and for all $\ell\in \{0,\dots, 2\}$ we add the following edges that connect the gadgets $H_0, F_0,F_1, G_0,\dots, G_{6m-2}$ and $T_{0,0},T_{0,1},T_{0,2}$ up to $T_{m-1,0},T_{m-1,1},T_{m-1,2}$ to $A^\tau_\varphi$:
\begin{center}
\begin{tikzpicture}[new set = import nodes]
\begin{scope}[nodes={set=import nodes}]%
	\node (0) at (0,0) {\nscale{$\bot_0$}};
	\node (1) at (1.3,0) {\nscale{$h_{0,0}$}};
\graph {
	(import nodes);
			0 ->["\escale{$\oplus_0$}"]1;
			};
\end{scope}
\begin{scope}[xshift=2.2cm,nodes={set=import nodes}]%
	\node (0) at (0,0) {\nscale{$\bot_1$}};
	\node (1) at (1.3,0) {\nscale{$f_{0,0}$}};
\graph {
	(import nodes);
			0 ->["\escale{$\oplus_1$}"]1;
			};
\end{scope}
\begin{scope}[xshift=4.4cm,nodes={set=import nodes}]%
	\node (0) at (0,0) {\nscale{$\bot_2$}};
	\node (1) at (1.3,0) {\nscale{$f_{1,0}$}};
\graph {
	(import nodes);
			0 ->["\escale{$\oplus_2$}"]1;
			};
\end{scope}
\begin{scope}[xshift=6.6cm,nodes={set=import nodes}]%
	\node (0) at (0,0) {\nscale{$\bot_{i+3}$}};
	\node (1) at (1.5,0) {\nscale{$g_{i+3,0}$}};
\graph {
	(import nodes);
			0 ->["\escale{$\oplus_{i+3}$}"]1;
			};
\end{scope}
\begin{scope}[xshift=9.5cm,nodes={set=import nodes}]%
	\node (0) at (0,0) {\nscale{$\bot_{x+3\ell+n}$}};
	\node (1) at (2.1,0) {\nscale{$t_{\ell,n,0}$}};
\graph {
	(import nodes);
			0 ->["\escale{$\oplus_{x+3\ell+n}$}"]1;
			};
\end{scope}

\end{tikzpicture}
\end{center}

If $M$ is a one-in-three model of $\varphi$ then there is a $\tau$-region $(sup, sig)$ of $A^\tau_\varphi$ that solves $\alpha$.
The red colored area of the figures introducing the gadgets indicates already a positive support of some states.
In particular, if $s\in \{h_{j,0}, h_{j,3} \mid j\in \{0,\dots, 6m-1\}\}$ or $s\in \{f_{0,0}, f_{0,2}, f_{0,3}, f_{1,0},f_{1,1} \}$ $s\in \{g_{j,0}, g_{j,1} \mid j\in \{0,\dots, 6m-2\}\}$ then $sup(s)=1$.
The support values of the states of $T_{i,0},\dots, T_{i,2}$, where $i\in \{0,\dots, m-1\}$, are defined in accordance to which of the events $X_{i_0}, X_{i_1}, X_{i_2}$ belongs to $M$.
The red colored area above sketches the situation where $X_{i_0}\in M$.
Moreover, we define $sup(s)=0$ for all $s\in \bot$.
Let $e\in E(A^\tau_\varphi)\setminus \oplus$.
We define $sig(e)=\inp$ if $e\in \{k\}\cup M$ and $sig(e)=\set $ if $e\in $.
For all $i\in \{0,\dots, m-1\}$ and clauses $\{X_{i_0}, X_{i_1}, X_{i_2}\}$ and all $j\in \{0,1,2\}$ we set $sig(e)=\set$ if $e=n$ or $e\in \{v_j,p_j\mid j\in \{0,\dots, 3m-1\}\}$ or $e=x_{i_j}$ and $X_{i_j}\in M$.
Otherwise, we define $sig(e)=\nop$.
Finally, for all events $e\in\oplus$ and edges $s\edge{e}s'$ of $A$ we define $sig(e)=\set$ if $sup(s')=1$ and, otherwise, $sig(e)=\nop$.
\end{proof}

\textbf{Joining of $1$-bounded gadgets.} 
In the following, we consider types $\tau$ where $\tau$-synthesis from $1$-bounded inputs is NP-complete.
All gadgets $A_0,\dots, A_n$ of the reductions are directed paths $A_i=s^i_0\edge{e_1}\dots, \edge{e_n}s^i_n$ on pairwise distinct states $s^i_0,\dots, s^i_n$.
The joining is for all types the concatenation 
\begin{center}
\begin{tikzpicture}[new set = import nodes]
\begin{scope}[nodes={set=import nodes}]%
		\node (init) at (-1,0) {$A^\tau_\varphi=$};
		\foreach \i in {0,...,4} { \coordinate (\i) at (\i*1.4cm,0) ;}
		\node (0) at (0) {$A_0$};
		\node (1) at (1) {\nscale{$\bot_1$}};
		\node (2) at (2) {$A_2$};
		\node (3) at (3) {\nscale{$\bot_2$}};
		\node (4) at (4) {};
		\node (5) at (6.2,0) {$\dots$};
		\node (6) at (6.8,0) {};
		\node (7) at (8.2,0) {\nscale{$\bot_n$}};
		\node (8) at (9.6,0) {$A_n$};
\graph {
	(import nodes);
			0 ->["\escale{$\ominus_1$}"]1->["\escale{$\oplus_1$}"]2 ->["\escale{$\ominus_2$}"]3 ->["\escale{$\oplus_2$}"]4;
			6  ->["\escale{$\ominus_n$}"]7->["\escale{$\oplus_n$}"]8;
};
\end{scope}
\end{tikzpicture}
\end{center}
with fresh states $\bot_1,\dots,\bot_n$ and events $\ominus_1,\dots \ominus_n,\oplus_1,\dots \oplus_n$.
\begin{theorem}\label{the:nop_inp_out_set+used_free}
For any fixed $g\geq 1$, deciding if a $g$-bounded TS $A$ is $\tau$-solvable is NP-complete if $\tau=\{\nop, \inp,\out,\set\}\cup\omega$ or $\tau=\{\nop, \inp,\out,\res\} \cup \omega$ and $\omega \subseteq \{\used, \free\}$.
\end{theorem}
\begin{proof}
Our construction proves the claim for $\tau=\{\nop, \inp,\set,\out\}\cup\omega$ with $\omega \subseteq \{\used, \free\}$.
By Lemma~\ref{lem:isomorphy}, this proves the claim also for the other types.

The TS $A^\tau_\varphi$ has the following gadgets $H_0, H_1,H_2$ and $H_3$ (in this order): 
\begin{center}
\begin{tikzpicture}[new set = import nodes]
\begin{scope}[nodes={set=import nodes}]%
		\foreach \i in {0,...,8} { \coordinate (\i) at (\i*1.4cm,0) ;}
		\foreach \i in {1} {\fill[red!20, rounded corners] (\i) +(-0.5,-0.25) rectangle +(1.9,0.35);}
		\foreach \i in {4} {\fill[red!20, rounded corners] (\i) +(-0.5,-0.25) rectangle +(3.3,0.35);}
		\foreach \i in {8} {\fill[red!20, rounded corners] (\i) +(-0.5,-0.25) rectangle +(0.4,0.35);}

		\foreach \i in {0,...,8} { \node (\i) at (\i) {\nscale{$h_{0,\i}$}};}
\graph {
	(import nodes);
			0 ->["\escale{$k_0$}"]1->["\escale{$z_0$}"]2 ->["\escale{$o$}"]3 ->["\escale{$k_1$}"]4  ->["\escale{$z_1$}"]5->["\escale{$z_0$}"]6->["\escale{$o$}"]7->["\escale{$k_0$}"]8;
};
\end{scope}
\begin{scope}[yshift=-1.2cm,nodes={set=import nodes}]%
		\foreach \i in {0,...,2} { \coordinate (\i) at (\i*1.5cm,0) ;}
		\foreach \i in {2} {\fill[red!20, rounded corners] (\i) +(-0.5,-0.25) rectangle +(0.4,0.35);}
		\foreach \i in {0,...,2} { \node (\i) at (\i) {\nscale{$h_{1,\i}$}};}
\graph {
	(import nodes);
			0 ->["\escale{$z_0$}"]1->["\escale{$k_0$}"]2;
			};
\end{scope}
\begin{scope}[xshift=4cm, yshift=-1.2cm,nodes={set=import nodes}]%
		\foreach \i in {0,...,2} { \coordinate (\i) at (\i*1.5cm,0) ;}
		\foreach \i in {2} {\fill[red!20, rounded corners] (\i) +(-0.5,-0.25) rectangle +(0.4,0.35);}
		\foreach \i in {0,...,2} { \node (\i) at (\i) {\nscale{$h_{2,\i}$}};}
\graph {
	(import nodes);
			0 ->["\escale{$z_1$}"]1->["\escale{$k_0$}"]2;
			};
\end{scope}
\begin{scope}[xshift=8cm,yshift=-1.2cm,nodes={set=import nodes}]%
		\foreach \i in {0,...,2} { \coordinate (\i) at (\i*1.5cm,0) ;}
		\foreach \i in {1} {\fill[red!20, rounded corners] (\i) +(-0.5,-0.25) rectangle +(1.9,0.4);}
		\foreach \i in {0,...,2} { \node (\i) at (\i) {\nscale{$h_{3,\i}$}};}
\graph {
	(import nodes);
			0 ->["\escale{$k_0$}"]1->["\escale{$k_1$}"]2;
			};
\end{scope}
\end{tikzpicture}
\end{center}
If $\used\in\tau$ then $A^\tau_\varphi$ has the following gadget $H_4$: 
\begin{center}
\begin{tikzpicture}[new set = import nodes]
\begin{scope}[nodes={set=import nodes}]%
		\foreach \i in {0,...,3} { \coordinate (\i) at (\i*1.4cm,0) ;}
		\foreach \i in {0} {\fill[red!20, rounded corners] (\i) +(-0.5,-0.25) rectangle +(4.6,0.4);}
		\foreach \i in {0,...,3} { \node (\i) at (\i) {\nscale{$h_{4,\i}$}};}
\graph {
	(import nodes);
			0 ->["\escale{$k_1$}"]1->["\escale{$z_0$}"]2 ->["\escale{$k_1$}"]3;
			};
\end{scope}
\end{tikzpicture}
\end{center}
For all $i\in \{0,\dots, m-1\}$, the TS $A^\tau_\varphi$ has for the clause $\zeta_i=\{X_{i_0}, X_{i_1}, X_{i_2}\}$ and the variable $X_i\in V(\varphi)$ the following gadgets $T_i$ and $B_i$, respectively:
\begin{center}
\begin{tikzpicture}[new set = import nodes]
\begin{scope}[nodes={set=import nodes}]%
		\foreach \i in {0,...,5} { \coordinate (\i) at (\i*1.3cm,0) ;}
		\foreach \i in {0} {\fill[red!20, rounded corners] (\i) +(-0.5,-0.25) rectangle +(1.7,0.4);}
		\foreach \i in {5} {\fill[red!20, rounded corners] (\i) +(-0.45,-0.25) rectangle +(0.35,0.4);}
		\foreach \i in {0,...,5} { \node (\i) at (\i) {\nscale{$t_{i,\i}$}};}
\graph {
	(import nodes);
			0 ->["\escale{$k_1$}"]1->["\escale{$X_{i_0}$}"]2 ->["\escale{$X_{i_1}$}"]3->["\escale{$X_{i_2}$}"]4->["\escale{$k_0$}"]5; 
			};
\end{scope}
\begin{scope}[xshift=7.5cm,nodes={set=import nodes}]%
		\foreach \i in {0,...,2} { \coordinate (\i) at (\i*1.3cm,0) ;}
		\foreach \i in {2} {\fill[red!20, rounded corners] (\i) +(-0.45,-0.25) rectangle +(0.35,0.4);}
		\foreach \i in {0,...,2} { \node (\i) at (\i) {\nscale{$b_{i,\i}$}};}
\graph {
	(import nodes);
			0 ->["\escale{$X_i$}"]1->["\escale{$k_0$}"]2;  
			};
\end{scope}
\end{tikzpicture}
\end{center}
The gadget $H_0$ provides the atom $\alpha=(k_0,h_{0,6})$.
Moreover, the gadgets $H_0,\dots, H_4$ ensure that if $(sup, sig)$ is a $\tau$-region solving $\alpha$ then $sig(k_0)=\out$ and $sig(k_1)\in \{\out, \set\}$.
In particular, $H_4$ prevents the solvability of $(k_0,h_{0,6})$ by $\used$.
As a result, such a region implies $sup(t_{i,1})=1$, $sup(t_{i,4})=0$ and $sup(b_{i,1})=0$ for all $i\in \{0,\dots, m-1\}$.
On the one hand, by $sup(b_{i,1})=0$ for all $i\in \{0,\dots, m-1\}$ we have $sig(X)\not\in \{\out, \set\}$ for all $X\in V(\varphi)$.
On the other hand, by  $sup(t_{i,1})=1$ and $sup(t_{i,4})=0$ there is an event $X\in \{X_{i_0}, X_{i_1},X_{i_2}\}$ such that $sig(X)=\inp$.
Since no variable event has an incoming signature we obtain immediately $sig(Y)\not=\inp$ for $Y\in \{X_{i_0}, X_{i_1},X_{i_2}\}\setminus\{X\}$.
Thus, $M=\{X\in V(\varphi) \mid sig(X)=\inp\}$ is a one-in-three model of $\varphi$.

We argue that $H_0,\dots, H_4$ behave as announced.
Let $(sup, sig)$ be a region that solves $(k_0,h_{0,6})$.
If $sig(k_0)=\inp$ then $sup(h_{0,6})=0$ and $sig(h_{0,7})=1$, implying $sig(o)\in \{\out,\set\}$ and $sup(h_{0,3})=1$.
Thus, there is an event $e\in \{k_1, z_0,z_1\}$ with $sig(e)=\inp$.
By $sig(k_0)=\inp$ we have $sup(h_{1,1})=sup(h_{2,1})=1$ and $sup(h_{3,1})=0$ implying $sig(e)\not=\inp$ for all $e\in \{k_1, z_0,z_1\}$, a contradiction.

If $sig(k_0)=\free$ then $sup(h_{0,6})=1$ and $sup(h_{0,1})=sup(h_{0,7})=sup(h_{1,1})=0$ which implies $sig(o)=\inp$ and $sup(h_{0,2})=1$.
By $sup(h_{0,1})=0$ and $sup(h_{0,2})=1$ we have $sig(z_0)\in \{\out, \set\}$ which by $sup(h_{1,1})=0$ is a contradiction.

If $sig(k_0)=\used$ then $sup(h_{0,6})=0$ and $sup(h_{0,1})=sup(h_{0,7})=sup(h_{1,1})=sup(h_{2,1})=1$.
This implies $sig(o)\in \{\out,\set\}$ and $sup(h_{0,3})=1$.
Thus, by $sup(h_{0,6})=0$ there is an event $e\in \{k_1,z_0,z_1\}$ with $sig(e)=\inp$.
By $sup(h_{1,1})=sup(h_{2,1})=1$, we have $e\not\in \{z_0,z_1\}$.
If $sig(k_1)=\inp$ then $sup(h_{4,1})=0$ and $sup(h_{4,2})=1$, implying $sig(z_0)\in \{\out,\set\}$ and $sup(h_{0,6})=1$.
This is a contradiction.
Altogether, this proves $sig(k_0)\not\in \{\inp,\used,\free\}$.

Consequently, we obtain $sig(k_0)=\out$ and $sup(h_{0,6})=1$ which implies $sig(o)=\inp$ and $sup(h_{0,3})=0$.
By $sup(h_{0,6})=1$, this implies that there is an event  $e\in \{k_1,z_0,z_1\}$ with $sig(e)\in \{\out, \set\}$.
Again by $sig(k_0)=\out$, we have $sup(h_{1,1})=sup(h_{2,1})=0$, which implies $e=k_1$. 
The signatures $sig(k_0)=\out$ and $sig(k_1)\in\{\out,\set\}$ and the construction of $T_0,\dots, T_{m-1}$ and $B_0,\dots, B_{m-1}$ ensure that $M=\{X\in V(\varphi) \mid sig(X)=\inp\}$ is a one-in-three model of $\varphi$:
By $sig(k_0)=\out$ and $sig(k_1)\in\{\out,\set\}$ we have $sup(t_{i,1})=1$ and $sup(t_{i,4})=sup(b_{i,1})=0$ for all $i\in \{0,\dots, m-1\}$.
By $sup(t_{i,1})=1$ and $sup(t_{i,4})=0$, there is an event $X\in \zeta_i$ such that $sig(X)=\inp$.
Moreover, by $sup(b_{i,1})=0$, we get $sig(X_i)\not\in \enter$ for all $i\in \{0,\dots, m-1\}$.
Thus, $X$ is unambiguous and thus $M$ a searched model.

In reverse, if $M$ is a one-in-three model of $\varphi$ then there is a $\tau$-region $(sup, sig)$ that solves $\alpha$.
The red colored area above sketches states with a positive support.
Which states of $T_{i}$, besides of $t_{i,0}, t_{i,1}$ and $t_{i,5}$, get a positive support depends for all $i\in \{0,\dots, m-1\}$ on which of $X_{i_0}, X_{i_1},X_{i_2}$ belongs to $M$.
The red colored area above sketches the case $X_{i_0}\in M$.
Moreover, we define $sup(s)=1$ if $s=b_{i,0}$ and $X_i\in M$ or if $s\in \bot$.
For $sig$ we define $sig(k_1)=\set$ and for all $e\in E(A^\tau_\varphi)\setminus \{k_1\}$ and $s\edge{e}s'\in A^\tau_\varphi$ we define $sig(e)=\out$ if and $sup(s')>sup(s)$ and $sig(e)=\inp$ if $sup(s)>sup(s')$ and else $sig(e)=\nop$.
\end{proof}

\begin{theorem}\label{the:nop_inp_set_free+used}
For any $g\geq 1$, deciding if a $g$-bounded TS $A$ is $\tau$-solvable is NP-complete if $\tau=\{\nop, \inp, \set, \free\}$ or $\tau=\{\nop, \inp, \set, \used, \free\}$ or $\tau=\{\nop, \out, \res, \used\}$ or $\tau=\{\nop, \out, \res, \used, \free\}$.
\end{theorem}
\begin{proof}
Our reduction proves the claim for $\tau=\{\nop, \inp, \set, \free\}$ and $\tau=\{\nop, \inp, \set, \used, \free\}$ and thus 
by Lemma~\ref{lem:isomorphy}, for the other types, too.

The TS $A^\tau_\varphi$ has the following gadgets $H_0$ and $H_1$ providing the atom $(k_0,h_{0,3})$:
\begin{center}
\begin{tikzpicture}[new set = import nodes]
\begin{scope}[nodes={set=import nodes}]%
		\foreach \i in {0,...,6} { \coordinate (\i) at (\i*1.4cm,0) ;}
		\foreach \i in {2} {\fill[red!20, rounded corners] (\i) +(-0.35,-0.25) rectangle +(3.65,0.4);}
		\foreach \i in {0,...,6} { \node (\i) at (\i*1.5cm,0) {\nscale{$h_{0,\i}$}};}
\graph {
	(import nodes);
			0 ->["\escale{$k_0$}"]1->["\escale{$k_1$}"]2 ->["\escale{$z_0$}"]3 ->["\escale{$k_1$}"]4  ->["\escale{$z_1$}"]5->["\escale{$k_0$}"]6;

			};
\end{scope}
\begin{scope}[yshift=-1.2cm,nodes={set=import nodes}]%
		\foreach \i in {0,...,3} { \coordinate (\i) at (\i*1.4cm,0) ;}
		\foreach \i in {0,...,3} { \node (\i) at (\i*1.5cm,0) {\nscale{$h_{1,\i}$}};}
\graph {
	(import nodes);
			0 ->["\escale{$k_0$}"]1->["\escale{$z_0$}"]2 ->["\escale{$k_0$}"]3;
			};
\end{scope}
\end{tikzpicture}
\end{center}
For all $i\in \{0,\dots, m-1\}$, the $A^\tau_\varphi$ for the clause $\zeta_i=\{X_{i_0}, X_{i_1}, X_{i_2}\}$ and the variable $X_i\in V(\varphi)$ the gadgets $T_i$ and gadget $B_i$, respectively, as previously defined for Theorem~\ref{the:nop_inp_set_free+used}.
%
%
%
%
The gadgets $H_0$ and $H_1$ ensure that a $\tau$-region $(sup, sig)$ solving $(k_0,h_{0,3})$ satisfies $sig(k_0)=\free$ and $sig(k_1)=\set$.
This implies $sup(t_{i,1})=1$ and $sup(t_{i,4})=sup(b_{i,2})=0$ for all $i\in \{0,\dots, m-1\}$.
By $sup(t_{i,1})=1$ and $sup(t_{i,4})=0$, there is an event $X\in \zeta_i$ such that $sig(X)=\inp$ and, by $sup(b_{i,2})=0$ for all $i\in \{0,\dots, m-1\}$, we have $sig(X)\not=\set$ for all $X\in V(\varphi)$.
Thus, the event $X\in \zeta_i$ is unique and $M=\{X\in V(\varphi) \mid sig(X)=\inp\}$ is a one-in-three model.

We briefly argue that $H_0$ and $H_1$ perform as announced:
Let $(sup, sig)$ be a $\tau$-region that solves $\alpha$.
If $sig(k_0)=\inp$ then $sup(h_{1,1})=0$ and $sup(h_{1,2})=1$ which implies $sig(z_0)=\set$ and thus $sup(h_{0,3})=1$, a contradiction.
Hence, $sig(k_{0})\not=\inp$.
If $sig(k_0)=\used$ then $sup(h_{0,1})=sup(h_{1,2})=1$ and $sup(h_{0,3})=0$.
Consequently, $sig(z_0)=\inp$ or $sig(k_1)=\inp$ but this contradicts $sup(h_{1,2})=1$ and $sup(h_{0,3})=0$.
Thus, $sig(k_0)\not=\used$.
Thus, we have $sig(k_0)=\free$ and $sup(h_{0,3})=1$, which implies that one of $k_1,z_0$ has a $\set$-signature.
By $sig(k_0)=\free$, we get $sup(h_{1,3})=0$ and thus $sig(k_1)=\set$. 

If $M$ is a one-in-three model of $\varphi$ then we can define an $\alpha$ solving region similar to the one of Theorem~\ref{the:nop_inp_out_set+used_free}, where we replace $sig(k_0)=\inp$ by $sig(k_0)=\free$.
\end{proof}
\begin{theorem}\label{the:nop_inp_res_swap+used_free}
For any fixed $g\geq 1$, deciding if a $g$-bounded TS $A$ is $\tau$-solvable is NP-complete if  $\tau=\{\nop, \inp,\res,\swap\}\cup\omega$ and $\omega\subseteq\{\used,\free\}$.
\end{theorem}
\begin{proof}
The TS $A^\tau_\varphi$ has the following gadgets $H_0, H_1, H_2$ and $H_3$: 
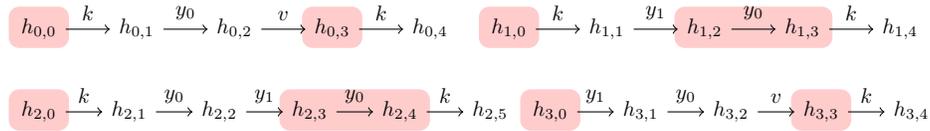
\begin{center}
\begin{tikzpicture}[new set = import nodes]
\begin{scope}[nodes={set=import nodes}]

		\foreach \i in {0,...,4} { \coordinate (\i) at (\i*1.3cm,0) ;}
		\foreach \i in {0,3} {\fill[red!20, rounded corners] (\i) +(-0.4,-0.25) rectangle +(0.4,0.3);}
		\foreach \i in {0,...,4} { \node (\i) at (\i) {\nscale{$h_{0,\i}$}};}
\graph {
	(import nodes);
			0 ->["\escale{$k$}"]1->["\escale{$y_0$}"]2->["\escale{$v$}"]3->["\escale{$k$}"]4;  
			};
\end{scope}
\begin{scope}[xshift=6.25cm,nodes={set=import nodes}]%

		\foreach \i in {0,...,4} { \coordinate (\i) at (\i*1.3cm,0) ;}
		\foreach \i in {0} {\fill[red!20, rounded corners] (\i) +(-0.4,-0.25) rectangle +(0.4,0.3);}
		\foreach \i in {2} {\fill[red!20, rounded corners] (\i) +(-0.4,-0.25) rectangle +(1.7,0.3);}
		\foreach \i in {0,...,4} { \node (\i) at (\i) {\nscale{$h_{1,\i}$}};}
\graph {
	(import nodes);
			0 ->["\escale{$k$}"]1->["\escale{$y_1$}"]2->["\escale{$y_0$}"]3->["\escale{$k$}"]4;  
			};
\end{scope}
\begin{scope}[yshift=-1.1cm,nodes={set=import nodes}]%

		\foreach \i in {0,...,5} { \coordinate (\i) at (\i*1.2cm,0) ;}
		\foreach \i in {0} {\fill[red!20, rounded corners] (\i) +(-0.4,-0.25) rectangle +(0.4,0.3);}
		\foreach \i in {3} {\fill[red!20, rounded corners] (\i) +(-0.4,-0.25) rectangle +(1.6,0.3);}
		\foreach \i in {0,...,5} { \node (\i) at (\i) {\nscale{$h_{2,\i}$}};}
\graph {
	(import nodes);
			0 ->["\escale{$k$}"]1->["\escale{$y_0$}"]2->["\escale{$y_1$}"]3->["\escale{$y_0$}"]4->["\escale{$k$}"]5;  
			};
\end{scope}
\begin{scope}[xshift=6.8cm, yshift=-1.1cm,nodes={set=import nodes}]%

		\foreach \i in {0,...,4} { \coordinate (\i) at (\i*1.2cm,0) ;}
		\foreach \i in {0,3} {\fill[red!20, rounded corners] (\i) +(-0.4,-0.25) rectangle +(0.4,0.3);}
		\foreach \i in {0,...,4} { \node (\i) at (\i) {\nscale{$h_{3,\i}$}};}
\graph {
	(import nodes);
			0 ->["\escale{$y_1$}"]1->["\escale{$y_0$}"]2->["\escale{$v$}"]3->["\escale{$k$}"]4;  
			};
\end{scope}
\end{tikzpicture}
\end{center}
The gadgets $H_0,\dots, H_3$ provide the atom $\alpha=(k,h_{0,2})$ and ensure that a $\tau$-region $(sup, sig)$ solving $\alpha$ satisfies $sig(k)=\inp$ and $sup(h_{0,2})=0$.
%
The TS $A^\tau_\varphi$ has the following gadgets $F_0, F_1$ and for all $j\in \{0,\dots, 10\}$ the gadget $G_j$:
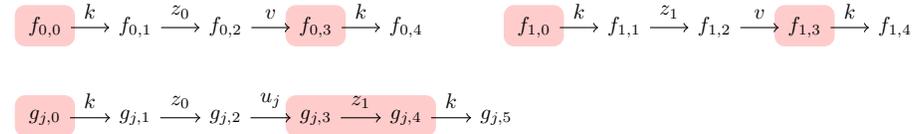
\begin{center}
\begin{tikzpicture}[new set = import nodes]
\begin{scope}[nodes={set=import nodes}]
		
		\foreach \i in {0,...,4} { \coordinate (\i) at (\i*1.2cm,0) ;}
		\foreach \i in {0,3} {\fill[red!20, rounded corners] (\i) +(-0.4,-0.25) rectangle +(0.4,0.3);}
		\foreach \i in {0,...,4} { \node (\i) at (\i) {\nscale{$f_{0,\i}$}};}
\graph {
	(import nodes);
			0 ->["\escale{$k$}"]1->["\escale{$z_0$}"]2->["\escale{$v$}"]3->["\escale{$k$}"]4;  
			};
\end{scope}
\begin{scope}[xshift=6.5cm,nodes={set=import nodes}]
		
		\foreach \i in {0,...,4} { \coordinate (\i) at (\i*1.2cm,0) ;}		
		\foreach \i in {0,3} {\fill[red!20, rounded corners] (\i) +(-0.4,-0.25) rectangle +(0.4,0.3);}
		\foreach \i in {0,...,4} { \node (\i) at (\i) {\nscale{$f_{1,\i}$}};}
\graph {
	(import nodes);
			0 ->["\escale{$k$}"]1->["\escale{$z_1$}"]2->["\escale{$v$}"]3->["\escale{$k$}"]4;  
			};
\end{scope}
\begin{scope}[yshift=-1.2cm,nodes={set=import nodes}]
		
		\foreach \i in {0,...,5} { \coordinate (\i) at (\i*1.2cm,0) ;}
		\foreach \i in {0} {\fill[red!20, rounded corners] (\i) +(-0.4,-0.25) rectangle +(0.4,0.3);}
		\foreach \i in {3} {\fill[red!20, rounded corners] (\i) +(-0.4,-0.25) rectangle +(1.6,0.3);}
		\foreach \i in {0,...,5} { \node (\i) at (\i) {\nscale{$g_{j,\i}$}};}
\graph {
	(import nodes);
			0 ->["\escale{$k$}"]1->["\escale{$z_0$}"]2->["\escale{$u_j$}"]3->["\escale{$z_1$}"]4->["\escale{$k$}"]5;  
			};
\end{scope}
\end{tikzpicture}
\end{center}
For all $j\in \{0,\dots, 10\}$, the gadgets $F_0,F_1, G_j$ ensure $sig(u_j)=\swap$ for any $\tau$-region $(sup, sig)$ solving $\alpha$.

For all $i\in \{0,\dots, m-1\}$, the TS $A^\tau_\varphi$ has for the clause $\zeta_i=\{X_{i_0}, X_{i_1}, X_{i_2}\}$ some gadgets $T_{i,0},\dots, T_{i,6}$ and $B_i$.
The purpose of these gadgets is to make the one-and-three satisfiability of $\varphi$ and the solvability of $\alpha$ the same. 
In particular, the TS $T_{i,0}$ is defined by:
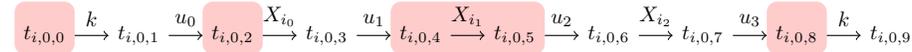
\begin{center}
\begin{tikzpicture}[new set = import nodes]
\begin{scope}[nodes={set=import nodes}]
		
		\foreach \i in {0,...,9} { \coordinate (\i) at (\i*1.25cm,0) ;}
		\foreach \i in {0,2,8} {\fill[red!20, rounded corners] (\i) +(-0.4,-0.25) rectangle +(0.4,0.45);}
		\foreach \i in {4} {\fill[red!20, rounded corners] (\i) +(-0.4,-0.25) rectangle +(1.65,0.45);}
		\foreach \i in {0,...,9} { \node (\i) at (\i) {\nscale{$t_{i,0,\i}$}};}
\graph {
	(import nodes);
			0 ->["\escale{$k$}"]1->["\escale{$u_{0}$}"]2->["\escale{$X_{i_0}$}"]3->["\escale{$u_{1}$}"]4->["\escale{$X_{i_1}$}"]5->["\escale{$u_{2}$}"]6->["\escale{$X_{i_2}$}"]7->["\escale{$u_{3}$}"]8->["\escale{$k$}"]9;
			};
\end{scope}
\end{tikzpicture}
\end{center}
The gadgets $T_{i,1}, T_{i,2}$ and $T_{i,3}$ are defined (in this order) as follows:
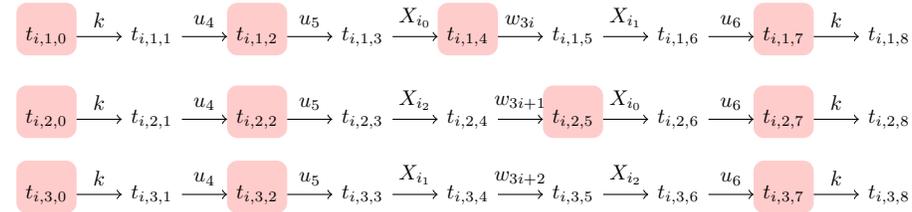
\begin{center}
\begin{tikzpicture}[new set = import nodes]
\begin{scope}[nodes={set=import nodes}]
		
		\foreach \i in {0,...,8} { \coordinate (\i) at (\i*1.4cm,0) ;}
		\foreach \i in {0,2,4,7} {\fill[red!20, rounded corners] (\i) +(-0.4,-0.25) rectangle +(0.4,0.45);}
		\foreach \i in {0,...,8} { \node (\i) at (\i) {\nscale{$t_{i,1,\i}$}};}
\graph {
	(import nodes);
			0 ->["\escale{$k$}"]1->["\escale{$u_{4}$}"]2->["\escale{$u_{5}$}"]3->["\escale{$X_{i_0}$}"]4->["\escale{$w_{3i}$}"]5->["\escale{$X_{i_1}$}"]6->["\escale{$u_{6}$}"]7->["\escale{$k$}"]8;
			};
\end{scope}
\begin{scope}[yshift=-1.1cm,nodes={set=import nodes}]
		
		\foreach \i in {0,...,8} { \coordinate (\i) at (\i*1.4cm,0) ;}
		\foreach \i in {0,2,5,7} {\fill[red!20, rounded corners] (\i) +(-0.4,-0.25) rectangle +(0.4,0.45);}
		\foreach \i in {0,...,8} { \node (\i) at (\i) {\nscale{$t_{i,2,\i}$}};}
\graph {
	(import nodes);
			0 ->["\escale{$k$}"]1->["\escale{$u_{4}$}"]2->["\escale{$u_{5}$}"]3->["\escale{$X_{i_2}$}"]4->["\escale{$w_{3i+1}$}"]5->["\escale{$X_{i_0}$}"]6->["\escale{$u_{6}$}"]7->["\escale{$k$}"]8;
			};
\end{scope}
\begin{scope}[yshift=-2.1cm,nodes={set=import nodes}]
		
		\foreach \i in {0,...,8} { \coordinate (\i) at (\i*1.4cm,0) ;}
		\foreach \i in {0,2,7} {\fill[red!20, rounded corners] (\i) +(-0.4,-0.25) rectangle +(0.4,0.45);}
		\foreach \i in {0,...,8} { \node (\i) at (\i) {\nscale{$t_{i,3,\i}$}};}
\graph {
	(import nodes);
			0 ->["\escale{$k$}"]1->["\escale{$u_{4}$}"]2->["\escale{$u_{5}$}"]3->["\escale{$X_{i_1}$}"]4->["\escale{$w_{3i+2}$}"]5->["\escale{$X_{i_2}$}"]6->["\escale{$u_{6}$}"]7->["\escale{$k$}"]8;
			};
\end{scope}
\end{tikzpicture}
\end{center}
Moreover, the gadgets $T_{i,4}, T_{i,5}$ and $T_{i,6}$ are defined like this: 
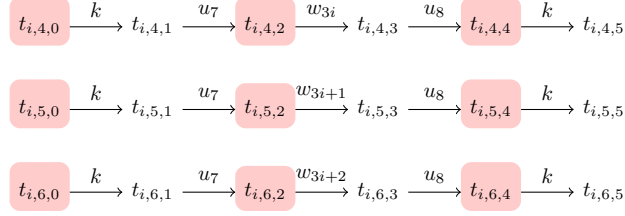
\begin{center}
\begin{tikzpicture}[new set = import nodes]
\begin{scope}[nodes={set=import nodes}]
		
		\foreach \i in {0,...,5} { \coordinate (\i) at (\i*1.5cm,0) ;}
		\foreach \i in {0,2,4} {\fill[red!20, rounded corners] (\i) +(-0.4,-0.25) rectangle +(0.4,0.4);}
		\foreach \i in {0,...,5} { \node (\i) at (\i) {\nscale{$t_{i,4,\i}$}};}
\graph {
	(import nodes);
			0 ->["\escale{$k$}"]1->["\escale{$u_{7}$}"]2->["\escale{$w_{3i}$}"]3->["\escale{$u_8$}"]4->["\escale{$k$}"]5;
			};
\end{scope}
\begin{scope}[yshift=-1.1cm, nodes={set=import nodes}]
		
		\foreach \i in {0,...,5} { \coordinate (\i) at (\i*1.5cm,0) ;}
		\foreach \i in {2} {\fill[red!20, rounded corners] (\i) +(-0.4,-0.25) rectangle +(0.4,0.35);}
		\foreach \i in {0,4} {\fill[red!20, rounded corners] (\i) +(-0.4,-0.25) rectangle +(0.4,0.4);}
		\foreach \i in {0,...,5} { \node (\i) at (\i) {\nscale{$t_{i,5,\i}$}};}
\graph {
	(import nodes);
			0 ->["\escale{$k$}"]1->["\escale{$u_{7}$}"]2->["\escale{$w_{3i+1}$}"]3->["\escale{$u_8$}"]4->["\escale{$k$}"]5;
			};
\end{scope}
\begin{scope}[yshift=-2.2cm,nodes={set=import nodes}]
		
		\foreach \i in {0,...,5} { \coordinate (\i) at (\i*1.5cm,0) ;}
		\foreach \i in {2} {\fill[red!20, rounded corners] (\i) +(-0.4,-0.25) rectangle +(0.4,0.35);}
		\foreach \i in {0,4} {\fill[red!20, rounded corners] (\i) +(-0.4,-0.25) rectangle +(0.4,0.4);}
		\foreach \i in {0,...,5} { \node (\i) at (\i) {\nscale{$t_{i,6,\i}$}};}
\graph {
	(import nodes);
			0 ->["\escale{$k$}"]1->["\escale{$u_{7}$}"]2->["\escale{$w_{3i+2}$}"]3->["\escale{$u_8$}"]4->["\escale{$k$}"]5;
			};
\end{scope}
\end{tikzpicture}
\end{center}
Finally, the gadget $B_i$ is defined as follows:
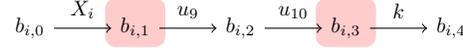
\begin{center}
\begin{tikzpicture}[new set = import nodes]
\begin{scope}[nodes={set=import nodes}]
		
		\foreach \i in {0,...,4} { \coordinate (\i) at (\i*1.4cm,0) ;}
		\foreach \i in {1,3} {\fill[red!20, rounded corners] (\i) +(-0.4,-0.25) rectangle +(0.4,0.4);}							\foreach \i in {0,...,4} { \node (\i) at (\i) {\nscale{$b_{i,\i}$}};}
\graph {
	(import nodes);
			0 ->["\escale{$X_i$}"]1->["\escale{$u_{9}$}"]2->["\escale{$u_{10}$}"]3->["\escale{$k$}"]4;
			};
\end{scope}

\end{tikzpicture}
\end{center}
Let $(sup, sig)$ be a $\tau$-region solving $\alpha$.
We first argue that the gadgets $H_0,\dots,H_3$ and $F_0,F_1$ and $G_0,\dots, G_{10}$ ensure that a $\tau$-region $(sup, sig)$ solving $\alpha$ satisfies $sig(k)=\inp$, $sup(h_{0,2})=0$ and $sig(u_0)=\dots=sig(u_{10})=\swap$.

If $sig(k)=\free$ and $sup(h_{0,2})=1$ then $s\edge{k}s'$ implies $sup(s)=sup(s')=0$.
Especially, by $sup(h_{0,1})=0$ and $sup(h_{0,2})=1$ we have $sig(y_0)=\swap$.
Moreover, by $sup(h_{2,1})=sup(h_{2,4})=0$ and $sig(y_0)=\swap$ we have that $sup(h_{2,2})=sup(h_{2,3})=1$.
This implies $sig(y_1)\in \{\nop,\used\}$.
By $sup(h_{1,1})=0$ and $h_{1,1}\edge{y_1}$ this implies $sig(y_1)=\nop$ and thus $sup(h_{1,2})=0$.
Furthermore, by $sup(h_{1,2})=sup(h_{1,3})=0$ and $h_{1,2}\edge{y_0}h_{1,3}$ this implies $sig(y_0)\not=\swap$, a contradiction.
Thus, we have $sig(k)\not=\free$.

If $sig(k)=\used$ and $sup(h_{0,2})=0$ then $s\edge{k}s'$ implies $sup(s)=sup(s')=1$.
Thus, we get $sup(h_{0,1})=sup(h_{0,3})=sup(h_{1,3})=1$ which with $sup(h_{0,2})=0$ implies $sig(y_0)=sig(v)=\swap$.
Moreover, $sup(h_{1,3})=1$ and $sig(y_0)=\swap$ imply $sup(h_{1,2})=0$.
By $sup(h_{1,1})=1$, this implies $sig(y_1)\in \{\inp,\res\}$.
Finally, $sup(h_{3,3})=1$ and $sig(v)=sig(y_0)=\swap$ imply $sup(h_{3,1})=1$.
This contradicts $sig(y_1)\in \{\inp,\res\}$.
Thus, $sig(k)\not=\used$.
Altogether, this shows that $sig(k)=\inp$ and $sup(h_{0,2})=0$, which implies $sig(v)=\swap$.

By $sig(k)=\inp$ we have $sup(f_{0,1})=sup(f_{1,1})=sup(g_{j,1})=0$ and $sup(f_{0,3})=sup(f_{1,3})=sup(g_{j,4})=1$.
By $sig(v)=\swap$, this implies $sup(f_{0,2})=sup(f_{1,2})=0$ and thus $sig(z_0), sig(z_1)\in \{\nop,\res,\free\}$.
Moreover, $sup(g_{j,1})=0$, $sup(g_{j,4})=1$ and $sig(z_0), sig(z_1)\in \{\nop,\res,\free\}$ imply $sup(g_{j,2})=0$ and $sup(g_{j,3})=1$ and thus $sig(u_j)=\swap$.

Let $i\in \{0,\dots,m-1\}$.
We now show that $T_{i,0},\dots, T_{i,6}$ and $B_i$ collaborate as announced.
By $sig(k)=\inp$ and $sig(u_{9})=sig(u_{10})=\swap$, we have $sup(b_{i,1})=1$ for all $i\in \{0,\dots, m-1\}$.
Since $\edge{X_i}b_{i,1}$ for all $i\in \{0,\dots, m-1\}$, the gadget $B_i$ ensures for all $X\in V(\varphi)$ that $s\edge{X}s'$ and $sup(s)\not=sup(s')$ imply $sig(X)=\swap$.

The gadget $T_{i,0}$ work like this:
By $sig(k)=\inp$ we get that $sup(t_{i,0,1})=0$ and $sup(t_{i,0,8})=1$.
Consequently, the image $sup(t_{i,0,1})\lEdge{sig(X_{i_0})}\dots \lEdge{sig(u_3)}sup(t_{i,0,8})$ of the path $t_{i,0,1}\edge{X_{i_0}}\dots\edge{u_3}t_{i,0,8}$ performs an uneven number of state changes from $0$ to $1$ in $\tau$.
Since $sig(u_0)=\dots=sig(u_3)=\swap$, the events $u_0,\dots, u_3$ perform an even number of state changes.
Thus, either all of $X_{i_0}, X_{i_1}, X_{i_2}$ are mapped to $\swap$ or exactly one of them.
The construction of $T_{i,1},\dots, T_{i,6}$ guarantees that there is exactly one variable event mapped to \swap.

In particular, the gadgets $T_{i,4}, T_{i,5}$ and $T_{i,6}$ ensure that if $e\in \{w_{3i}, w_{3i+1} ,w_{3i+2}\}$ then $sig(e)\not\in \{\nop,\used\}$.
We argue for $w_{3i}$:
By $sig(k)=\inp$ we get $sup(t_{i,4,1})=0$ and $sup(t_{i,4,4})=1$ which, by $sig(u_7)=sig(u_8)=\swap$, implies $sup(t_{i,4,2})=1$ and $sup(t_{i,4,3})=0$.
Clearly, this implies $sig(w_{3i})\not\in \{\nop,\used\}$.
Similarly, we obtain that $sig(w_{3i+1})\not\in \{\nop,\used\}$ and $sig(w_{3i+2})\not\in \{\nop,\used\}$.
 
Finally, the gadgets $T_{i,1}, T_{i,2}$ and $T_{i,3}$ ensure that never two variable events of the same clause can have a swap signature:
By $sig(k)=\inp$ we get that $sup(t_{i,1,1})=0$ and $sup(t_{i,1,7})=1$ which with $sig(u_4)=sig(u_5)=sig(u_6)=\swap$ implies $sup(t_{i,1,3})=0$ and $sup(t_{i,1,6})=0$.
Thus, if $sig(X_{i_0})=sig(X_{i_1})=\swap$ then $sup(t_{i,1,4})=sup(t_{i,1,5})=1$ which implies $sig(w_{3i})\in \{\nop, \used \}$, a contradiction.
Similarly, one uses $T_{i,2}$ and $T_{i,3}$ to show that neither $X_{i_0}$ and $X_{i_2}$ nor $X_{i_1}$ and $X_{i_2}$ can simultaneously be mapped to $\swap$.
As $i$ was arbitrary, there is exactly one variable per clause that is mapped to $\swap$.
Thus, $M=\{X\in V(\varphi) \mid sig(X)=\swap\}$ is a one-in-three model of $\varphi$.

Reversely, a one-in-three model $M$ of $\varphi$ allows a $\tau$-region $(sup, sig)$ that solves $\alpha$:
The red colored area above indicates which states of $H_0,\dots, H_3$, $F_0,F_1$, $G_0,\dots, G_{10}$ and $T_{0,4},T_{0,5},T_{0,6},\dots, T_{m-1,4},T_{m-1,5},T_{m-1,6}$ have positive support.
Moreover, we define $sup(s)=1$ for all $s\in \bot$.
Which states of $T_{i,0}, \dots, T_{i,3}$, where $i\in \{0,\dots, m-1\}$, besides of $k$'s sources get a positive support depends on which of $X_{i_0}, X_{i_1}, X_{i_2}$ belong to $M$.
The red colored area sketches the situation for $X_{i_0}\in M$.
It is easy to see that there is for all $e\in E(A^\tau_\varphi)$ a fitting $sig$-value making $(sup, sig)$ a (solving) $\tau$-region where $sig(k)=\inp$ and $sup(h_{0,2})=0$.
\end{proof}

\begin{theorem}\label{the:nop_inp_set_swap+out_res_used_free}
For any fixed $g\geq 1$, deciding if a $g$-bounded TS $A$ is $\tau$-solvable is NP-complete if $\tau=\{\nop,\inp,\set,\swap\}\cup\omega$ and $\omega\subseteq \{\out,\res,\used,\free\}$ or if $\tau=\{\nop,\out,\res,\swap\}\cup\omega$ and $\omega\subseteq \{\inp,\set,\used,\free\}$.
\end{theorem}
\begin{proof}
We present the reduction for the types built by $\tau=\{\nop, \inp,\set,\swap\}\cup\omega$ where $\omega \subseteq \{\out, \res, \used, \free\}$.
Again, the other types are covered by Lemma~\ref{lem:isomorphy}.

The TS $A^\tau_\varphi$ has the following gadgets $H_0,H_1,H_2$ and $H_3$: 
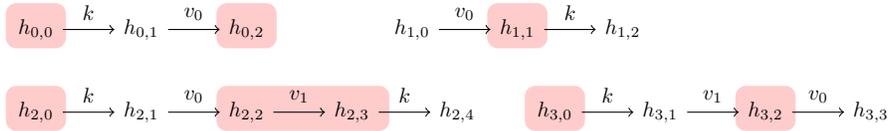
\begin{center}
\begin{tikzpicture}[new set = import nodes]
\begin{scope}[nodes={set=import nodes}]%

		\foreach \i in {0,...,2} { \coordinate (\i) at (\i*1.4cm,0) ;}
		\foreach \i in {0,2} {\fill[red!20, rounded corners] (\i) +(-0.4,-0.25) rectangle +(0.4,0.35);}
		\foreach \i in {0,...,2} { \node (\i) at (\i) {\nscale{$h_{0,\i}$}};}
\graph {(import nodes);
			0 ->["\escale{$k$}"]1->["\escale{$v_0$}"]2;  
			};
\end{scope}
\begin{scope}[xshift=5cm,nodes={set=import nodes}]%

		\foreach \i in {0,...,2} { \coordinate (\i) at (\i*1.4cm,0) ;}
		\foreach \i in {1} {\fill[red!20, rounded corners] (\i) +(-0.4,-0.25) rectangle +(0.4,0.35);}
		\foreach \i in {0,...,2} { \node (\i) at (\i) {\nscale{$h_{1,\i}$}};}
\graph {(import nodes);
			0 ->["\escale{$v_0$}"]1->["\escale{$k$}"]2;  
			};
\end{scope}
\begin{scope}[yshift=-1.1cm, nodes={set=import nodes}]%
		
		\foreach \i in {0,...,4} { \coordinate (\i) at (\i*1.4cm,0) ;}
		\foreach \i in {0} {\fill[red!20, rounded corners] (\i) +(-0.4,-0.25) rectangle +(0.4,0.35);}
		\foreach \i in {2} {\fill[red!20, rounded corners] (\i) +(-0.4,-0.25) rectangle +(1.9,0.35);}
		\foreach \i in {0,...,4} { \node (\i) at (\i) {\nscale{$h_{2,\i}$}};}
\graph {(import nodes);
			0 ->["\escale{$k$}"]1->["\escale{$v_0$}"]2->["\escale{$v_1$}"]3->["\escale{$k$}"]4;  
			};
\end{scope}

\begin{scope}[xshift=6.9cm, yshift=-1.1cm,nodes={set=import nodes}]%
		
		\foreach \i in {0,...,3} { \coordinate (\i) at (\i*1.4cm,0) ;}
		\foreach \i in {0,2} {\fill[red!20, rounded corners] (\i) +(-0.4,-0.25) rectangle +(0.4,0.35);}
		\foreach \i in {0,...,3} { \node (\i) at (\i) {\nscale{$h_{3,\i}$}};}
\graph {
	(import nodes);
			0 ->["\escale{$k$}"]1->["\escale{$v_1$}"]2->["\escale{$v_0$}"]3;  
			};
\end{scope}
\end{tikzpicture}
\end{center}
If $\tau\cap\{\used,\free\}\not=\emptyset$ then $A^\tau_\varphi$ has also the following gadgets $H_4,\dots, H_9 $:
\begin{center}
\begin{tikzpicture}[new set = import nodes]
\begin{scope}[nodes={set=import nodes}]%
		\foreach \i in {0,...,4} { \coordinate (\i) at (\i*1.2cm,0) ;}
		\foreach \i in {0,3} {\fill[red!20, rounded corners] (\i) +(-0.4,-0.25) rectangle +(0.4,0.35);}
		\foreach \i in {0,...,4} { \node (\i) at (\i) {\nscale{$h_{4,\i}$}};}
\graph {
	(import nodes);
			0 ->["\escale{$k$}"]1->["\escale{$x$}"]2->["\escale{$v_0$}"]3->["\escale{$k$}"]4;  
			};
\end{scope}
\begin{scope}[xshift=6cm,nodes={set=import nodes}]%
		\foreach \i in {0,...,4} { \coordinate (\i) at (\i*1.2cm,0) ;}
		\foreach \i in {0} {\fill[red!20, rounded corners] (\i) +(-0.4,-0.25) rectangle +(0.4,0.35);}
		\foreach \i in {2} {\fill[red!20, rounded corners] (\i) +(-0.4,-0.25) rectangle +(1.6,0.35);}
		\foreach \i in {0,...,4} { \node (\i) at (\i) {\nscale{$h_{5,\i}$}};}
\graph {
	(import nodes);
			0 ->["\escale{$k$}"]1->["\escale{$v_0$}"]2->["\escale{$x$}"]3->["\escale{$k$}"]4;  
			};
\end{scope}
\begin{scope}[yshift=-1.1cm,nodes={set=import nodes}]%
		\foreach \i in {0,...,4} { \coordinate (\i) at (\i*1.2cm,0) ;}
		\foreach \i in {0,3} {\fill[red!20, rounded corners] (\i) +(-0.4,-0.25) rectangle +(0.4,0.35);}
		\foreach \i in {0,...,4} { \node (\i) at (\i) {\nscale{$h_{6,\i}$}};}
\graph {
	(import nodes);
			0 ->["\escale{$k$}"]1->["\escale{$x$}"]2->["\escale{$y_0$}"]3->["\escale{$k$}"]4;  
			};
\end{scope}
\begin{scope}[xshift=6cm,yshift=-1.1cm,nodes={set=import nodes}]%
		\foreach \i in {0,...,4} { \coordinate (\i) at (\i*1.2cm,0) ;}
		\foreach \i in {0} {\fill[red!20, rounded corners] (\i) +(-0.4,-0.25) rectangle +(0.4,0.35);}
		\foreach \i in {2} {\fill[red!20, rounded corners] (\i) +(-0.4,-0.25) rectangle +(1.6,0.35);}
		\foreach \i in {0,...,4} { \node (\i) at (\i) {\nscale{$h_{7,\i}$}};}
\graph {
	(import nodes);
			0 ->["\escale{$k$}"]1->["\escale{$y_0$}"]2->["\escale{$x$}"]3->["\escale{$k$}"]4;  
			};
\end{scope}
\end{tikzpicture}
\end{center}
\begin{center}
\begin{tikzpicture}[new set = import nodes]
\begin{scope}[yshift=-2.2cm,nodes={set=import nodes}]%
		\foreach \i in {0,...,4} { \coordinate (\i) at (\i*1.2cm,0) ;}
		\foreach \i in {0,3} {\fill[red!20, rounded corners] (\i) +(-0.4,-0.25) rectangle +(0.4,0.35);}
		\foreach \i in {0,...,4} { \node (\i) at (\i) {\nscale{$h_{8,\i}$}};}
\graph {
	(import nodes);
			0 ->["\escale{$k$}"]1->["\escale{$x$}"]2->["\escale{$y_1$}"]3->["\escale{$k$}"]4;  
			};
\end{scope}
\begin{scope}[xshift=6cm,yshift=-2.2cm,nodes={set=import nodes}]%
		\foreach \i in {0,...,4} { \coordinate (\i) at (\i*1.2cm,0) ;}
		\foreach \i in {0} {\fill[red!20, rounded corners] (\i) +(-0.4,-0.25) rectangle +(0.4,0.35);}
		\foreach \i in {2} {\fill[red!20, rounded corners] (\i) +(-0.4,-0.25) rectangle +(1.6,0.35);}
		\foreach \i in {0,...,4} { \node (\i) at (\i) {\nscale{$h_{9,\i}$}};}
\graph {
	(import nodes);
			0 ->["\escale{$k$}"]1->["\escale{$y_1$}"]2->["\escale{$x$}"]3->["\escale{$k$}"]4;  
			};
\end{scope}
\begin{scope}[yshift=-3.3cm,nodes={set=import nodes}]%
		\foreach \i in {0,...,4} { \coordinate (\i) at (\i*1.2cm,0) ;}
		\foreach \i in {0,3} {\fill[red!20, rounded corners] (\i) +(-0.4,-0.25) rectangle +(0.4,0.35);}
		\foreach \i in {0,...,4} { \node (\i) at (\i) {\nscale{$h_{10,\i}$}};}
\graph {
	(import nodes);
			0 ->["\escale{$k$}"]1->["\escale{$x$}"]2->["\escale{$y_2$}"]3->["\escale{$k$}"]4;  
			};
\end{scope}
\begin{scope}[xshift=6cm,yshift=-3.3cm,nodes={set=import nodes}]%
		\foreach \i in {0,...,4} { \coordinate (\i) at (\i*1.2cm,0) ;}
		\foreach \i in {0} {\fill[red!20, rounded corners] (\i) +(-0.4,-0.25) rectangle +(0.4,0.35);}
		\foreach \i in {2} {\fill[red!20, rounded corners] (\i) +(-0.4,-0.25) rectangle +(1.6,0.35);}
		\foreach \i in {0,...,4} { \node (\i) at (\i) {\nscale{$h_{11,\i}$}};}
\graph {
	(import nodes);
			0 ->["\escale{$k$}"]1->["\escale{$y_2$}"]2->["\escale{$x$}"]3->["\escale{$k$}"]4;  
			};
\end{scope}
\begin{scope}[yshift=-4.4cm,nodes={set=import nodes}]%
		\foreach \i in {0,...,5} { \coordinate (\i) at (\i*1.4cm,0) ;}
		\foreach \i in {0} {\fill[red!20, rounded corners] (\i) +(-0.4,-0.25) rectangle +(0.4,0.35);}
		\foreach \i in {2} {\fill[red!20, rounded corners] (\i) +(-0.4,-0.25) rectangle +(3.2,0.35);}
		\foreach \i in {0,...,5} { \node (\i) at (\i) {\nscale{$h_{12,\i}$}};}
\graph {
	(import nodes);
			0 ->["\escale{$k$}"]1->["\escale{$y_0$}"]2->["\escale{$y_1$}"]3->["\escale{$y_2$}"]4->["\escale{$k$}"]5;  
			};
\end{scope}
\end{tikzpicture}
\end{center}
The gadgets $H_0,\dots, H_3$ and $H_4,\dots, H_9$, if added, provide $\alpha=(k,h_{3,3})$ and ensure that a $\tau$-region $(sup,sig)$ solving $\alpha$ satisfies $sig(k)\in \{\inp,\out\}$.

The TS $A^\tau_\varphi$ adds the following gadgets $F_0,F_1,F_2$ and $G_i, N_i$, $\forall i\in \{0,\dots, 13\}$:
\begin{center}
\begin{tikzpicture}[new set = import nodes]
\begin{scope}[nodes={set=import nodes}]%
		\foreach \i in {0,...,4} { \coordinate (\i) at (\i*1.3cm,0) ;}
		\foreach \i in {0,3} {\fill[red!20, rounded corners] (\i) +(-0.4,-0.25) rectangle +(0.4,0.35);}
		\foreach \i in {0,...,4} { \node (\i) at (\i) {\nscale{$f_{0,\i}$}};}
\graph {
	(import nodes);
			0 ->["\escale{$k$}"]1->["\escale{$z_0$}"]2->["\escale{$v_0$}"]3->["\escale{$k$}"]4;
			};
\end{scope}
\begin{scope}[xshift=6.3cm, nodes={set=import nodes}]%
		\foreach \i in {0,...,4} { \coordinate (\i) at (\i*1.3cm,0) ;}
		\foreach \i in {0,3} {\fill[red!20, rounded corners] (\i) +(-0.4,-0.25) rectangle +(0.4,0.35);}
		\foreach \i in {0,...,4} { \node (\i) at (\i) {\nscale{$f_{1,\i}$}};}
\graph {
	(import nodes);
			0 ->["\escale{$k$}"]1->["\escale{$z_1$}"]2->["\escale{$v_0$}"]3->["\escale{$k$}"]4;
			};
\end{scope}
\begin{scope}[yshift=-1.2cm, nodes={set=import nodes}]
		
		\foreach \i in {0,...,5} { \coordinate (\i) at (\i*1.3cm,0) ;}
		\foreach \i in {0,4} {\fill[red!20, rounded corners] (\i) +(-0.4,-0.25) rectangle +(0.4,0.35);}
		\foreach \i in {0,...,5} { \node (\i) at (\i) {\nscale{$f_{2,\i}$}};}
\graph {
	(import nodes);
			0 ->["\escale{$k$}"]1->["\escale{$z_0$}"]2->["\escale{$z_1$}"]3->["\escale{$z_2$}"]4->["\escale{$k$}"]5;
			};
\end{scope}
\begin{scope}[yshift=-2.4cm,nodes={set=import nodes}]
		
		\foreach \i in {0,...,5} { \coordinate (\i) at (\i*1.05cm,0) ;}
		\foreach \i in {0} {\fill[red!20, rounded corners] (\i) +(-0.3,-0.25) rectangle +(0.3,0.35);}
		\foreach \i in {3} {\fill[red!20, rounded corners] (\i) +(-0.4,-0.25) rectangle +(1.4,0.35);}
		\foreach \i in {0,...,5} { \node (\i) at (\i) {\nscale{$g_{i,\i}$}};}
\graph {
	(import nodes);
			0 ->["\escale{$k$}"]1->["\escale{$z_0$}"]2->["\escale{$u_i$}"]3->["\escale{$z_1$}"]4->["\escale{$k$}"]5;
			};
\end{scope}
\begin{scope}[xshift=6.2cm,yshift=-2.4cm, nodes={set=import nodes}]
		
		\foreach \i in {0,...,5} { \coordinate (\i) at (\i*1.05cm,0) ;}
		\foreach \i in {0,2,4} {\fill[red!20, rounded corners] (\i) +(-0.3,-0.25) rectangle +(0.3,0.35);}
		\foreach \i in {0,...,5} { \node (\i) at (\i) {\nscale{$n_{i,\i}$}};}
\graph {
	(import nodes);
			0 ->["\escale{$k$}"]1->["\escale{$z_2$}"]2->["\escale{$u_i$}"]3->["\escale{$v_0$}"]4->["\escale{$k$}"]5;
			};
\end{scope}
\end{tikzpicture}
\end{center}
The gadgets $F_0,F_1,F_2$ and $G_0,N_0,\dots, G_{13}, N_{13}$ guarantee that if $(sup, sig)$ solves $\alpha$ then $sig(u_i)=\swap$.
Similarly to the reduction of Theorem~\ref{the:nop_inp_res_swap+used_free}, the TS $A^\tau_\varphi$ has for every $i\in \{0,\dots, m-1\}$ gadgets $T_{i,0},\dots, T_{i,6}$ and $B_i$ to make the one-in-three satisfiability of $\varphi$ and the $\tau$-solvability of $\alpha$ the same.
These gadgets and the ones for Theorem~\ref{the:nop_inp_res_swap+used_free} have basically the same intention.
However, since the current types have different interactions, the peculiarity of these gadgets is slightly different.
In particular, $A^\tau_\varphi$ has for each clause $\zeta_i=\{X_{i_0}, X_{i_1}, X_{i_2}\}$ the following gadget $T_{i,0}$:
\begin{center}
\begin{tikzpicture}[new set = import nodes]
\begin{scope}[nodes={set=import nodes}]%
		
		\foreach \i in {0,...,9} { \coordinate (\i) at (\i*1.25cm,0) ;}
		\foreach \i in {0,2,8} {\fill[red!20, rounded corners] (\i) +(-0.3,-0.25) rectangle +(0.3,0.4);}
		\foreach \i in {4} {\fill[red!20, rounded corners] (\i) +(-0.3,-0.25) rectangle +(1.6,0.4);}
		\foreach \i in {0,...,9} { \node (\i) at (\i) {\nscale{$t_{i,0,\i}$}};}
\graph {
	(import nodes);
			0 ->["\escale{$k$}"]1->["\escale{$u_{0}$}"]2->["\escale{$X_{i_0}$}"]3->["\escale{$u_{1}$}"]4->["\escale{$X_{i_1}$}"]5->["\escale{$u_{2}$}"]6->["\escale{$X_{i_2}$}"]7->["\escale{$u_{3}$}"]8->["\escale{$k$}"]9;
			};
\end{scope}
\end{tikzpicture}
\end{center}
Moreover, the gadgets $T_{i,1}, T_{i,2}$ and $T_{i,3}$ are defined as follows:
\begin{center}
\begin{tikzpicture}[new set = import nodes]
\begin{scope}[nodes={set=import nodes}]
		
		\foreach \i in {0,...,8} { \coordinate (\i) at (\i*1.4cm,0) ;}
		\foreach \i in {0,2,7} {\fill[red!20, rounded corners] (\i) +(-0.4,-0.25) rectangle +(0.4,0.4);}
		\foreach \i in {4} {\fill[red!20, rounded corners] (\i) +(-0.4,-0.25) rectangle +(1.8,0.4);}
		\foreach \i in {0,...,8} { \node (\i) at (\i) {\nscale{$t_{i,1,\i}$}};}
\graph {
	(import nodes);
			0 ->["\escale{$k$}"]1->["\escale{$u_{4}$}"]2->["\escale{$X_{i_0}$}"]3->["\escale{$w_{3i}$}"]4->["\escale{$X_{i_1}$}"]5->["\escale{$u_{5}$}"]6->["\escale{$u_{6}$}"]7->["\escale{$k$}"]8;
			};
\end{scope}
\begin{scope}[yshift=-1.1cm,nodes={set=import nodes}]
		
		\foreach \i in {0,...,8} { \coordinate (\i) at (\i*1.4cm,0) ;}
		\foreach \i in {0,5,7} {\fill[red!20, rounded corners] (\i) +(-0.4,-0.25) rectangle +(0.4,0.4);}
		\foreach \i in {2} {\fill[red!20, rounded corners] (\i) +(-0.4,-0.25) rectangle +(1.75,0.4);}
		\foreach \i in {0,...,8} { \node (\i) at (\i) {\nscale{$t_{i,2,\i}$}};}
\graph {
	(import nodes);
			0 ->["\escale{$k$}"]1->["\escale{$u_{4}$}"]2->["\escale{$X_{i_2}$}"]3->["\escale{$w_{3i+1}$}"]4->["\escale{$X_{i_0}$}"]5->["\escale{$u_{5}$}"]6->["\escale{$u_{6}$}"]7->["\escale{$k$}"]8;
			};
\end{scope}
\begin{scope}[yshift=-2.2cm,nodes={set=import nodes}]
		
		\foreach \i in {0,...,8} { \coordinate (\i) at (\i*1.4cm,0) ;}
		\foreach \i in {0,7} {\fill[red!20, rounded corners] (\i) +(-0.4,-0.25) rectangle +(0.4,0.4);}
		\foreach \i in {2} {\fill[red!20, rounded corners] (\i) +(-0.4,-0.25) rectangle +(4.7,0.4);}
		\foreach \i in {0,...,8} { \node (\i) at (\i) {\nscale{$t_{i,3,\i}$}};}
\graph {
	(import nodes);
			0 ->["\escale{$k$}"]1->["\escale{$u_{4}$}"]2->["\escale{$X_{i_1}$}"]3->["\escale{$w_{3i+2}$}"]4->["\escale{$X_{i_2}$}"]5->["\escale{$u_{5}$}"]6->["\escale{$u_{6}$}"]7->["\escale{$k$}"]8;
			};
\end{scope}
\end{tikzpicture}
\end{center}
Furthermore, the gadgets $T_{i,4}, T_{i,5}$ and $T_{i,6}$ are defined by 
\begin{center}
\begin{tikzpicture}[new set = import nodes]
\begin{scope}[nodes={set=import nodes}]%
		
		\foreach \i in {0,...,7} { \coordinate (\i) at (\i*1.5cm,0) ;}
		\foreach \i in {0,2,4,6} {\fill[red!20, rounded corners] (\i) +(-0.4,-0.25) rectangle +(0.4,0.4);}
		\foreach \i in {0,...,7} { \node (\i) at (\i) {\nscale{$t_{i,4,\i}$}};}
\graph {
	(import nodes);
			0 ->["\escale{$k$}"]1->["\escale{$u_{7}$}"]2->["\escale{$u_{8}$}"]3->["\escale{$w_{3i}$}"]4->["\escale{$u_{9}$}"]5->["\escale{$u_{10}$}"]6->["\escale{$k$}"]7;
			};
\end{scope}
\end{tikzpicture}
\end{center}
\begin{center}
\begin{tikzpicture}[new set = import nodes]
\begin{scope}[yshift=-1.2cm,nodes={set=import nodes}]
		
		\foreach \i in {0,...,7} { \coordinate (\i) at (\i*1.5cm,0) ;}
		\foreach \i in {0,2,4,6} {\fill[red!20, rounded corners] (\i) +(-0.4,-0.25) rectangle +(0.4,0.4);}
		\foreach \i in {0,...,7} { \node (\i) at (\i) {\nscale{$t_{i,5,\i}$}};}
\graph {
	(import nodes);
			0 ->["\escale{$k$}"]1->["\escale{$u_{7}$}"]2->["\escale{$u_{8}$}"]3->["\escale{$w_{3i+1}$}"]4->["\escale{$u_{9}$}"]5->["\escale{$u_{10}$}"]6->["\escale{$k$}"]7;
			};
\end{scope}
\end{tikzpicture}
\end{center}
\begin{center}
\begin{tikzpicture}[new set = import nodes]
\begin{scope}[yshift=-2.4cm,nodes={set=import nodes}]
		
		\foreach \i in {0,...,7} { \coordinate (\i) at (\i*1.5cm,0) ;}
		\foreach \i in {0,2,4,6} {\fill[red!20, rounded corners] (\i) +(-0.4,-0.25) rectangle +(0.4,0.4);}
		\foreach \i in {0,...,7} { \node (\i) at (\i) {\nscale{$t_{i,6,\i}$}};}
\graph {
	(import nodes);
			0 ->["\escale{$k$}"]1->["\escale{$u_{7}$}"]2->["\escale{$u_{8}$}"]3->["\escale{$w_{3i+2}$}"]4->["\escale{$u_{9}$}"]5->["\escale{$u_{10}$}"]6->["\escale{$k$}"]7;
			};
\end{scope}
\end{tikzpicture}
\end{center}
Finally, the TS $A^\tau_\varphi$ has for all $i\in \{0,\dots, m-1\}$ the following gadget $B_i$:
\begin{center}
\begin{tikzpicture}[new set = import nodes]
\begin{scope}[nodes={set=import nodes}]
		
		\foreach \i in {0,...,3} { \coordinate (\i) at (\i*1.3cm,0) ;}
		\foreach \i in {2} {\fill[red!20, rounded corners] (\i) +(-0.4,-0.25) rectangle +(0.4,0.4);}
		\foreach \i in {0,...,3} { \node (\i) at (\i) {\nscale{$b_{i,\i}$}};}
\graph {
	(import nodes);
			0 ->["\escale{$X_i$}"]1->["\escale{$u_{11}$}"]2->["\escale{$k$}"]3;
			};
\end{scope}

\end{tikzpicture}
\end{center}
We briefly argue for the announced functionality of the gadgets.
Let $(sup, sig)$ be a region that solves $\alpha$.
If $sig(k)=\free$ then $sup(h_{3,3})=1$ and $s\edge{k}s'$ implies $sup(s)=sup(s')=0$.
Since $sup(h_{3,1})=0$ and $sup(h_{3,3})=1$, there is an event $e\in \{v_0,v_1\}$ such that $sig(e)\in \{\out,\set,\swap\}$.
If $sig(v_0)\in \{\out,\set,\swap\}$ then by $sup(h_{1,1})=0$ we get $sig(v_{0})=\swap$.
Moreover, if $sig(v_1)\in \{\out,\set,\swap\}$, which implies $sig(h_{3,2})=1$, then by $sup(h_{2,3})=0$ we get $sig(v_1)=\swap$.
By $sig(v_1)=\swap$ and $sup(h_{2,3})=0$ we get $sup(h_{2,2})=1$ which implies with $sup(h_{1,1})$ that $sig(v_0)=\swap$.
Thus, in any case we get $sig(v_0)=\swap$.
By $sig(v_0)=\swap$ and $sup(h_{4,3})=sup(h_{5,1})=0$ we obtain  $sup(h_{4,2})=sup(h_{5,2})=1$ which implies $sig(x)=\swap$.
Using this and $sup(s)=sup(s')=0$ if $s\edge{k}s'$, we have that $sup(h_{j,2})=1$ for all $j\in \{6,\dots, 11\}$.
This implies $sig(y_0)=sig(y_1)=sig(y_2)=\swap$.
Consequently, by $h_{9,1}\edge{y_0}\dots\edge{y_2}h_{9,4}$, there is an uneven number of state changes from $0$ to $0$ in $\tau$.
This is a contradiction.
Thus, $sig(k)\not=\free$.

Similarly, if $sig(k)=\used$ and $sup(h_{3,3})=0$ then there is an uneven number of state changes from $1$ to $1$ in $\tau$, a contradiction.
Thus, $sig(k)\not=\used$.

We conclude that $sig(k)=\inp$ and $sup(h_{3,3})=0$ which implies $sig(v_0)\not\in\{\out,\set\}$ and $sup(s)=0$ and $sup(s')=1$ if $s\edge{k}s'$.
Thus, by $sup(h_{2,1})=0$ and $sup(h_{2,3})=1$ there is an event $e\in \{v_0,v_1\}$ such that $sig(e)\in \{\out,\set,\swap\}$.
If $e=v_0$ then $sig(v_0)=\swap$.
Moreover, if $e=v_1$ then $sup(h_{3,2})=1$ which with $sup(h_{3,3})=0$ and $sup(h_{1,1})=1$ implies $sig(v_0)=\swap$.
Consequently, any case implies $sig(v_0)=\swap$.
This results in $sig(u_j)=\swap$ for all $j\in \{0,\dots, 13\}$ as follows.
By $sup(f_{0,3})=sup(f_{1,3})=1$ and $sig(v)=\swap$ we obtain $sup(f_{0,2})=sup(f_{1,2})=0$ which with $sup(f_{0,1})=sup(f_{1,1})=0$ implies $sig(z_0), sig(z_1)\in \{\nop,\res,\free\}$.
Moreover, by $sig(z_0), sig(z_1)\in \{\nop,\res,\free\}$ and $sup(f_{2,1})=0$ we get $sup(f_{2,3})=0$ which with $sup(f_{2,4})=1$ implies $sig(z_2)\in \{\out,\set,\swap\}$.
By $sig(z_0)\in \{\nop,\res,\free\}$ and $sup(g_{i,1})=0$, we get $sup(g_{i,2})=0$.
Furthermore, $sig(z_1)\in \{\nop,\res,\free\}$ and $sup(g_{i,4})=1$ yields $sig(z_1)=\nop$ and $sup(g_{i,3})=1$.
This implies $sig(u_i)\in \{\out,\set,\swap\}$.
Finally, by $sup(n_{i,1})=0$ and $sig(z_2)\in \{\out,\set,\swap\}$, we get $sup(n_{i,1})=1$ and, by $sup(n_{i,4})=1$ and $sig(v_0)=\swap$, we have $sup(n_{i,3})=0$.
Since $sig(u_i)\in \{\out,\set,\swap\}$, this yields $sig(u_i)=\swap$ for all $i\in \{0,\dots, 13\}$.

The gadgets $T_{i,0},\dots, T_{i,6}$, where $i\in \{0,\dots, m-1\}$, use $sig(k)=\inp$ and $sig(u_i)=\swap$ for all $i\in \{0,\dots, 13\}$ similarly to the ones of Theorem~\ref{the:nop_inp_res_swap+used_free} to ensure that $M=\{X\in V(\varphi) \mid sig(X)=\swap\}$ is a one-in-three model of $\varphi$:
By $sup(t_{i,4,6})=sup(t_{i,5,6})=sup(t_{i,6,6})=1$ and $sig(u_{5})=sig(u_{6})=\swap$ we have $sup(t_{i,4,4})=sup(t_{i,5,4})=sup(t_{i,6,4})=1$ for all $i\in \{0,\dots, m-1\}$.
Thus, if $X\in V(\varphi)$, $s\edge{X}s'$ and $sup(s)\not=sup(s')$ then $sig(X)=\swap$.
Using this, one argues quite similar to the proof of Theorem~\ref{the:nop_inp_res_swap+used_free} that $T_{i,0},\dots, T_{i,6}$ collaborate in way that there is exactly one variable event $X\in \{X_{i_0}, X_{i_1}, X_{i_2}\}$ such that $sig(X)=\swap$.
Thus, $M$ is a corresponding model.
Moreover, if $sig(k)=\out$ and $sup(h_{3,3})=1$ then we obtain again that $sig(u_i)=\swap$ for all $i\in \{0,\dots, 13\}$ which also guarantees that $M$ is a searched model.

Reversely, if $M$ is a one-in-three model of $\varphi$ then we can define analogously to Theorem~\ref{the:nop_inp_res_swap+used_free} a $\tau$-region solving $\alpha$.
\end{proof}

\begin{theorem}[\cite{DBLP:conf/apn/Tredup19}]
For any fixed $g\geq 1$, deciding if a $g$-bounded TS $A$ is $\tau$-solvable is NP-complete if $\tau\in \{\nop,\inp,\out\}\cup\{\used,\free\}$.
\end{theorem}
\begin{proof}
The claim follows directly from our result of \cite{DBLP:conf/apn/Tredup19}.
There we use $1$-bounded cycle free gadgets to prove that synthesis of (pure) $b$-bounded is NP-complete.
The joining of \cite{DBLP:conf/apn/Tredup19} yields a $2$-bounded TS.
However, it is easy to see that the $1$-bounded joining of this paper fits, too.
Moreover, the (pure) $1$-bounded type is isomorphic to $\{\nop,\inp,\out,\used\}$ ($\{\nop,\inp,\out\}$) and, by symmetry, $\tau$-solving ESSP atoms by $\used$ is equivalent to solving them by $\free$.
\end{proof}

\section{Polynomial Time Results}%

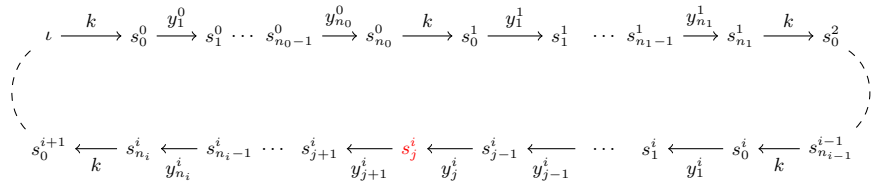
\begin{figure}[b!]
\begin{center}
\begin{tikzpicture}[new set = import nodes]
\begin{scope}[nodes={set=import nodes}]%
		\node (0) at (0,0) {\nscale{$\iota$}};
		\node (s00) at (1.2,0) {\nscale{$s^0_0$}};
		\node (s01) at (2.2,0) {\nscale{$s^0_1$}};
		\node (dots0) at (2.6,0) {\nscale{$\dots$}};
		\node (s0n_1) at (3.2,0) {\nscale{$s^0_{n_0-1}$}};
		\node (s0n) at (4.4,0) {\nscale{$s^0_{n_0}$}};
		\node (s10) at (5.6,0) {\nscale{$s^1_0$}};
		\node (s11) at (6.8,0) {\nscale{$s^1_1$}};
		\node (dots1) at (7.4,0) {\nscale{$\dots$}};
		\node (s1n_1) at (8,0) {\nscale{$s^1_{n_1-1}$}};
		\node (s1n) at (9.2,0) {\nscale{$s^1_{n_1}$}};
		\node (s20) at (10.4,0) {\nscale{$s^2_0$}};
		\node (si_1ni_1) at (10.4,-1.5) {\nscale{$s^{i-1}_{n_{i-1}}$}};
		\node (si0) at (9.2,-1.5) {\nscale{$s^i_0$}};
		\node (si1) at (8,-1.5) {\nscale{$s^i_1$}};
		\node (dotsi) at (7.4,-1.5) {\nscale{$\dots$}};
		\coordinate (c) at (7,-1.5);
		
		\node (sij_1) at (6,-1.5) {\nscale{$s^i_{j-1}$}};
		\node (sij) at (4.8,-1.5) {\nscale{\textcolor{red}{$s^i_j$}}};
		\node (sij+1) at (3.6,-1.5) {\nscale{$s^i_{j+1}$}};
		\node (dotsi) at (3,-1.5) {\nscale{$\dots$}};
		\node (sin_1) at (2.4,-1.5) {\nscale{$s^i_{n_i-1}$}};
		\node (sin) at (1.2,-1.5) {\nscale{$s^i_{n_i}$}};
		\node (si+10) at (0,-1.5) {\nscale{$s^{i+1}_0$}};
\graph {
	(import nodes);
			0 ->["\escale{$k$}"]s00->["\escale{$y^0_1$}"]s01;  
			s0n_1->["\escale{$y^0_{n_0}$}"]s0n->["\escale{$k$}"]s10 ->["\escale{$y^1_1$}"]s11;
			s1n_1->["\escale{$y^1_{n_1}$}"]s1n->["\escale{$k$}"]s20;
			si_1ni_1->["\escale{$k$}"]si0->["\escale{$y^i_1$}"]si1;
			
			c  ->["\escale{$y^i_{j-1}$}"]sij_1->["\escale{$y^i_{j}$}"]sij ->["\escale{$y^i_{j+1}$}"]sij+1 ;
			sin_1->["\escale{$y^i_{n_i}$}"]sin ->["\escale{$k$}"]si+10;
			s20--[bend left=50,dashed]si_1ni_1;
			si+10--[bend left=50,dashed]0;
			};
\end{scope}
\end{tikzpicture}
\end{center}
\caption{A sketch of a cyclic $1$-bounded input $A$ with ESSP atom $\alpha=(k,s^i_j)$.}\label{fig:polynomial}
\end{figure}
\begin{theorem}\label{the:nop_inp_set_1bounded}
For any fixed $g < 2$, one can decide in polynomial time if a $g$-bounded TS $A$ is $\tau$-solvable if $\tau=\{\nop,\inp,\set\}$ or $\tau=\{\nop,\inp,\set,\used\}$ or $\tau=\{\nop,\out,\res\}$ or $\tau=\{\nop,\out,\res,\free\}$ or $\tau=\{\nop,\set,\res\}\cup\omega$ with non-empty $\omega\subseteq\{\inp,\out,\used,\free\}$.
\end{theorem}
\begin{proof}
If $A$ is $\tau$-solvable then no event $e$ of $A$ occurs twice in a row.
Otherwise, the SSP atom $(s',s'')$ of a sequence $s\edge{e}s'\edge{e}s''$ is not $\tau$-solvable.
Thus, in what follows, we assume that $A$ has no event occurring twice in a row.
Moreover, it is easy to see that a $1$-bounded TS $A=s_0\edge{e_1}\dots\edge{e_m}s_m$ is a simple directed path on pairwise distinct states $s_0,\dots, s_m$ or a directed cycle, that is, all states $s_0,\dots, s_m$ except $s_0$ and $s_m$ are pairwise distinct.
This proof proceeds as follows.
First, we assume that $\tau=\{\nop,\inp,\set\}$ and that $A$ is a directed cycle and argue that the $\tau$-solvability of a given ESSP atom $(k,s)$ or a SSP atom $(s,s')$ of $A$ is decidable in polynomial time.
Secondly, we argue that the presented algorithmic approach is applicable to directed paths, too.
Thirdly, we show that the procedure introduced for $\{\nop,\inp,\set\}$ can be extended to $\{\nop,\inp,\set,\used\}$.
By Lemma~\ref{lem:isomorphy}, this proves the claim for $\{\nop,\out,\res\}$ and $\{\nop,\out,\res,\free\}$, too. 
After that we investigate the case where $\tau=\{\nop,\set,\res\}\cup\omega$ with non-empty $\omega\subseteq\{\inp,\out,\used,\free\}$.
We argue that it is sufficient to decide the $\{\nop,\inp,\res,\set\}$- and $\{\nop,\res,\set,\used\}$-solvability of $A$ and that this is doable in polynomial time.
The corresponding procedures again modify those introduced for $\{\nop,\inp,\set\}$.

Let $\tau=\{\nop,\inp,\set\}$ and $A$ be $1$-bounded (cycle) TS with event $k\in E(A)$ that occurs $m$ times.
Since $A$ is a cycle, we can assume that $k$ occurs at $A$'s initial state: $\iota\edge{k}$.
Moreover, since $k$ does not occur twice in a row, its occurrences partition $A$ into $m$ $k$-free subsequences $I_0,\dots, I_{m-1}$ such that $I_i=s^i_0\Edge{y^i_1}s^i_1\dots s^i_{n_i-1}\Edge{y^i_{n_i}}s^i_{n_i}$, $i\in \{0,\dots, m-1\}$, and $s^{m-1}_{n_{m-1}}=\iota$, cf. Figure~\ref{fig:polynomial}.

Obviously, defining $sup(\iota)=1$, $sig(k)=\inp$ and $sig(e)=\set$ for all $e\in E(A)\setminus \{k\}$ inductively yields a region $(sup, sig)$ solving the ESSP atoms $(k,s)$ where $\edge{e}s$.
Thus, it remains to consider the case $\neg(\edge{k}s)$.
Since $\neg (\edge{k}s)$, there is an $i\in \{0,\dots, m-1\}$ such that $s$ is a state of the $i$-th subsequence $I_i$.
In particular, there is a $j\in \{1,\dots,n_i-1 \}$ such that $s=s^i_j$.
The state $s^i_j$ divides $I_i$ into the sequences $I^0_i=s^i_0\Edge{y^i_1}\dots \Edge{y^i_j}s^i_j$ and $I^1_i=s^i_j\Edge{y^i_{j+1}}\dots \Edge{y^i_{n_i}}s^i_{n_i}$, cf. Figure~\ref{fig:polynomial}.

If $(sup, sig)$ is a region that solves $\alpha$ then $sig(k)=\inp$ and $sup(s^i_j)=0$ is true.
This implies for all $\ell\in \{0,\dots, m-1\}$ that $sup(s^\ell_0)=0$ and $sup(s^\ell_{n_\ell})=1$.
Thus, it remains to define the signature of the events of $\bigcup_{\ell=0}^{m-1}E(I_\ell)$ such that $0\lEdge{sig(y^\ell_1)}\dots \lEdge{sig(y^\ell_{n_\ell})}1$, for all $\ell\in \{0,\dots, m-1\}\setminus \{i\}$, and $0\lEdge{sig(y^i_1)}\dots \lEdge{sig(y^i_j)}0$ and $0\lEdge{sig(y^i_{j+1})}\dots \lEdge{sig(y^i_{n_i})}1$.

If there is for all $\ell\in \{0,\dots,m-1\}\setminus\{i\}$ an event $e_\ell\in E(I_\ell)$ such that $e_\ell\not\in E(I^0_i)$ and if there is an event $e_i\in E(I^1_i)$ so that $e_i\not\in E(I^0_i)$ then $sup(\iota)=1$, $sig(k)=\inp$, $sig(e_\ell)=\set$ for all $\ell\in \{0,\dots, \ell\}$ and $sig(e)=\nop$ for all $e\in E(A)\setminus \{k,e_0,\dots,e_\ell\}$ yields a $\tau$-region $(sup, sig)$ of $A$ that solves $\alpha$.
Clearly, whether $A$ satisfies this property is decidable in polynomial time.

Otherwise, there is a sequence $I\in \{I_0,\dots, I_{i-1}, I^1_i,I_{i+1},\dots, I_{m-1}\}$ so that $E(I)\subseteq E(I^0_i)$.
Thus, if $(sup, sig)$ is a $\tau$-region that solves $\alpha$ then there is a $\ell \in \{1,\dots, j-1\}$ such that $sig(y^i_{\ell})=\set$.
Consequently, there has to be a $\ell' \in \{\ell+1,\dots, j\}$ such that $sig(y^i_{\ell'})=\inp$ and, in particular, $sig(y^i_{\ell''})=\nop$ for all $\ell''\in \{\ell'+1,\dots,j\}$.
Using this, one finds that $(sup, sig)$ implies a region $(sup',sig')$ that solves $\alpha$ and gets along with at most two $\inp$-events.
More exactly, defining $sup'(\iota)=1$, $sig'(k)=sig'(y^i_{\ell'})=\inp$, $sig'(e)=\nop$ for all $e\in \{y^i_{\ell'+1}, \dots, y^i_{j}\}$ and $sig'(e)=\set$ for all $e\in E(A)\setminus (\{k,y^i_{\ell'},\dots, y^i_j\})$ yields a valid $\tau$-region $(sup', sig')$ that solves $\alpha$.
Since $(sup, sig)$ was arbitrary, these deliberations show that in the second case the atom $\alpha$ is $\tau$-solvable if and only if there is a corresponding region $(sup',sig')$.
This yields the following polynomial procedure that decides whether $\alpha$ is $\tau$-solvable:
For $\ell$ from $j$ to $2$ test if $(sup_\ell, sig_\ell)$ (inductively) defined by $sup_\ell(\iota)=1$, $sig_\ell(y^i_\ell)=\inp$, $sig_\ell(y^i_{\ell'})=\nop$ for all $\ell'\in \{\ell+1,\dots, j\}$ and $sig_\ell(e)=\set$ for all $e\in E(A)\setminus (\{k,y^i_{\ell},\dots, y^i_j\})$ yields a $\tau$-region of $A$. 
If the test succeeds for any iteration then return \textsf{yes}, otherwise return \textsf{no}.

We can modify this approach to test the $\tau$-solvability of an SSP atom $\beta=(s,s')$ as follows.
Since $A=\iota\edge{e_1}\dots \edge{e_m}\iota$ is a cycle we can assume without loss of generality that $s=\iota$ and $s'=s_i$ for some $i\in \{1,\dots, m-1\}$.
The states $\iota$ and $s_j$ partition $A$ into two subsequences $I_0=\iota\edge{e_1}\dots\edge{e_i}s_i$ and $I_1=s_i\edge{e_{i+1}}\dots \edge{e_m}\iota$.
If $\beta$ is a solvable by a region $(sup', sig')$ such that $sup'(\iota)=1$ and $sup'(s_i)=0$ then there is an event $e\in I_0$ such that $sig(e)=\inp$.
In particular, there is a region $(sup, sig)$ as follows:
$sup(\iota)=1$, $sig(e_j)=\inp$ and $j\in \{1,\dots, i\}$, $sig(e_{\ell})=\nop$ for all $\ell\in \{j+1,\dots, i\}$ and $sig(e)=\set$ for all $e\in E(A)\setminus \{e_j,\dots, e_i\}$.
Similar to the approach for $\alpha$, we can check if such a region exists in polynomial time.
Moreover, the case where $sup(\iota)=0$ and $sup(s_j)=1$ works symmetrically.

So far we have shown that the $\tau$-solvability of (E)SSP atoms of $A$ are decidable in polynomial time if $A$ is a cycle.
If $A=\iota \edge{e_1}\dots\edge{e_m}s_m$ is a directed path then its \emph{cycle extension} $A_c$ has a fresh event $\oplus\not\in E(A)$ and is defined by $A_c=\iota\edge{e_1}\dots\edge{e_m}s_m\edge{\oplus}\iota$.
The event $\oplus$ is unique thus an (E)SSP atom of $A$ is solvable by a $\tau$-region of $A$ if and only if it is by a $\tau$-region of $A_c$.
Thus, we can decide the solvability of atoms of $A$ via $A_c$.
Altogether, this proves that the $\tau$-solvability of (E)SSP atoms of $1$-bounded inputs is decidable in polynomial time.
Since we have at most $\vert S\vert^2+\vert E\vert\cdot\vert S\vert$ atoms to solve, the decidability of the $\{\nop,\inp,\set\}$-solvability for $1$-bounded TS is polynomial.

Similar to the discussion for $\tau=\{\nop,\inp,\set\}$, one argues that the following assertion is true:
If $\tau=\{\nop,\inp,\set,\used\}$ then there is a $\tau$-region $(sup', sig')$ with $sig'(k)=\used$ and $sup(s^i_j)=0$ if and only if there is a $\tau$-region $(sup, sig)$ and an number $\ell\in \{1,\dots, j\}$ such that $sup(\iota)=1$, $sig(k)=\used$, $sig(y^i_\ell)=\inp$, $sig(y^i_{\ell'})=\nop$ for all $\ell'\in \{\ell+1,\dots,j\}$ and $sig(e)=\set$ for all $e\in E(A)\setminus\{k,y^i_\ell,\dots, y^i_j\}$.
Clearly, the procedure introduced for $\{\nop,\inp,\set\}$ can be extended appropriately to a procedure that works for $\{\nop,\inp,\set,\used\}$.

It remains to investigate the case where $\tau=\{\nop,\res,\set\}\cup\omega$ with non-empty $\omega\subseteq\{\inp,\out,\used,\free\}$.
For a start, let's argue that deciding the $\tau$-solvability is equivalent to deciding the $\{\nop,\inp,\res,\set\}$-solvability or the $\{\nop,\res,\set,\used\}$-solvability of $A$.
This can be seen as follows:
If $(sup, sig)$ is a region that solves an ESSP atom $\alpha=(k,s)$ such that $sig(k)=\inp$ then there is a $\{\nop,\inp,\res,\set\}$-region $(sup', sig')$ that solves $(k,s)$, too.
The region $(sup',sig')$ origins from $(sup,sig)$ by $sup'=sup$, $sig'(k)=\inp$ and for all $e\in E(A)\setminus\{k\}$ by $sig'(e)=\nop$ if $sig(e)\in \{\nop,\used,\free\}$, $sig'(e)=\res$ if $sig(e)\in \{\inp,\res\}$ and, finally, $sig'(e)=\set$ if $sig(e)\in \{\out,\set\}$.
Similarly, one argues that $\alpha$ is $\tau$-solvable such that $sig(k)=\out$ if and only if it is $\{\nop,\out,\res,\set\}$-solvable.
Moreover, $\{\nop,\inp,\res,\set\}$ and $\{\nop,\out,\res,\set\}$ are isomorphic thus $\tau$-solvability with $\inp$ or $\out$ reduces to $\{\nop,\inp,\res,\set\}$-solvability.
Similarly, the $\tau$-solvability with $\used$ or $\free$ reduces to $\{\nop,\res,\set, \used\}$-solvability.
It is easy to see that the procedure introduced for $\{\nop,\inp,\set\}$ can be extended to the types $\{\nop,\inp,\res,\set\}$ and $\{\nop,\res,\set, \used\}$.
The only difference is that we now search for an event $y^i_\ell$ such that $sig(y^i_\ell)=\res$ instead of $sig(y^i_\ell)=\inp$.

Finally, we observe that a SSP atom $\beta=(s,s')$ is $\tau$-solvable if and only if it is $\{\nop,\res,\set\}$-solvable.
The states $s$ an $s'$ induce again a partition $I_0$ and $I_1$ of $A$ and we can adapt the approach above to $\{\nop,\res,\set\}$.
\end{proof}

\begin{theorem}\label{trivial}
For any fixed $g\in \mathbb{N}$, deciding whether a $g$-bounded TS $A$ is $\tau$-solvable is polynomial if one of the following conditions is true:
\begin{enumerate}
\item
$\tau=\{\nop,\inp,\free\}$ or $\tau=\{\nop,\inp,\used, \free\}$ or $\tau=\{\nop,\out,\used\}$ or $\tau=\{\nop,\out,\used, \free\}$ and $g < 2$.
\item
$\tau=\{\nop, \set, \res \} \cup \omega$ and $\emptyset \not=\omega \subseteq \{\used, \free\}$ and $g < 3$.
\item
$\tau=\tau'\cup\omega$ and $\tau'\in \{\{\nop,\set,\swap\}, \{\nop,\res,\swap\}, \{\nop,\res,\set, \swap\} \}$ and $\emptyset\not=\omega\subseteq\{\used,\free\}$ and $g < 2$.
\item
$\tau\in\{ \{\nop, \inp\}, \{\nop,\inp,\used\}, \{\nop,\out\}, \{\nop,\out,\free\} \}$ or $\tau\in \mathcal{T}=\{\{\nop,\set,\swap\}, \{\nop,\res,\swap\}, \{\nop,\set,\res\}, \{\nop,\set,\res,\swap\}\}$,
\end{enumerate}
\end{theorem}
\begin{proof}
(1):
It is easy to see that $A$ is a loop, $A\cong s\edge{e}s$ or that $A$ is cycle free, since there is an unsolvable SSP atom otherwise.
Moreover, if an event $e$ occurs twice consecutively, $s\edge{e}s'\edge{e}s''$, then $(s,s')$ is not $\tau$-solvable.
Thus, for every $e\in E(A)$ there is a $s\in S(A)$ such that $(e,s)$ has to be solved by $sig(e)=\inp$ ($sig(e)=\out$) and $sup(s)=0$ ($sup(s)=1$).
If $e$ occurs twice on the directed path $A$ then such a region does not exist.
On the other hand, $A$ is $\tau$-solvable if every event occurs exactly once.
Consequently, $A$ is $\tau$-solvable if and only if it is $1$-bounded and every event occurs exactly once.

(2):
Since ESSP atoms of $\tau$-solvable inputs $A$ are only solvable by $\used$ and $\free$, we have that if $s\edge{e}s'\in A$ then $s'\edge{e}s''\in A$.
If $s=s''\not=s'$ or if $s,s',s''$ are pairwise distinct then $(s,s')$ is not $\tau$-solvable.
This implies $s'\edge{e}s'$.
As a result, $\tau$-solvable inputs have the shape

\begin{center}
\begin{tikzpicture}[new set = import nodes]
\begin{scope}[nodes={set=import nodes}]%
		
		\node (init) at (-0.5,0) {$A=$};
		\node (0) at (0,0) {\nscale{$\iota$}};
		\node (1) at (1.5,0) {\nscale{$s_1$}};
		\node (dots) at (2.25,0) {\nscale{$\dots$}};
		\node (2) at (3,0) {\nscale{$s_{m-1}$}};
		\node (3) at (4.5,0) {\nscale{$s_m$}};
		
		\path (1) edge [->, out=130,in=50,looseness=3] node[above] {\nscale{$e_1$} } (1);	
		\path (3) edge [->, out=130,in=50,looseness=3] node[above] {\nscale{$e_m$} } (3);
		\graph {
	(import nodes);
			0 ->["\escale{$e_0$}"]1;
			2->["\escale{$e_m$}"]3;
			};
\end{scope}
\end{tikzpicture}
\end{center}
Thus, if the \emph{loop erasement} $A'$ of $A$ origins from $A$ by erasing all loops $s\edge{e}s$, that is, $A'=\iota\edge{e_1}\dots\edge{e_m}s_m$, then deciding the $\tau$-solvability of $A$ reduces to deciding if $A'$ has the $\tau$-SSP and if all ESSP atoms $(e,s)$ with $\neg(\edge{e}s)$ of $A'$ are $\tau$-solvable.
This is doable in polynomial time by the approach of Theorem~\ref{the:nop_inp_set_1bounded}. 

(3):
Since ESSP atoms of an input $A$ are only solvable by $\used$ and $\free$, if $s\edge{e}s'$ and $s\not=s'$ then $s'\edge{e}$.
If $s\edge{e}s'\edge{e}s''\edge{e}s'''\in A$ and $s,s',s'', s'''$ are pairwise different, then the SSP atom $(s',s''')$ is not $\tau$-solvable.
As a consequence, $\tau$-solvable inputs can have at most 3 different states.

(4):
Let $\tau\in \{\{\nop, \inp\}, \{\nop,\inp,\used\}\}$.
If there is an event $e\in E(A)$ not occurring at $A$'s initial state $\iota$ then $(\iota, e)$ is not $\tau$-solvable. 
Thus, if $A$ is $\tau$-solvable then every event $e$ of $A$ occurs at $\iota$.
Similarly, if $\tau\in\mathcal{T}$, then ESSP atoms are not $\tau$-solvable thus, every event occurs at $\iota$.
Since $A$ is $g$-bounded, it is $\vert E(A)\vert \leq g$.
Thus, $A$ has at most $2\cdot \vert \tau \vert^g$ $\tau$-regions and, since $g$ is fixed, deciding the $\tau$-solvability is polynomial by brut-force.
By Lemma~\ref{lem:isomorphy}, this  proves the claim.
\end{proof}

\section{Conclusion}%
In this paper, we fully characterize the computational complexity of \nop-equipped Boolean Petri nets from $g$-bounded transition systems, where the bound $g\in \mathbb{N}$ is chosen in advance.
Our results show that if $\tau$-synthesis is hard then it remains hard even for low bounds $g$.
Moreover, they also show that when g becomes very small, sometimes it makes the difference between hardness and tractability,  cf. Figure~\ref{fig:summary} \S 1 - \S 3 and \S 9, but sometimes it does not, cf. Figure~\ref{fig:summary} \S 4 - \S 7.
In this sense, the parameter $g$ helps to recognize interactions that contribute to the power of a type.
By Theorem~\ref{the:nop_inp_set_2bounded} and Theorem~\ref{the:nop_inp_set_1bounded}, $\{\nop, \inp,\set\}$-synthesis is hard if $g\geq 2$ and tractable if $g < 2$, respectively.
By Theorem~\ref{the:nop_inp_set_free+used}, $\{\nop,\inp,\set,\free\}$-synthesis remains hard for all $g \geq 1$.
Thus, if restricted to $1$-bounded inputs then the test interaction $\free$ makes the difference between hardness and tractability of synthesis.
Surprisingly enough, by Theorem~\ref{the:nop_inp_set_1bounded}, replacing $\free$ by $\used$ makes synthesis from $1$-bounded TS tractable again.
It remains future work, to characterize the computational complexity of synthesis for the remaining 128 types which do not contain \nop.
Moreover, since $\tau$-synthesis generally remains hard even for (small) fixed $g$, the bound of a TS is ruled out for FPT-algorithms.
Future work might be concerned with parameterizing $\tau$-synthesis by the \emph{dependence number} of the searched $\tau$-net:
If $N=(P, T, f, M_0)$ is a Boolean net, $p\in P$ and if the \emph{dependence number} $d_p$ of $p$ is defined by $d_p=\vert \{ t\in T\mid f(p,t)\not=\nop\} \vert $ then the \emph{dependence number} $d$ of $N$ is defined by $d=\text{max}\{ d_p \mid p\in P\}$.
At first glance, $d$ appears to be a promising parameter for FPT-approaches because this parameterization puts the problem into the complexity class XP:
Since a $\tau$-region of $A=(S,E,\delta, \iota)$ is determined by $sup(\iota)$ and $sig$, for each (E)SSP atom $\alpha$ there are at most $2\cdot\vert\tau\vert^d\cdot\sum_{i=0}^d \binom{\vert E\vert}{i}$ fitting $\tau$-regions solving $\alpha$.
Thus, by $\vert \tau\vert \leq 8$, $\tau$-synthesis parameterized by $d$ is decidable in $\mathcal{O}(\vert E\vert^d \cdot \vert S\vert \cdot \text{max}\{\vert S\vert ,\vert E\vert\} )$.

\bibliography{myBibliographyTopnoc}%

\begin{appendix}
\section{The $\tau$-Solvability of $\alpha$ Implies the $\tau$-Solvability of $A^\tau_\varphi$}

In this section, we continue the proofs of Theorem~\ref{the:nop_inp_free} - Theorem~\ref{the:nop_inp_set_swap+out_res_used_free} and show that the solvability of $\alpha$ always implies the (E)SSP of $A^\tau_\varphi$.
To do so, we have to argue that all SSP atoms $(s,s')\in S(A^\tau_\varphi)^2$ and all relevant ESSP atoms $S(A^\tau_\varphi)\times E(A^\tau_\varphi)$ are solvable.
However, enumerating all involved regions and presenting them explicitly would be tedious and unreasonable.
Fortunately, due to the uniqueness of the connector events $\oplus$, $\ominus$ and states $\bot$, the symmetry of the gadgets and the fact that several events occur only once per gadget, this is not necessary.
In the following, we use, among others, the following techniques to prove the solvability of (E)SSP atoms.
First, we present regions $(sup, sig)$ implicitly by defining $sup(\iota)$ and $sig$.
This actually determines $(sup, sig)$ completely, to be seen as follows.
Let $s\in S(A^\tau_\varphi)$.
Since $A^\tau_\varphi$ is reachable, there is a path $\iota=s_0\edge{e_1}s_1\dots s_{n-1}\edge{e_n}s_n=s$ from $\iota$ to $s$ in $A^\tau_\varphi$.
Thus, inductively defining $sup(s_{i})=\delta_\tau(sup(s_{i-1},sig(e_{i}))$ for all $i\in \{1,\dots, n\}$ yields $sup(s)$.
By the arbitrariness of $s$, this shows that $sup(\iota)$ and $sig$ define $(sup, sig)$ completely.
Second, we often restrict the presentation of a solving region to the \enquote{interesting} gadgets of $A^\tau_\varphi$.
That is, if we apply this technique, then it is easy to see that the sketched part of the region can be extended from the gadgets shown to $A^\tau_\varphi$.
Third, graphical representations of the corresponding parts of $A^\tau_\varphi$ easily show that the sketched regions are valid.
In particular, the presented figures have colored areas.
The colored areas refer always to the states that are mapped to $1$, the other states of the presented gadgets are assumed to be mapped to $0$.
This allows to easily recognize that $s\edge{e}s'$ implies $sup(s)\edge{sig(e)}sup(s')$, in particular, if $sig(e)$ execute a state change in $\tau$.

\subsection{Continuation of the Proof of Theorem~\ref{the:nop_inp_free}}
The following table implicitly defines regions $(sup, sig)$ via $sup(\bot_0)$ and $sig$.
In particular, $sup(\bot_0)=1$ for all $\tau$-regions presented.
Moreover, for all $e\in E(A^\tau_\varphi)$ we define $sig(e)=\inp$ if $e$ is listed in the $\inp$-column, $sig(e)=\free$ if it is listed in the $\free$-column and else $sig(e)=\nop$.
\begin{flushleft}
\begin{longtable}{c|p{5cm}|p{5cm}}
$R$ & $\inp$ & $\free$ \\ \hline
$R_{i,0}$  & $\ominus_{i+1}$ & $\ominus_{i+2},\dots, \ominus_{m}$, $\oplus_{i+1},\dots,\oplus_m$, $k_0,k_1$   \\
$R_{i,1}$  & $\oplus_i, \ominus_{i+1} $ & \\
$R_\bot$ & $\oplus_0, \dots, \oplus_m$ & $E(A^\tau_\varphi)\setminus \{k_0,k_1, \}$ \\
$R_{k_0,i}$  & $k_0,X_{i_0}$ & \\
$R_{k_1,i}$  & $k_0,X_{i_3}$, $\{\oplus_j\in \oplus \mid X_{i_3}\not\in \zeta_j\}$ & \\
$R_{X_i,0}$ & $X_{i}, \ominus_m$ & \\
$R_{X_i,1}$  & $X_i, k_0$ & \\
$R_{X_i,2}$ & $X_i$, $\{\oplus_j\in \oplus \mid X_i\not\in \zeta_j\}$ & \\
\end{longtable}
\end{flushleft}
In the following, if $e\in \{z_0,\dots, z_{16m-1}\}$ then we say $e$ is a $z$-event.
Let $i\in \{0,\dots, m-1\} $ be arbitrary but fixed.

($\oplus,\ominus$):
Let $S_{i,0}=\{\bot_0,\dots, \bot_i\}\cup \bigcup_{j=0}^i S(T_j)$ and $E_{i,0}=\{ \ominus_{i+2},\dots, \ominus_{m}\}\cup \{\oplus_{i+1},\dots,\oplus_m, k_0,k_1 \}$.
The region $R_{i,0}$ solves all atoms $(e,s)\in \{\ominus_{i+1}\}\times  (S(A^\tau_\varphi)\setminus S_{i,0})$ and $(e,s)\in E_{i,0} \times S_{i,0}$.
The region $R_{i,1}$ solves all atoms $(e,s) \in \{\ominus_{i+1}\} \times S(T_i)$ and $(e,s) \in \{\oplus_i\}\times  (S(T_i)\cup S_{i,0})$.
The region $R=(sup, sig)$ defined by $sup(\bot_0)=1$, $sig(\oplus_m)=\inp$, $sig(k_0)=sig(k_1)=\free$ and $sig(e)=\nop$ if $e\in E(A^\tau_\varphi)\setminus\{\oplus_m,k_0,k_1\}$ solves (all relevant) $(e,s)\in \{ \oplus_m, k_0, k_1\} \times \{\bot_m,m_0,m_1\}$.
Thus, if $e\in \oplus\cup\ominus$, then $e$ is solvable.

($k_0, k_1$):
The region $R_\bot$ solves all $(e,s)\in \{k_0,k_1\}\times \{\bot_0,\dots, \bot_m\}$.
The region $R_{k_0,i}$ solves all $(e,s)\in \{k_0\}\times (S(T_i) \cup \{m_1,m_2\})$.
This proves the solvability of $k_0$.
The region $R_{k_1,i}$ solves all $(e,s)\in \{k_1\}\times \{t_{i,0},t_{i,1}, t_{i,2}, t_{i,4}, m_2\}$ and the $(sup, sig)$ with $sup(\bot_0)=1$ and $sig(k_1)=\inp$ solves $k_1$ at all $\bigcup_{j=0}^{m-1}\{t_{j,4}\}\cup\{m_2\}$.
This proves the solvability of $k_1$.

($X_i$):
The region $R_{X_i,0}$ solves all $(e,s)\in \{X_i\}\times \{\bot_6,m_0,m_1,m_2\}$ and the region $R_{X_i,1}$ solves all $(e,s) \in \{X_i\}\times (\{t_{j,\ell+1},\dots, t_{j,5}\mid j\in \{0,\dots,m-1\}, t_{j,\ell}\edge{X_i} \})$.
Moreover, the region $R_\bot$ solves $X_i$ at the states of $\bot$.
Finally, $R_{X_i,2}$ solves $X_i$ at the states of $T_j$ where $X_i\not\in \zeta_j$.

Since $i$ was arbitrary, this proves the $\tau$-ESSP.
One easily verifies that the presented regions also prove the $\tau$-SSP.

\subsection{Continuation of the Proof of Theorem~\ref{lem:nop_set_res+used_free}}
Let $(k_0,h_{0,2})$ be $\tau$-solvable.
For a start, we show that $A^\tau_\varphi$ has $\tau$-ESSP.
Since $A^\tau_\varphi$ is $\{\nop,\res,\set, \used\}$-solvable if and only if it is $\{\nop,\res,\set, \free\}$-solvable, in the following we only show that $A^\tau_\varphi$ is $\{\nop,\res,\set, \used\}$-solvable.

($e\in \oplus$):
Let $e\in \oplus$ and $s\edge{e}s'\in A^\tau_\varphi$.
We define $sup(s)=sup(s')=1$ and $sup(s'')=0$ for all $s''\in S(A^\tau_\varphi)\setminus \{s,s'\}$ and $sig(e)=\used$.
For all $e'\in E(A^\tau_\varphi)\setminus\{e\} $ and $s\edge{e'}s'\in A^\tau_\varphi$ we define $sig(e')=\set$ if $sup(s)=0$ and $sup(s')=1$, $sig(e')=\res$ if $sup(s)=1$ and $sup(s')=0$ and otherwise $sig(e')=\nop$.
Clearly, this yields a $\tau$-region $(sup,sig)$ that solves $(e,q)$ for all relevant $q\in S(A^\tau_\varphi)$.

($c_{2j}, c_{2j+1}$):
To solve $(c_{2j}, f_{j,0})$ and $(c_{2j}, f_{j,0})$, for all $s\in S(A^\tau_\varphi)$ we define $sup(s)=1$ if $s\in \{f_{j,1}, f_{j,2}, f_{j,4}\}\cup \{s'\in S(T_{i,\ell})\mid z_j\in E(T_{i,\ell})\}$ and else $sup(s)=0$.
Moreover, we let $sig(c_{2j})=\used$, $sig(c_{2j+1})=\res$ and $sig(z_j)=\set$.
It is easy to see that $sig$ can be defined appropriately on all remaining events $e\in E(A^\tau_\varphi)\setminus \{c_{2j},c_{2j+1}, z_j\}$ to get a proper region $(sup, sig)$. 

To solve $(c_{2j}, f_{j,4})$, for all $s\in S(A^\tau_\varphi)$  we define $sup(s)=1$ if $s\in \{f_{j,0}, f_{j,1}, f_{j,2}\}$ and else $sup(s)=0$.
Moreover, we let $sig(c_{2j})=\used$, $sig(c_{2j+1})=\res$ and choose the signature of the other events appropriately to get a proper region $(sup, sig)$.
Finally, to solve $(c_{2_j}, s)$ for all $s\in S(A^\tau_\varphi)\setminus S(F_j)$ we let $sup(s)=1$ if $s\in S(F_j)$ and else $sup(s)=0$ and let $sig(c_{2j})=\used$.
Again, we can define the signature $sig$ on the remaining events to get a proper $(sup, sig)$.

Solving $(c_{2j+1}, s)$ for all $s\in S(A^\tau_\varphi)\setminus \{f_{j,2}, f_{j,3}\}$ is even simpler: 
We define $sup(f_{j,2})=sup(f_{j,3})=1$ and $sup(s)=0$ for all $s\in S(A^\tau_\varphi)\setminus \{f_{j,2}, f_{j,3}\}$ and $sig(c_{2j+1})=\used$, $sig(c_{2j})=\set$, $sig(z_j)=\res$ and $sig(e)=\nop$ for all remaining events $e$.

$(z_j)$:
Let $(i,\ell)\in \{0,\dots, m-1\}\times \{0,\dots,3\}$ such that $z_j\in E(T_{i,\ell})$.
To $\tau$-solve all atoms $(z_j, s)$ for all relevant $s\in (S(A^\tau_\varphi)\setminus S(T_{i,\ell}))$, we define $sup=(S(F_j)\setminus \{f_{j,2}\})\cup S(T_{i,\ell})$ and $sig(z_j)=\used$, $sig(c_{2j})=\res$ and $sig(c_{2j+1})$.
Clearly, $sig$ can be defined appropriately on the remaining events such that $(sup, sig)$ becomes a region that solves $(z_{j}, s)$.

The red area of the following figure sketches how to solve $(z_{16i+9}, s)$ where $s\in \{t_{i,2,0}, t_{i,2,1}\}$ and the blue area sketches this for all $s\in \{t_{i,2,4},\dots, t_{i,2,7}\}$.
The signature of $z_{16i+9}$ is for both cases $\used$.
The states of the colored area are mapped to $1$.
In both cases, it is easy to see that we can extend the sketched part to a fitting region $(sup, sig)$ of $A^\tau_\varphi$.
Moreover, this approach works for all $z$-events $z_j$ and the corresponding states of $T_{i,\ell}$ where $z_j \in E(T_{i,\ell})$.
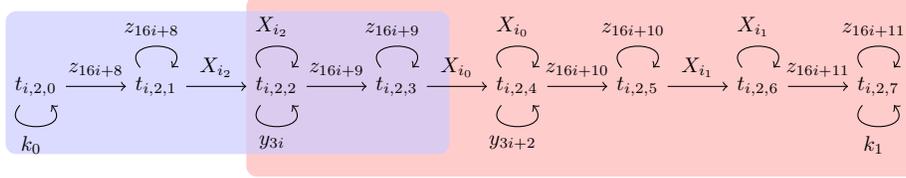
\begin{flushleft}
\begin{tikzpicture}[new set = import nodes]
\begin{scope}[yshift=-5cm,nodes={set=import nodes}]
	\foreach \i in {0,...,7} {\coordinate (\i) at (\i*1.6cm,0);}
	\foreach \i in {2} {\fill[red!20, rounded corners] (\i) +(-0.4,-1.2) rectangle +(8.5,1.2);}
	\foreach \i in {0} {\fill[blue!20, rounded corners, opacity=0.7] (\i) +(-0.4,-0.9) rectangle +(5.5,1);}
	\foreach \i in {0,...,7} {\node (t\i) at (\i) {\nscale{$t_{i,2,\i}$}};}
	
\path (t0) edge [->, out=-130,in=-50,looseness=3] node[below, align =left] {\nscale{$k_0$} } (t0);	
\path (t1) edge [->, out=130,in=50,looseness=3] node[above, align =left] {\nscale{$z_{16i+8}$} } (t1);

\path (t2) edge [->, out=130,in=50,looseness=3] node[above, align =left] {\nscale{$X_{i_2}$} } (t2);
\path (t2) edge [->, out=-130,in=-50,looseness=3] node[below, align =left] {\nscale{$y_{3i}$} } (t2);

\path (t3) edge [->, out=130,in=50,looseness=3] node[above, align =left] {\nscale{$z_{16i+9}$} } (t3);

\path (t4) edge [->, out=130,in=50,looseness=3] node[above, align =left] {\nscale{$X_{i_0}$} } (t4);
\path (t4) edge [->, out=-130,in=-50,looseness=3] node[below, align =left] {\nscale{$y_{3i+2}$} } (t4);

\path (t5) edge [->, out=130,in=50,looseness=3] node[above, align =left] {\nscale{$z_{16i+10}$} } (t5);
\path (t6) edge [->, out=130,in=50,looseness=3] node[above, align =left] {\nscale{$X_{i_1}$} } (t6);

\path (t7) edge [->, out=130,in=50,looseness=3] node[above, align =left] {\nscale{$z_{16i+11} $} } (t7);
\path (t7) edge [->, out=-130,in=-50,looseness=3] node[below, align =left] {\nscale{$k_1$} } (t7);
\graph {
	(import nodes);
			t0 ->["\escale{$z_{16i+8}$}"]t1;
			t1 ->["\escale{$X_{i_2}$}"]t2;
			t2 ->["\escale{$z_{16i+9}$}"]t3;
			t3 ->["\escale{$X_{i_0}$}"]t4;
			t4 ->["\escale{$z_{16i+10}$}"]t5;
			t5 ->["\escale{$X_{i_1}$}"]t6;
			t6 ->["\escale{$z_{16i+11}$}"]t7;
			};
\end{scope}
\end{tikzpicture}
\end{flushleft}
$(X_j, y_\ell)$:
Let $j\in \{0,\dots, m-1\}$, $\ell\in \{0,\dots, 3m-1\}$ be arbitrary but fixed.
Let $i_0,i_1,i_2\in \{0,\dots, m-1\}$ be pairwise distinct such that $X_j\in \zeta_{i_0}\cap\zeta_{i_1}\cap\zeta_{i_2} $.
To solve $X_{j}$ simultaneously at the states of the gadgets in which $X_j$ does not appear, we define $sup=\bigcup_{i\in\{i_0,i_1,i_2 \}}(S(T_{i,0})\cup S(T_{i,1})\cup S(T_{i,2}))$ and $sig(X_j)=\used$ and extend $sig$ appropriately to get a fitting region $(sup, sig)$.
The red colored area of the following figure sketches the situation for $X_{i_1}$.
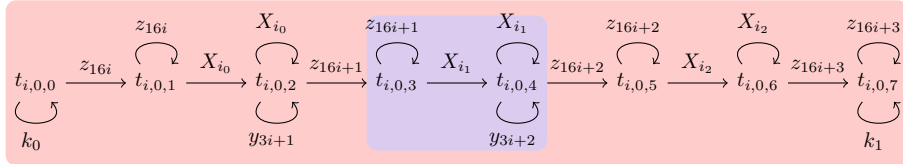
\begin{flushleft}
\begin{tikzpicture}[new set = import nodes]
\begin{scope}[nodes={set=import nodes}]
	\foreach \i in {0,...,7} {\coordinate (\i) at (\i*1.6cm,0);}
	\foreach \i in {0} {\fill[red!20, rounded corners] (\i) +(-0.4,-1.1) rectangle +(11.6,1.1);}
	\foreach \i in {3} {\fill[blue!20, rounded corners, opacity=0.7] (\i) +(-0.4,-0.9) rectangle +(2,0.9);}
	\foreach \i in {0,...,7} {\node (t\i) at (\i) {\nscale{$t_{i,0,\i}$}};}
	
\path (t0) edge [->, out=-130,in=-50,looseness=3] node[below, align =left] {\nscale{$k_0$} } (t0);	
\path (t1) edge [->, out=130,in=50,looseness=3] node[above, align =left] {\nscale{$z_{16i}$} } (t1);

\path (t2) edge [->, out=130,in=50,looseness=3] node[above, align =left] {\nscale{$X_{i_0}$} } (t2);
\path (t2) edge [->, out=-130,in=-50,looseness=3] node[below, align =left] {\nscale{$y_{3i+1}$} } (t2);

\path (t3) edge [->, out=130,in=50,looseness=3] node[above, align =left] {\nscale{$z_{16i+1}$} } (t3);

\path (t4) edge [->, out=130,in=50,looseness=3] node[above, align =left] {\nscale{$X_{i_1}$} } (t4);
\path (t4) edge [->, out=-130,in=-50,looseness=3] node[below, align =left] {\nscale{$y_{3i+2}$} } (t4);

\path (t5) edge [->, out=130,in=50,looseness=3] node[above, align =left] {\nscale{$z_{16i+2}$} } (t5);
\path (t6) edge [->, out=130,in=50,looseness=3] node[above, align =left] {\nscale{$X_{i_2}$} } (t6);

\path (t7) edge [->, out=130,in=50,looseness=3] node[above, align =left] {\nscale{$z_{16i+3} $} } (t7);
\path (t7) edge [->, out=-130,in=-50,looseness=3] node[below, align =left] {\nscale{$k_1$} } (t7);
\graph {
	(import nodes);
			t0 ->["\escale{$z_{16i}$}"]t1;
			t1 ->["\escale{$X_{i_0}$}"]t2;
			t2 ->["\escale{$z_{16i+1}$}"]t3;
			t3 ->["\escale{$X_{i_1}$}"]t4;
			t4 ->["\escale{$z_{16i+2}$}"]t5;
			t5 ->["\escale{$X_{i_2}$}"]t6;
			t6 ->["\escale{$z_{16i+3}$}"]t7;
			};
\end{scope}
\end{tikzpicture}
\end{flushleft}
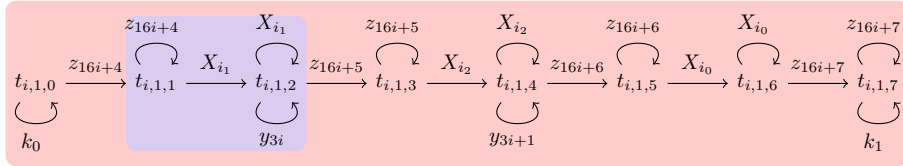
\begin{flushleft}
\begin{tikzpicture}[new set = import nodes]
\begin{scope}[yshift=-2.25cm,nodes={set=import nodes}]
	\foreach \i in {0,...,7} {\coordinate (\i) at (\i*1.6cm,0);}
	\foreach \i in {0} {\fill[red!20, rounded corners] (\i) +(-0.4,-1.1) rectangle +(11.6,1.1);}
	\foreach \i in {1} {\fill[blue!20, rounded corners, opacity=0.7] (\i) +(-0.4,-0.9) rectangle +(2,0.9);}	
	\foreach \i in {0,...,7} {\node (t\i) at (\i) {\nscale{$t_{i,1,\i}$}};}
	
\path (t0) edge [->, out=-130,in=-50,looseness=3] node[below, align =left] {\nscale{$k_0$} } (t0);	
\path (t1) edge [->, out=130,in=50,looseness=3] node[above, align =left] {\nscale{$z_{16i+4}$} } (t1);

\path (t2) edge [->, out=130,in=50,looseness=3] node[above, align =left] {\nscale{$X_{i_1}$} } (t2);
\path (t2) edge [->, out=-130,in=-50,looseness=3] node[below, align =left] {\nscale{$y_{3i}$} } (t2);

\path (t3) edge [->, out=130,in=50,looseness=3] node[above, align =left] {\nscale{$z_{16i+5}$} } (t3);

\path (t4) edge [->, out=130,in=50,looseness=3] node[above, align =left] {\nscale{$X_{i_2}$} } (t4);
\path (t4) edge [->, out=-130,in=-50,looseness=3] node[below, align =left] {\nscale{$y_{3i+1}$} } (t4);

\path (t5) edge [->, out=130,in=50,looseness=3] node[above, align =left] {\nscale{$z_{16i+6}$} } (t5);
\path (t6) edge [->, out=130,in=50,looseness=3] node[above, align =left] {\nscale{$X_{i_0}$} } (t6);
\path (t7) edge [->, out=130,in=50,looseness=3] node[above, align =left] {\nscale{$z_{16i+7} $} } (t7);
\path (t7) edge [->, out=-130,in=-50,looseness=3] node[below, align =left] {\nscale{$k_1$} } (t7);
\graph {
	(import nodes);
			t0 ->["\escale{$z_{16i+4}$}"]t1;
			t1 ->["\escale{$X_{i_1}$}"]t2;
			t2 ->["\escale{$z_{16i+5}$}"]t3;
			t3 ->["\escale{$X_{i_2}$}"]t4;
			t4 ->["\escale{$z_{16i+6}$}"]t5;
			t5 ->["\escale{$X_{i_0}$}"]t6;
			t6 ->["\escale{$z_{16i+7}$}"]t7;
			};
\end{scope}
\end{tikzpicture}
\end{flushleft}
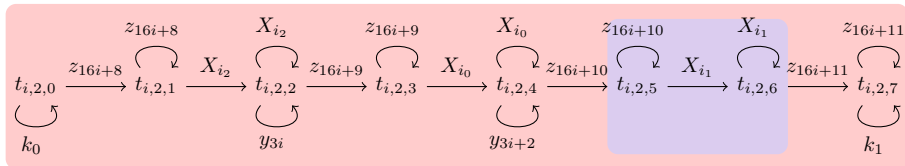
\begin{flushleft}
\begin{tikzpicture}[new set = import nodes]
\begin{scope}[yshift=-5cm,nodes={set=import nodes}]
	\foreach \i in {0,...,7} {\coordinate (\i) at (\i*1.6cm,0);}
	\foreach \i in {0} {\fill[red!20, rounded corners] (\i) +(-0.4,-1.1) rectangle +(11.6,1.1);}
	\foreach \i in {5} {\fill[blue!20, rounded corners, opacity=0.7] (\i) +(-0.4,-0.9) rectangle +(2,0.9);}
	\foreach \i in {0,...,7} {\node (t\i) at (\i) {\nscale{$t_{i,2,\i}$}};}
	
\path (t0) edge [->, out=-130,in=-50,looseness=3] node[below, align =left] {\nscale{$k_0$} } (t0);	
\path (t1) edge [->, out=130,in=50,looseness=3] node[above, align =left] {\nscale{$z_{16i+8}$} } (t1);

\path (t2) edge [->, out=130,in=50,looseness=3] node[above, align =left] {\nscale{$X_{i_2}$} } (t2);
\path (t2) edge [->, out=-130,in=-50,looseness=3] node[below, align =left] {\nscale{$y_{3i}$} } (t2);

\path (t3) edge [->, out=130,in=50,looseness=3] node[above, align =left] {\nscale{$z_{16i+9}$} } (t3);

\path (t4) edge [->, out=130,in=50,looseness=3] node[above, align =left] {\nscale{$X_{i_0}$} } (t4);
\path (t4) edge [->, out=-130,in=-50,looseness=3] node[below, align =left] {\nscale{$y_{3i+2}$} } (t4);

\path (t5) edge [->, out=130,in=50,looseness=3] node[above, align =left] {\nscale{$z_{16i+10}$} } (t5);
\path (t6) edge [->, out=130,in=50,looseness=3] node[above, align =left] {\nscale{$X_{i_1}$} } (t6);

\path (t7) edge [->, out=130,in=50,looseness=3] node[above, align =left] {\nscale{$z_{16i+11} $} } (t7);
\path (t7) edge [->, out=-130,in=-50,looseness=3] node[below, align =left] {\nscale{$k_1$} } (t7);
\graph {
	(import nodes);
			t0 ->["\escale{$z_{16i+8}$}"]t1;
			t1 ->["\escale{$X_{i_2}$}"]t2;
			t2 ->["\escale{$z_{16i+9}$}"]t3;
			t3 ->["\escale{$X_{i_0}$}"]t4;
			t4 ->["\escale{$z_{16i+10}$}"]t5;
			t5 ->["\escale{$X_{i_1}$}"]t6;
			t6 ->["\escale{$z_{16i+11}$}"]t7;
			};
\end{scope}
\end{tikzpicture}
\end{flushleft}
To solve $X_j$ at the remaining states, we define $sup=\{s,s'\in S(A^\tau_\varphi)\mid s\edge{X_j}s'\}\cup \{s\in S(F_j)\mid \exists s'\in S(A^\tau_\varphi): \Edge{z_j}s'\Edge{X_j}\}$, $sig(X_j)=\used$ and, again, extend $sig$ appropriately.
The blue colored area above sketches the situation for $X_{i_1}$ in $T_{i,0}, T_{i,1}, T_{i,2}$.
The event $y_\ell$ can be solved quite similar.

$(k_0, k_1, k_2, k_3)$:
The atom $(k_0, m_{0,2})$ is solved by the region that solves $\alpha$. 
It is easy to see, that all remaining relevant atoms $(k_0, s)$ are solvable, too.
The solution of $(k_1, m_{0,0})$ works symmetrically to the solution of $(k_0, m_{0,2})$, in particular, it requires a one-and-three  model of $\varphi$.
Again, it is easy to see that the remaining atoms $(k_2,s)$ are solvable.
The same is true for all atoms $(e,s)\in \{k_2, k_3\}\times S(A^\tau_\varphi)$ in question.
Finally, one easily finds out that $A^\tau_\varphi$ has the $\tau$-SSP.

\subsection{Continuation of the Proof of Theorem~\ref{the:nop_inp_set_2bounded}}
The proof that the $\tau$-solvability of $\alpha$ implies $A^\tau_\varphi$'s $\tau$-(E)SSP can be found in~\cite{DBLP:journals/corr/abs-1806-03703}.

\subsection{Continuation of the Proof of Theorem~\ref{the:nop_inp_out_set+used_free}}

$(k_0)$:
The following figure sketches the solvability of $(k_0,s)$ for all $s\in \{h_{0,3}, h_{0,2}\}\cup \bigcup_{i=0}^{m-1}\{b_{i,0}\}$.
The presented part of the signature $sig$ is defined by $sig(k_0)=\inp$, $sig(e)=\set$ for all $e\in \{k_1,z_1\}\cup V(\varphi)$ and $sig(z_0)=sig(o)=\nop$.
Surely, $sig$ can be properly extended to the remaining events.
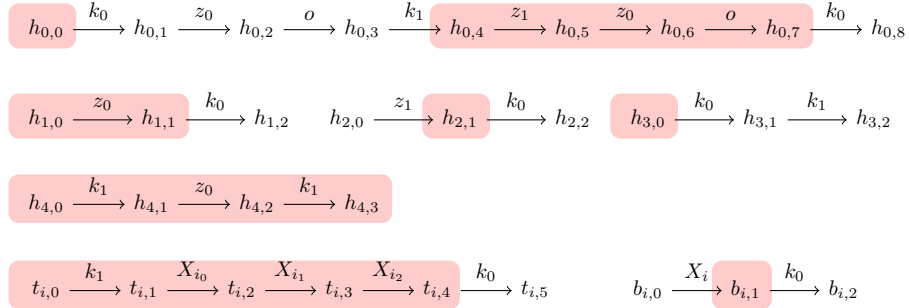
\begin{flushleft}
\begin{tikzpicture}[new set = import nodes]
\begin{scope}[nodes={set=import nodes}]%
		\foreach \i in {0,...,8} { \coordinate (\i) at (\i*1.4cm,0) ;}
		\foreach \i in {0} {\fill[red!20, rounded corners] (\i) +(-0.5,-0.25) rectangle +(0.4,0.35);}
		\foreach \i in {4} {\fill[red!20, rounded corners] (\i) +(-0.5,-0.25) rectangle +(4.6,0.35);}
		\foreach \i in {0,...,8} { \node (\i) at (\i) {\nscale{$h_{0,\i}$}};}
\graph {
	(import nodes);
			0 ->["\escale{$k_0$}"]1->["\escale{$z_0$}"]2 ->["\escale{$o$}"]3 ->["\escale{$k_1$}"]4  ->["\escale{$z_1$}"]5->["\escale{$z_0$}"]6->["\escale{$o$}"]7->["\escale{$k_0$}"]8;
};
\end{scope}
\begin{scope}[yshift=-1.2cm,nodes={set=import nodes}]%
		\foreach \i in {0,...,2} { \coordinate (\i) at (\i*1.5cm,0) ;}
		\foreach \i in {0} {\fill[red!20, rounded corners] (\i) +(-0.5,-0.25) rectangle +(1.9,0.35);}
		\foreach \i in {0,...,2} { \node (\i) at (\i) {\nscale{$h_{1,\i}$}};}
\graph {
	(import nodes);
			0 ->["\escale{$z_0$}"]1->["\escale{$k_0$}"]2;
			};
\end{scope}
\begin{scope}[xshift=4cm, yshift=-1.2cm,nodes={set=import nodes}]%
		\foreach \i in {0,...,2} { \coordinate (\i) at (\i*1.5cm,0) ;}
		\foreach \i in {1} {\fill[red!20, rounded corners] (\i) +(-0.5,-0.25) rectangle +(0.4,0.35);}
		\foreach \i in {0,...,2} { \node (\i) at (\i) {\nscale{$h_{2,\i}$}};}
\graph {
	(import nodes);
			0 ->["\escale{$z_1$}"]1->["\escale{$k_0$}"]2;
			};
\end{scope}
\begin{scope}[xshift=8cm,yshift=-1.2cm,nodes={set=import nodes}]%
		\foreach \i in {0,...,2} { \coordinate (\i) at (\i*1.5cm,0) ;}
		\foreach \i in {0} {\fill[red!20, rounded corners] (\i) +(-0.5,-0.25) rectangle +(0.4,0.35);}
		\foreach \i in {0,...,2} { \node (\i) at (\i) {\nscale{$h_{3,\i}$}};}
\graph {
	(import nodes);
			0 ->["\escale{$k_0$}"]1->["\escale{$k_1$}"]2;
			};
\end{scope}
\end{tikzpicture}
\end{flushleft}
\begin{flushleft}
\begin{tikzpicture}[new set = import nodes]
\begin{scope}[nodes={set=import nodes}]%
		\foreach \i in {0,...,3} { \coordinate (\i) at (\i*1.4cm,0) ;}
		\foreach \i in {0} {\fill[red!20, rounded corners] (\i) +(-0.5,-0.25) rectangle +(4.6,0.4);}
		\foreach \i in {0,...,3} { \node (\i) at (\i) {\nscale{$h_{4,\i}$}};}
\graph {
	(import nodes);
			0 ->["\escale{$k_1$}"]1->["\escale{$z_0$}"]2 ->["\escale{$k_1$}"]3;
			};
\end{scope}
\end{tikzpicture}
\end{flushleft}
\begin{flushleft}
\begin{tikzpicture}[new set = import nodes]
\begin{scope}[nodes={set=import nodes}]%
		\foreach \i in {0,...,5} { \coordinate (\i) at (\i*1.3cm,0) ;}
		\foreach \i in {0} {\fill[red!20, rounded corners] (\i) +(-0.5,-0.25) rectangle +(5.5,0.4);}
		\foreach \i in {0,...,5} { \node (\i) at (\i) {\nscale{$t_{i,\i}$}};}
\graph {
	(import nodes);
			0 ->["\escale{$k_1$}"]1->["\escale{$X_{i_0}$}"]2 ->["\escale{$X_{i_1}$}"]3->["\escale{$X_{i_2}$}"]4->["\escale{$k_0$}"]5; 
			};
\end{scope}
\begin{scope}[xshift=8cm,nodes={set=import nodes}]%
		\foreach \i in {0,...,2} { \coordinate (\i) at (\i*1.3cm,0) ;}
		\foreach \i in {1} {\fill[red!20, rounded corners] (\i) +(-0.45,-0.25) rectangle +(0.35,0.4);}
		\foreach \i in {0,...,2} { \node (\i) at (\i) {\nscale{$b_{i,\i}$}};}
\graph {
	(import nodes);
			0 ->["\escale{$X_i$}"]1->["\escale{$k_0$}"]2;  
			};
\end{scope}
\end{tikzpicture}
\end{flushleft}
Let $i\in \{0,\dots, m-1\}$ be arbitrary but fixed.
The following figure sketches the solvability of $(k_0,s)$ for all $s\in \{h_{1,0}, t_{i,2}, t_{i,3}, t_{i,5}\}$, where $sig(k_0)=\inp$.
It is easy to see that the $sig$ is extendable to the other events.
The remaining atoms with $k_0$ are solved by the region that solves $\alpha$.
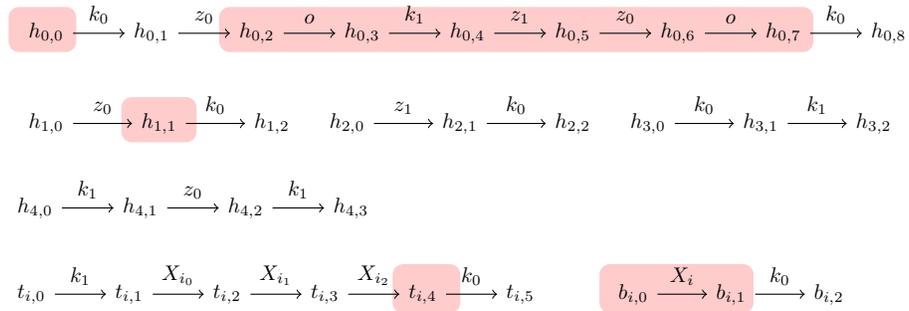
\begin{flushleft}
\begin{tikzpicture}[new set = import nodes]
\begin{scope}[nodes={set=import nodes}]%
		\foreach \i in {0,...,8} { \coordinate (\i) at (\i*1.4cm,0) ;}
		\foreach \i in {0} {\fill[red!20, rounded corners] (\i) +(-0.5,-0.25) rectangle +(0.4,0.35);}
		\foreach \i in {2} {\fill[red!20, rounded corners] (\i) +(-0.5,-0.25) rectangle +(7.4,0.35);}
		\foreach \i in {0,...,8} { \node (\i) at (\i) {\nscale{$h_{0,\i}$}};}
\graph {
	(import nodes);
			0 ->["\escale{$k_0$}"]1->["\escale{$z_0$}"]2 ->["\escale{$o$}"]3 ->["\escale{$k_1$}"]4  ->["\escale{$z_1$}"]5->["\escale{$z_0$}"]6->["\escale{$o$}"]7->["\escale{$k_0$}"]8;
};
\end{scope}
\begin{scope}[yshift=-1.2cm,nodes={set=import nodes}]%
		\foreach \i in {0,...,2} { \coordinate (\i) at (\i*1.5cm,0) ;}
		\foreach \i in {1} {\fill[red!20, rounded corners] (\i) +(-0.5,-0.25) rectangle +(0.5,0.35);}
		\foreach \i in {0,...,2} { \node (\i) at (\i) {\nscale{$h_{1,\i}$}};}
\graph {
	(import nodes);
			0 ->["\escale{$z_0$}"]1->["\escale{$k_0$}"]2;
			};
\end{scope}
\begin{scope}[xshift=4cm, yshift=-1.2cm,nodes={set=import nodes}]%
		\foreach \i in {0,...,2} { \coordinate (\i) at (\i*1.5cm,0) ;}
		\foreach \i in {0,...,2} { \node (\i) at (\i) {\nscale{$h_{2,\i}$}};}
\graph {
	(import nodes);
			0 ->["\escale{$z_1$}"]1->["\escale{$k_0$}"]2;
			};
\end{scope}
\begin{scope}[xshift=8cm,yshift=-1.2cm,nodes={set=import nodes}]%
		\foreach \i in {0,...,2} { \coordinate (\i) at (\i*1.5cm,0) ;}
		\foreach \i in {0,...,2} { \node (\i) at (\i) {\nscale{$h_{3,\i}$}};}
\graph {
	(import nodes);
			0 ->["\escale{$k_0$}"]1->["\escale{$k_1$}"]2;
			};
\end{scope}
\end{tikzpicture}
\end{flushleft}
\begin{flushleft}
\begin{tikzpicture}[new set = import nodes]
\begin{scope}[nodes={set=import nodes}]%
		\foreach \i in {0,...,3} { \coordinate (\i) at (\i*1.4cm,0) ;}
		\foreach \i in {0,...,3} { \node (\i) at (\i) {\nscale{$h_{4,\i}$}};}
\graph {
	(import nodes);
			0 ->["\escale{$k_1$}"]1->["\escale{$z_0$}"]2 ->["\escale{$k_1$}"]3;
			};
\end{scope}
\end{tikzpicture}
\end{flushleft}
\begin{flushleft}
\begin{tikzpicture}[new set = import nodes]
\begin{scope}[nodes={set=import nodes}]%
		\foreach \i in {0,...,5} { \coordinate (\i) at (\i*1.3cm,0) ;}
		\foreach \i in {4} {\fill[red!20, rounded corners] (\i) +(-0.4,-0.25) rectangle +(0.5,0.4);}
		\foreach \i in {0,...,5} { \node (\i) at (\i) {\nscale{$t_{i,\i}$}};}
\graph {
	(import nodes);
			0 ->["\escale{$k_1$}"]1->["\escale{$X_{i_0}$}"]2 ->["\escale{$X_{i_1}$}"]3->["\escale{$X_{i_2}$}"]4->["\escale{$k_0$}"]5; 
			};
\end{scope}
\begin{scope}[xshift=8cm,nodes={set=import nodes}]%
		\foreach \i in {0,...,2} { \coordinate (\i) at (\i*1.3cm,0) ;}
		\foreach \i in {0} {\fill[red!20, rounded corners] (\i) +(-0.45,-0.25) rectangle +(1.6,0.4);}
		\foreach \i in {0,...,2} { \node (\i) at (\i) {\nscale{$b_{i,\i}$}};}
\graph {
	(import nodes);
			0 ->["\escale{$X_i$}"]1->["\escale{$k_0$}"]2;  
			};
\end{scope}
\end{tikzpicture}
\end{flushleft}

$(o)$:
The red colored area of the following figure sketches the solvability of $o$ for all $s\in S(A^\tau_\varphi)\setminus \{h_{1,1},h_{1,2}, h_{4,2}, h_{4,3} \}$ and the blue one does this for the remaining states.
In both cases, the signature of $o$ is $sig(o)=\inp$ and the signature of the remaining events can be chosen appropriately.
\begin{flushleft}
\begin{tikzpicture}[new set = import nodes]
\begin{scope}[nodes={set=import nodes}]%
		\foreach \i in {0,...,8} { \coordinate (\i) at (\i*1.4cm,0) ;}
		\foreach \i in {2,6} {\fill[red!20, rounded corners] (\i) +(-0.5,-0.5) rectangle +(0.4,0.5);}
		\foreach \i in {0} {\fill[blue!20, rounded corners, opacity=0.7] (\i) +(-0.5,-0.3) rectangle +(3.2,0.3);}
		\foreach \i in {5} {\fill[blue!20, rounded corners, opacity=0.7] (\i) +(-0.5,-0.3) rectangle +(1.8,0.3);}
		\foreach \i in {0,...,8} { \node (\i) at (\i) {\nscale{$h_{0,\i}$}};}
\graph {
	(import nodes);
			0 ->["\escale{$k_0$}"]1->["\escale{$z_0$}"]2 ->["\escale{$o$}"]3 ->["\escale{$k_1$}"]4  ->["\escale{$z_1$}"]5->["\escale{$z_0$}"]6->["\escale{$o$}"]7->["\escale{$k_0$}"]8;
};
\end{scope}
\begin{scope}[yshift=-1.2cm,nodes={set=import nodes}]%
		\foreach \i in {0,...,2} { \coordinate (\i) at (\i*1.5cm,0) ;}
		\foreach \i in {1} {\fill[red!20, rounded corners] (\i) +(-0.5,-0.25) rectangle +(1.8,0.35);}
		\foreach \i in {0,...,2} { \node (\i) at (\i) {\nscale{$h_{1,\i}$}};}
\graph {
	(import nodes);
			0 ->["\escale{$z_0$}"]1->["\escale{$k_0$}"]2;
			};
\end{scope}
\begin{scope}[xshift=4cm, yshift=-1.2cm,nodes={set=import nodes}]%
		\foreach \i in {0,...,2} { \coordinate (\i) at (\i*1.5cm,0) ;}
		\foreach \i in {1} {\fill[blue!20, rounded corners, opacity=0.7] (\i) +(-0.5,-0.25) rectangle +(1.8,0.35);}
		\foreach \i in {0,...,2} { \node (\i) at (\i) {\nscale{$h_{2,\i}$}};}
\graph {
	(import nodes);
			0 ->["\escale{$z_1$}"]1->["\escale{$k_0$}"]2;
			};
\end{scope}
\begin{scope}[xshift=8cm,yshift=-1.2cm,nodes={set=import nodes}]%
		\foreach \i in {0,...,2} { \coordinate (\i) at (\i*1.5cm,0) ;}
		\foreach \i in {0,...,2} { \node (\i) at (\i) {\nscale{$h_{3,\i}$}};}
\graph {
	(import nodes);
			0 ->["\escale{$k_0$}"]1->["\escale{$k_1$}"]2;
			};
\end{scope}
\end{tikzpicture}
\end{flushleft}
\begin{flushleft}
\begin{tikzpicture}[new set = import nodes]
\begin{scope}[nodes={set=import nodes}]%
		\foreach \i in {0,...,3} { \coordinate (\i) at (\i*1.4cm,0) ;}
		\foreach \i in {2} {\fill[red!20, rounded corners] (\i) +(-0.5,-0.25) rectangle +(1.7,0.4);}
		\foreach \i in {0,...,3} { \node (\i) at (\i) {\nscale{$h_{4,\i}$}};}
\graph {
	(import nodes);
			0 ->["\escale{$k_1$}"]1->["\escale{$z_0$}"]2 ->["\escale{$k_1$}"]3;
			};
\end{scope}
\end{tikzpicture}
\end{flushleft}

$(z_0)$:
The red colored area of the following figure sketches the solution of $(z_0,s)$ for all $s\in S(H_0)\setminus \{h_{0,8}\}$ and the blue colored area does this for $(z_0,h_{0,8})$ and $(z_0,h_{4,0})$.
For both regions, $sig(z_0)=\inp$.
\begin{flushleft}
\begin{tikzpicture}[new set = import nodes]
\begin{scope}[nodes={set=import nodes}]%
		\foreach \i in {0,...,8} { \coordinate (\i) at (\i*1.4cm,0) ;}
		\foreach \i in {1,5,8} {\fill[red!20, rounded corners] (\i) +(-0.5,-0.5) rectangle +(0.4,0.5);}
		\foreach \i in {0,4} {\fill[blue!20, rounded corners, opacity=0.7,opacity=0.5] (\i) +(-0.5,-0.3) rectangle +(1.8,0.3);}
		\foreach \i in {0,...,8} { \node (\i) at (\i) {\nscale{$h_{0,\i}$}};}
\graph {
	(import nodes);
			0 ->["\escale{$k_0$}"]1->["\escale{$z_0$}"]2 ->["\escale{$o$}"]3 ->["\escale{$k_1$}"]4  ->["\escale{$z_1$}"]5->["\escale{$z_0$}"]6->["\escale{$o$}"]7->["\escale{$k_0$}"]8;
};
\end{scope}
\begin{scope}[yshift=-1.2cm,nodes={set=import nodes}]%
		\foreach \i in {0,...,2} { \coordinate (\i) at (\i*1.5cm,0) ;}
		\foreach \i in {0,2} {\fill[red!20, rounded corners] (\i) +(-0.5,-0.5) rectangle +(0.4,0.5);}
		\foreach \i in {0} {\fill[blue!20, rounded corners, opacity=0.7] (\i) +(-0.5,-0.3) rectangle +(0.4,0.33);}
		\foreach \i in {0,...,2} { \node (\i) at (\i) {\nscale{$h_{1,\i}$}};}
\graph {
	(import nodes);
			0 ->["\escale{$z_0$}"]1->["\escale{$k_0$}"]2;
			};
\end{scope}
\begin{scope}[xshift=4cm, yshift=-1.2cm,nodes={set=import nodes}]%
		\foreach \i in {0,...,2} { \coordinate (\i) at (\i*1.5cm,0) ;}
		\foreach \i in {1} {\fill[red!20, rounded corners] (\i) +(-0.5,-0.25) rectangle +(1.8,0.35);}
		\foreach \i in {0,...,2} { \node (\i) at (\i) {\nscale{$h_{2,\i}$}};}
\graph {
	(import nodes);
			0 ->["\escale{$z_1$}"]1->["\escale{$k_0$}"]2;
			};
\end{scope}
\begin{scope}[xshift=8cm,yshift=-1.2cm,nodes={set=import nodes}]%
		\foreach \i in {0,...,2} { \coordinate (\i) at (\i*1.5cm,0) ;}
		\foreach \i in {1} {\fill[red!20, rounded corners] (\i) +(-0.5,-0.5) rectangle +(1.8,0.5);}
		\foreach \i in {2} {\fill[blue!20, rounded corners, opacity=0.7] (\i) +(-0.5,-0.3) rectangle +(0.3,0.3);}
		\foreach \i in {0,...,2} { \node (\i) at (\i) {\nscale{$h_{3,\i}$}};}
\graph {
	(import nodes);
			0 ->["\escale{$k_0$}"]1->["\escale{$k_1$}"]2;
			};
\end{scope}
\end{tikzpicture}
\end{flushleft}
\begin{flushleft}
\begin{tikzpicture}[new set = import nodes]
\begin{scope}[nodes={set=import nodes}]%
		\foreach \i in {0,...,3} { \coordinate (\i) at (\i*1.4cm,0) ;}
		\foreach \i in {0} {\fill[red!20, rounded corners] (\i) +(-0.5,-0.5) rectangle +(1.8,0.5);}
		\foreach \i in {1,3} {\fill[blue!20, rounded corners, opacity=0.7] (\i) +(-0.5,-0.3) rectangle +(0.4,0.3);}
		\foreach \i in {0,...,3} { \node (\i) at (\i) {\nscale{$h_{4,\i}$}};}
\graph {
	(import nodes);
			0 ->["\escale{$k_1$}"]1->["\escale{$z_0$}"]2 ->["\escale{$k_1$}"]3;
			};
\end{scope}
\end{tikzpicture}
\end{flushleft}
\begin{flushleft}
\begin{tikzpicture}[new set = import nodes]
\begin{scope}[nodes={set=import nodes}]%
		\foreach \i in {0,...,5} { \coordinate (\i) at (\i*1.3cm,0) ;}
		\foreach \i in {5} {\fill[red!20, rounded corners] (\i) +(-0.4,-0.5) rectangle +(0.5,0.6);}
		\foreach \i in {1} {\fill[blue!20, rounded corners, opacity=0.7] (\i) +(-0.45,-0.3) rectangle +(5.5,0.4);}
		\foreach \i in {0,...,5} { \node (\i) at (\i) {\nscale{$t_{i,\i}$}};}
\graph {
	(import nodes);
			0 ->["\escale{$k_1$}"]1->["\escale{$X_{i_0}$}"]2 ->["\escale{$X_{i_1}$}"]3->["\escale{$X_{i_2}$}"]4->["\escale{$k_0$}"]5; 
			};
\end{scope}
\begin{scope}[xshift=8cm,nodes={set=import nodes}]%
		\foreach \i in {0,...,2} { \coordinate (\i) at (\i*1.3cm,0) ;}
		\foreach \i in {2} {\fill[red!20, rounded corners] (\i) +(-0.45,-0.25) rectangle +(0.4,0.4);}
		\foreach \i in {0,...,2} { \node (\i) at (\i) {\nscale{$b_{i,\i}$}};}
\graph {
	(import nodes);
			0 ->["\escale{$X_i$}"]1->["\escale{$k_0$}"]2;  
			};
\end{scope}
\end{tikzpicture}
\end{flushleft}
The following figure sketches a region that solves $z_0$ at the remaining states, where $sig(z_0)=\inp$.
\begin{flushleft}
\begin{tikzpicture}[new set = import nodes]
\begin{scope}[nodes={set=import nodes}]%
		\foreach \i in {0,...,8} { \coordinate (\i) at (\i*1.4cm,0) ;}
		\foreach \i in {0,7} {\fill[red!20, rounded corners] (\i) +(-0.5,-0.3) rectangle +(1.8,0.4);}
		\foreach \i in {3} {\fill[red!20, rounded corners] (\i) +(-0.5,-0.3) rectangle +(3.2,0.4);}
		\foreach \i in {0,...,8} { \node (\i) at (\i) {\nscale{$h_{0,\i}$}};}
\graph {
	(import nodes);
			0 ->["\escale{$k_0$}"]1->["\escale{$z_0$}"]2 ->["\escale{$o$}"]3 ->["\escale{$k_1$}"]4  ->["\escale{$z_1$}"]5->["\escale{$z_0$}"]6->["\escale{$o$}"]7->["\escale{$k_0$}"]8;
};
\end{scope}
\begin{scope}[yshift=-1.2cm,nodes={set=import nodes}]%
		\foreach \i in {0,...,2} { \coordinate (\i) at (\i*1.5cm,0) ;}
		\foreach \i in {0} {\fill[red!20, rounded corners] (\i) +(-0.5,-0.3) rectangle +(0.4,0.3);}
		\foreach \i in {0,...,2} { \node (\i) at (\i) {\nscale{$h_{1,\i}$}};}
\graph {
	(import nodes);
			0 ->["\escale{$z_0$}"]1->["\escale{$k_0$}"]2;
			};
\end{scope}
\begin{scope}[xshift=4.5cm, yshift=-1.2cm,nodes={set=import nodes}]%
		\foreach \i in {0,...,3} { \coordinate (\i) at (\i*1.4cm,0) ;}
		\foreach \i in {0} {\fill[red!20, rounded corners] (\i) +(-0.5,-0.3) rectangle +(1.8,0.4);}
		\foreach \i in {0,...,3} { \node (\i) at (\i) {\nscale{$h_{4,\i}$}};}
\graph {
	(import nodes);
			0 ->["\escale{$k_1$}"]1->["\escale{$z_0$}"]2 ->["\escale{$k_1$}"]3;
			};
\end{scope}
\end{tikzpicture}
\end{flushleft}

$(z_1)$:
The red colored area of following figure sketches the solvability of $(z_1, s)$ for all $s\in S(H_0)\cup S(H_2)$.
The blue colored area shows the solvability of $z_1$ at the remaining states.
\begin{flushleft}
\begin{tikzpicture}[new set = import nodes]
\begin{scope}[nodes={set=import nodes}]%
		\foreach \i in {0,...,8} { \coordinate (\i) at (\i*1.4cm,0) ;}
		\foreach \i in {4} {\fill[red!20, rounded corners] (\i) +(-0.5,-0.5) rectangle +(0.4,0.5);}
		\foreach \i in {0} {\fill[blue!20, rounded corners, opacity=0.7] (\i) +(-0.5,-0.3) rectangle +(6,0.4);}
		\foreach \i in {0,...,8} { \node (\i) at (\i) {\nscale{$h_{0,\i}$}};}
\graph {
	(import nodes);
			0 ->["\escale{$k_0$}"]1->["\escale{$z_0$}"]2 ->["\escale{$o$}"]3 ->["\escale{$k_1$}"]4  ->["\escale{$z_1$}"]5->["\escale{$z_0$}"]6->["\escale{$o$}"]7->["\escale{$k_0$}"]8;
};
\end{scope}
\begin{scope}[yshift=-1.2cm,nodes={set=import nodes}]%
		\foreach \i in {0,...,2} { \coordinate (\i) at (\i*1.5cm,0) ;}
		\foreach \i in {0,...,2} { \node (\i) at (\i) {\nscale{$h_{1,\i}$}};}
\graph {
	(import nodes);
			0 ->["\escale{$z_0$}"]1->["\escale{$k_0$}"]2;
			};
\end{scope}
\begin{scope}[xshift=4cm, yshift=-1.2cm,nodes={set=import nodes}]%
		\foreach \i in {0,...,2} { \coordinate (\i) at (\i*1.5cm,0) ;}
		\foreach \i in {0} {\fill[red!20, rounded corners] (\i) +(-0.5,-0.5) rectangle +(0.4,0.5);}
		\foreach \i in {0} {\fill[blue!20, rounded corners, opacity=0.7] (\i) +(-0.5,-0.3) rectangle +(0.4,0.3);}		
		\foreach \i in {0,...,2} { \node (\i) at (\i) {\nscale{$h_{2,\i}$}};}
\graph {
	(import nodes);
			0 ->["\escale{$z_1$}"]1->["\escale{$k_0$}"]2;
			};
\end{scope}
\begin{scope}[xshift=8cm,yshift=-1.2cm,nodes={set=import nodes}]%
		\foreach \i in {0,...,2} { \coordinate (\i) at (\i*1.5cm,0) ;}
		\foreach \i in {0,...,2} { \node (\i) at (\i) {\nscale{$h_{3,\i}$}};}
\graph {
	(import nodes);
			0 ->["\escale{$k_0$}"]1->["\escale{$k_1$}"]2;
			};
\end{scope}
\end{tikzpicture}
\end{flushleft}
\begin{flushleft}
\begin{tikzpicture}[new set = import nodes]
\begin{scope}[nodes={set=import nodes}]%
		\foreach \i in {0,...,3} { \coordinate (\i) at (\i*1.4cm,0) ;}
		\foreach \i in {1} {\fill[red!20, rounded corners] (\i) +(-0.5,-0.3) rectangle +(3.2,0.4);}
		\foreach \i in {0,...,3} { \node (\i) at (\i) {\nscale{$h_{4,\i}$}};}
\graph {
	(import nodes);
			0 ->["\escale{$k_1$}"]1->["\escale{$z_0$}"]2 ->["\escale{$k_1$}"]3;
			};
\end{scope}
\end{tikzpicture}
\end{flushleft}
\begin{flushleft}
\begin{tikzpicture}[new set = import nodes]
\begin{scope}[nodes={set=import nodes}]%
		\foreach \i in {0,...,5} { \coordinate (\i) at (\i*1.3cm,0) ;}
		\foreach \i in {1} {\fill[red!20, rounded corners] (\i) +(-0.5,-0.3) rectangle +(5.5,0.4);}
		\foreach \i in {0,...,5} { \node (\i) at (\i) {\nscale{$t_{i,\i}$}};}
\graph {
	(import nodes);
			0 ->["\escale{$k_1$}"]1->["\escale{$X_{i_0}$}"]2 ->["\escale{$X_{i_1}$}"]3->["\escale{$X_{i_2}$}"]4->["\escale{$k_0$}"]5; 
			};
\end{scope}
\begin{scope}[xshift=7.5cm,nodes={set=import nodes}]%
		\foreach \i in {0,...,2} { \coordinate (\i) at (\i*1.3cm,0) ;}
		\foreach \i in {0,...,2} { \node (\i) at (\i) {\nscale{$b_{i,\i}$}};}
\graph {
	(import nodes);
			0 ->["\escale{$X_i$}"]1->["\escale{$k_0$}"]2;  
			};
\end{scope}
\end{tikzpicture}
\end{flushleft}

$(X_i)$:
Let $i\in \{0,\dots, m-1\}$ be arbitrary but fixed.
In the following figure we consider $X_{i_1}=X_i$ and sketch the solution of $(X_{i_1}, s)$ for all $s\in S(T_i)$.
For the signature holds $sig(X_{i_1})=\inp$ and $sig(X_{i_0})=\set$ and one can define the signature (support) for the other events (states) appropriately.
Moreover, it is easy to see that $X_i$ is also solvable at the remaining states.
\begin{flushleft}
\begin{tikzpicture}[new set = import nodes]
\begin{scope}[nodes={set=import nodes}]%
		\foreach \i in {0,...,5} { \coordinate (\i) at (\i*1.3cm,0) ;}
		\foreach \i in {2} {\fill[red!20, rounded corners] (\i) +(-0.4,-0.3) rectangle +(.4,0.4);}
		\foreach \i in {0,...,5} { \node (\i) at (\i) {\nscale{$t_{i,\i}$}};}
\graph {
	(import nodes);
			0 ->["\escale{$k_1$}"]1->["\escale{$X_{i_0}$}"]2 ->["\escale{$X_{i_1}$}"]3->["\escale{$X_{i_2}$}"]4->["\escale{$k_0$}"]5; 
			};
\end{scope}
\end{tikzpicture}
\end{flushleft}

It remains to argue for the solvability of $k_1$.
Unfortunately, this needs a lot of case analyses.
In particular, we have to consider the cases $\used\in \tau$ and $\used\not\in \tau$ separately.

($k_1$ and $\used\not\in \tau$):
Since $\used\not\in \tau$, $A^\tau_\varphi$ does not have $H_4$.
The red colored area of the following figure sketches the solvability of $(k_1,s)$ for all $s\in S(A^\tau_\varphi)\setminus \{h_{0,0}, h_{0,1}, h_{0,2}, h_{3,0}\}$.
The blue colored area sketches the solvability of $(k_1, s)$ for all $s\in \{h_{0,0}, h_{0,1}, h_{0,2}\}$.
In both cases, $sig(k_1)=\inp$.
It is easy to see that $(k_1,h_{3,0})$ is solvable.
This shows that $k_1$ is solvable if $\used\not\in \tau$.
\begin{flushleft}
\begin{tikzpicture}[new set = import nodes]
\begin{scope}[nodes={set=import nodes}]%
		\foreach \i in {0,...,8} { \coordinate (\i) at (\i*1.4cm,0) ;}
		\foreach \i in {0} {\fill[red!20, rounded corners] (\i) +(-0.5,-0.5) rectangle +(4.7,0.5);}
		\foreach \i in {3} {\fill[blue!20, rounded corners, opacity=0.7] (\i) +(-0.5,-0.3) rectangle +(0.4,0.4);}
		\foreach \i in {7} {\fill[blue!20, rounded corners, opacity=0.7] (\i) +(-0.5,-0.3) rectangle +(1.7,0.4);}
		\foreach \i in {0,...,8} { \node (\i) at (\i) {\nscale{$h_{0,\i}$}};}
\graph {
	(import nodes);
			0 ->["\escale{$k_0$}"]1->["\escale{$z_0$}"]2 ->["\escale{$o$}"]3 ->["\escale{$k_1$}"]4  ->["\escale{$z_1$}"]5->["\escale{$z_0$}"]6->["\escale{$o$}"]7->["\escale{$k_0$}"]8;
};
\end{scope}
\begin{scope}[yshift=-1.2cm,nodes={set=import nodes}]%
		\foreach \i in {0,...,2} { \coordinate (\i) at (\i*1.5cm,0) ;}
		\foreach \i in {0,...,2} { \node (\i) at (\i) {\nscale{$h_{1,\i}$}};}
\graph {
	(import nodes);
			0 ->["\escale{$z_0$}"]1->["\escale{$k_0$}"]2;
			};
\end{scope}
\begin{scope}[xshift=4cm, yshift=-1.2cm,nodes={set=import nodes}]%
		\foreach \i in {0,...,2} { \coordinate (\i) at (\i*1.5cm,0) ;}
		\foreach \i in {0,...,2} { \node (\i) at (\i) {\nscale{$h_{2,\i}$}};}
\graph {
	(import nodes);
			0 ->["\escale{$z_1$}"]1->["\escale{$k_0$}"]2;
			};
\end{scope}
\begin{scope}[xshift=8cm,yshift=-1.2cm,nodes={set=import nodes}]%
		\foreach \i in {0,...,2} { \coordinate (\i) at (\i*1.5cm,0) ;}
		\foreach \i in {0} {\fill[red!20, rounded corners] (\i) +(-0.5,-0.3) rectangle +(1.9,0.4);}
		\foreach \i in {0,...,2} { \node (\i) at (\i) {\nscale{$h_{3,\i}$}};}
\graph {
	(import nodes);
			0 ->["\escale{$k_0$}"]1->["\escale{$k_1$}"]2;
			};
\end{scope}
\end{tikzpicture}
\end{flushleft}
\begin{flushleft}
\begin{tikzpicture}[new set = import nodes]
\begin{scope}[nodes={set=import nodes}]%
		\foreach \i in {0,...,5} { \coordinate (\i) at (\i*1.3cm,0) ;}
		\foreach \i in {0} {\fill[red!20, rounded corners] (\i) +(-0.5,-0.3) rectangle +(0.4,0.4);}
		\foreach \i in {0,...,5} { \node (\i) at (\i) {\nscale{$t_{i,\i}$}};}
\graph {
	(import nodes);
			0 ->["\escale{$k_1$}"]1->["\escale{$X_{i_0}$}"]2 ->["\escale{$X_{i_1}$}"]3->["\escale{$X_{i_2}$}"]4->["\escale{$k_0$}"]5; 
			};
\end{scope}
\begin{scope}[xshift=7.5cm,nodes={set=import nodes}]%
		\foreach \i in {0,...,2} { \coordinate (\i) at (\i*1.3cm,0) ;}		
		\foreach \i in {0,...,2} { \node (\i) at (\i) {\nscale{$b_{i,\i}$}};}
\graph {
	(import nodes);
			0 ->["\escale{$X_i$}"]1->["\escale{$k_0$}"]2;  
			};
\end{scope}
\end{tikzpicture}
\end{flushleft}
($k_1$ and $\used \in \tau$):
The red colored area of the following figure sketches the solution of $k_1$ at all states of $A^\tau_\varphi$ but $h_{0,2}, h_{0,6},\dots, h_{0,8}$, $h_{1,1}, h_{1,2}$ and $h_{3,0}$.
For this region holds $sig(k_1)=\inp$.
The blue area sketches the solution of $(k_1,s)$ for all $s\in \{h_{0,2}, h_{1,1}, h_{1,2}\}$.
For this region holds $sig(k_1)=\used$.
\begin{flushleft}
\begin{tikzpicture}[new set = import nodes]
\begin{scope}[nodes={set=import nodes}]%
		\foreach \i in {0,...,8} { \coordinate (\i) at (\i*1.4cm,0) ;}
		\foreach \i in {2} {\fill[red!20, rounded corners] (\i) +(-0.5,-0.5) rectangle +(1.9,0.5);}
		\foreach \i in {6} {\fill[red!20, rounded corners] (\i) +(-0.5,-0.5) rectangle +(3.1,0.5);}
		\foreach \i in {3} {\fill[blue!20, rounded corners, opacity=0.7] (\i) +(-0.5,-0.3) rectangle +(7.3,0.4);}
		\foreach \i in {0,...,8} { \node (\i) at (\i) {\nscale{$h_{0,\i}$}};}
\graph {
	(import nodes);
			0 ->["\escale{$k_0$}"]1->["\escale{$z_0$}"]2 ->["\escale{$o$}"]3 ->["\escale{$k_1$}"]4  ->["\escale{$z_1$}"]5->["\escale{$z_0$}"]6->["\escale{$o$}"]7->["\escale{$k_0$}"]8;
};
\end{scope}
\begin{scope}[yshift=-1.2cm,nodes={set=import nodes}]%
		\foreach \i in {0,...,2} { \coordinate (\i) at (\i*1.5cm,0) ;}
		\foreach \i in {1} {\fill[red!20, rounded corners] (\i) +(-0.5,-0.5) rectangle +(1.9,0.5);}
		\foreach \i in {0,...,2} { \node (\i) at (\i) {\nscale{$h_{1,\i}$}};}
\graph {
	(import nodes);
			0 ->["\escale{$z_0$}"]1->["\escale{$k_0$}"]2;
			};
\end{scope}
\begin{scope}[xshift=4cm, yshift=-1.2cm,nodes={set=import nodes}]%
		\foreach \i in {0,...,2} { \coordinate (\i) at (\i*1.5cm,0) ;}
		\foreach \i in {0,...,2} { \node (\i) at (\i) {\nscale{$h_{2,\i}$}};}
\graph {
	(import nodes);
			0 ->["\escale{$z_1$}"]1->["\escale{$k_0$}"]2;
			};
\end{scope}
\begin{scope}[xshift=8cm,yshift=-1.2cm,nodes={set=import nodes}]%
		\foreach \i in {0,...,2} { \coordinate (\i) at (\i*1.5cm,0) ;}
		\foreach \i in {0} {\fill[red!20, rounded corners] (\i) +(-0.4,-0.5) rectangle +(1.9,0.5);}
		\foreach \i in {0} {\fill[blue!20, rounded corners, opacity=0.7] (\i) +(-0.4,-0.3) rectangle +(3.3,0.4);}
		\foreach \i in {0,...,2} { \node (\i) at (\i) {\nscale{$h_{3,\i}$}};}
\graph {
	(import nodes);
			0 ->["\escale{$k_0$}"]1->["\escale{$k_1$}"]2;
			};
\end{scope}
\end{tikzpicture}
\end{flushleft}
\begin{flushleft}
\begin{tikzpicture}[new set = import nodes]
\begin{scope}[nodes={set=import nodes}]%
		\foreach \i in {0,...,3} { \coordinate (\i) at (\i*1.4cm,0) ;}
		\foreach \i in {0,2} {\fill[red!20, rounded corners] (\i) +(-0.4,-0.5) rectangle +(0.5,0.5);}
		\foreach \i in {0} {\fill[blue!20, rounded corners, opacity=0.7] (\i) +(-0.4,-0.3) rectangle +(4.5,0.4);}
		\foreach \i in {0,...,3} { \node (\i) at (\i) {\nscale{$h_{4,\i}$}};}
\graph {
	(import nodes);
			0 ->["\escale{$k_1$}"]1->["\escale{$z_0$}"]2 ->["\escale{$k_1$}"]3;
			};
\end{scope}
\end{tikzpicture}
\end{flushleft}
\begin{flushleft}
\begin{tikzpicture}[new set = import nodes]
\begin{scope}[nodes={set=import nodes}]%
		\foreach \i in {0,...,5} { \coordinate (\i) at (\i*1.3cm,0) ;}
		\foreach \i in {0} {\fill[red!20, rounded corners] (\i) +(-0.4,-0.5) rectangle +(1.7,0.5);}
		\foreach \i in {0} {\fill[blue!20, rounded corners, opacity=0.7] (\i) +(-0.4,-0.3) rectangle +(6.8,0.4);}
		\foreach \i in {0,...,5} { \node (\i) at (\i) {\nscale{$t_{i,\i}$}};}
\graph {
	(import nodes);
			0 ->["\escale{$k_1$}"]1->["\escale{$X_{i_0}$}"]2 ->["\escale{$X_{i_1}$}"]3->["\escale{$X_{i_2}$}"]4->["\escale{$k_0$}"]5; 
			};
\end{scope}
\begin{scope}[xshift=7.5cm,nodes={set=import nodes}]%
		\foreach \i in {0,...,2} { \coordinate (\i) at (\i*1.3cm,0) ;}
		\foreach \i in {0,...,2} { \node (\i) at (\i) {\nscale{$b_{i,\i}$}};}
\graph {
	(import nodes);
			0 ->["\escale{$X_i$}"]1->["\escale{$k_0$}"]2;  
			};
\end{scope}
\end{tikzpicture}
\end{flushleft}
The red colored area of the following figure sketches the solution of $(k_1,s)$ for all $s\in \{h_{0,6}, h_{0,7}, h_{0,8}\}$.
Moreover, the blue area together with the states of the red colored area but $h_{3,0}$ sketches the solvability of $(k_0, h_{3,0})$.
For both regions, $sig(k_1)=\used$.
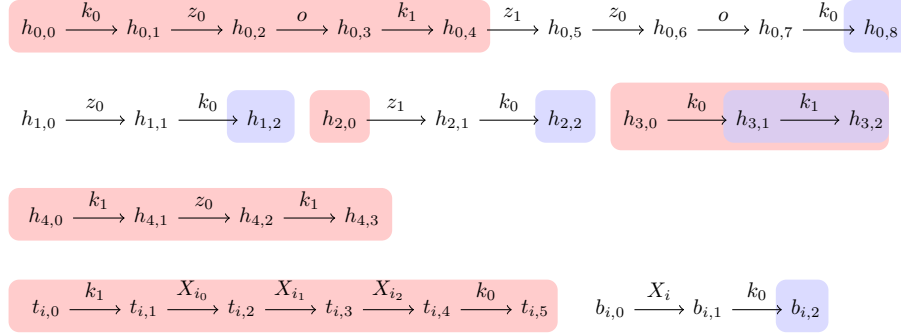
\begin{flushleft}
\begin{tikzpicture}[new set = import nodes]
\begin{scope}[nodes={set=import nodes}]%
		\foreach \i in {0,...,8} { \coordinate (\i) at (\i*1.4cm,0) ;}
		\foreach \i in {0} {\fill[red!20, rounded corners] (\i) +(-0.4,-0.3) rectangle +(6,0.4);}
		\foreach \i in {8} {\fill[blue!20, rounded corners, opacity=0.7] (\i) +(-0.5,-0.3) rectangle +(0.4,0.4);}
		\foreach \i in {0,...,8} { \node (\i) at (\i) {\nscale{$h_{0,\i}$}};}
\graph {
	(import nodes);
			0 ->["\escale{$k_0$}"]1->["\escale{$z_0$}"]2 ->["\escale{$o$}"]3 ->["\escale{$k_1$}"]4  ->["\escale{$z_1$}"]5->["\escale{$z_0$}"]6->["\escale{$o$}"]7->["\escale{$k_0$}"]8;
};
\end{scope}
\begin{scope}[yshift=-1.2cm,nodes={set=import nodes}]%
		\foreach \i in {0,...,2} { \coordinate (\i) at (\i*1.5cm,0) ;}
		\foreach \i in {2} {\fill[blue!20, rounded corners, opacity=0.7] (\i) +(-0.5,-0.3) rectangle +(0.4,0.4);}
		\foreach \i in {0,...,2} { \node (\i) at (\i) {\nscale{$h_{1,\i}$}};}
\graph {
	(import nodes);
			0 ->["\escale{$z_0$}"]1->["\escale{$k_0$}"]2;
			};
\end{scope}
\begin{scope}[xshift=4cm, yshift=-1.2cm,nodes={set=import nodes}]%
		\foreach \i in {0,...,2} { \coordinate (\i) at (\i*1.5cm,0) ;}
		\foreach \i in {0} {\fill[red!20, rounded corners] (\i) +(-0.4,-0.3) rectangle +(0.4,0.4);}
		\foreach \i in {2} {\fill[blue!20, rounded corners, opacity=0.7] (\i) +(-0.4,-0.3) rectangle +(0.4,0.4);}
		\foreach \i in {0,...,2} { \node (\i) at (\i) {\nscale{$h_{2,\i}$}};}
\graph {
	(import nodes);
			0 ->["\escale{$z_1$}"]1->["\escale{$k_0$}"]2;
			};
\end{scope}
\begin{scope}[xshift=8cm,yshift=-1.2cm,nodes={set=import nodes}]%
		\foreach \i in {0,...,2} { \coordinate (\i) at (\i*1.5cm,0) ;}
		\foreach \i in {0} {\fill[red!20, rounded corners] (\i) +(-0.4,-0.4) rectangle +(3.3,0.5);}
		\foreach \i in {1} {\fill[blue!20, rounded corners, opacity=0.7] (\i) +(-0.4,-0.3) rectangle +(1.8,0.4);}
		\foreach \i in {0,...,2} { \node (\i) at (\i) {\nscale{$h_{3,\i}$}};}
\graph {
	(import nodes);
			0 ->["\escale{$k_0$}"]1->["\escale{$k_1$}"]2;
			};
\end{scope}
\end{tikzpicture}
\end{flushleft}
\begin{flushleft}
\begin{tikzpicture}[new set = import nodes]
\begin{scope}[nodes={set=import nodes}]%
		\foreach \i in {0,...,3} { \coordinate (\i) at (\i*1.4cm,0) ;}
		\foreach \i in {0} {\fill[red!20, rounded corners] (\i) +(-0.5,-0.3) rectangle +(4.6,0.4);}
		\foreach \i in {0,...,3} { \node (\i) at (\i) {\nscale{$h_{4,\i}$}};}
\graph {
	(import nodes);
			0 ->["\escale{$k_1$}"]1->["\escale{$z_0$}"]2 ->["\escale{$k_1$}"]3;
			};
\end{scope}
\end{tikzpicture}
\end{flushleft}
\begin{flushleft}
\begin{tikzpicture}[new set = import nodes]
\begin{scope}[nodes={set=import nodes}]%
		\foreach \i in {0,...,5} { \coordinate (\i) at (\i*1.3cm,0) ;}
		\foreach \i in {0} {\fill[red!20, rounded corners] (\i) +(-0.5,-0.3) rectangle +(6.8,0.4);}
		\foreach \i in {0,...,5} { \node (\i) at (\i) {\nscale{$t_{i,\i}$}};}
\graph {
	(import nodes);
			0 ->["\escale{$k_1$}"]1->["\escale{$X_{i_0}$}"]2 ->["\escale{$X_{i_1}$}"]3->["\escale{$X_{i_2}$}"]4->["\escale{$k_0$}"]5; 
			};
\end{scope}
\begin{scope}[xshift=7.5cm,nodes={set=import nodes}]%
		\foreach \i in {0,...,2} { \coordinate (\i) at (\i*1.3cm,0) ;}
		\foreach \i in {2} {\fill[blue!20, rounded corners, opacity=0.7] (\i) +(-0.4,-0.3) rectangle +(0.3,0.4);}
		\foreach \i in {0,...,2} { \node (\i) at (\i) {\nscale{$b_{i,\i}$}};}
\graph {
	(import nodes);
			0 ->["\escale{$X_i$}"]1->["\escale{$k_0$}"]2;  
			};
\end{scope}
\end{tikzpicture}
\end{flushleft}
Altogether, this proves the solvability of $k_1$.
So far, we have proven that all relevant ESSP atoms $(e,s)\in (E(A^\tau_\varphi)\setminus(\oplus\cup\ominus))\times (S(A^\tau_\varphi)\setminus \bot)$ are solvable.
Moreover, the presented regions are also useful for the solvability of all $(e,s)\in (E(A^\tau_\varphi)\setminus(\oplus\cup\ominus))\times  \bot$.
This is due to the following fact.
Let $(e,s)\in (E(A^\tau_\varphi)\setminus(\oplus\cup\ominus))\times (S(A^\tau_\varphi)\setminus \bot)$; let $(sup, sig)$ be a region that solves $(e,s)$; let $s'\in \bot$ be arbitrary; let $e'\in \ominus$, $e''\in \oplus$ and $t,t'\in S(A^\tau_\varphi)\setminus \bot$ such that $t\edge{e'}s'\edge{e''}t'$.
If $\neg sup(s')\edge{e}$, then $(e,s')$ is already solved.
Otherwise, if $sup(s')\edge{e}$, then we obtain a slightly different region $(sup', sig')$ that solves $(e,s')$ as follows:
For all $s''\in S(A^\tau_\varphi)\setminus\{s'\}$, we let $sup'(s'')=sup(s'')$.
Moreover, we define $sup'(s')=sup(s)$, in particular, now holds $\neg sup'(s')\edge{e}$.
For all $e'''\in E(A^\tau_\varphi)\setminus \{e',e''\}$, we let $sig'(e''')=sig(e''')$.
Finally, if $sup(t)=sup(s')$ ($sup(s')=sup(t')$), then we define $sig(e')=\nop$ ($sig(e'')=\nop$); 
if $sup(t)>sup(s')$ ($sup(s')>sup(t')$), then we define $sig(e')=\set$ ($sig(e'')=\set$); 
if $sup(t) < sup(s')$ ($sup(s') < sup(t')$), then we define $sig(e')=\inp$ ($sig(e'')=\inp$). 

Consequently, to prove $A^\tau_\varphi$'s ESSP, it remains to argue that the events of $\oplus\cup\ominus$ are also solvable.
For the events of $\oplus$, this is immediately clear. 
Let $e\in \ominus$ be arbitrary but fixed.
By the definition of the joining, $e$ occurs at a \enquote{terminal} state $s_n$ of a gadget $A_n$, that is, $e=\ominus_n$.
Let $s\in S(A^\tau_\varphi)\setminus S(A_n)$ be arbitrary but fixed.
To solve $(\ominus_n,s)$, we define a region $(sup, sig)$ as follows.
For all $s'\in S(A^\tau_\varphi)$, if $s'\in S(A_n)$, then we define $sup(s')=1$ and else $sup(s')=0$.
For all $e'\in E(A^\tau_\varphi)$, we define $sig(e')=\inp$ if $e'=\oplus_n$, $sig(e')=\set$ if $e'=\ominus_{n-1}$ and else $sig(e')=\nop$.
This yields a region that solves $(\ominus_n, s)$.
Since $s\in S(A^\tau_\varphi)\setminus S(A_n)$ was arbitrary, it remains to argue that $(\ominus_n, s)$ is solvable if $s\in S(A_n)$.

If $A=s_0\edge{e_1}\dots\edge{e_n}s_n$ is a gadget of $A^\tau_\varphi$, then for all $i\in \{0,\dots, n-1\}$ there is a region of $A^\tau_\varphi$ such that $sup(s_n)=1$, $sup(s_i)=0$ and $sig(\ominus_n)=\inp$.
We sketch the situation for $A_n=H_0$ and, thus, $s_n=h_{0,8}$.
The red (blue) colored area of the following figure sketches a region such that $sup(h_{0,0})=\dots=sup(h_{0,4})=0$ ($sup(h_{0,5})=sup(h_{0,6})=0$) and $sup(h_{0,8})=1$.
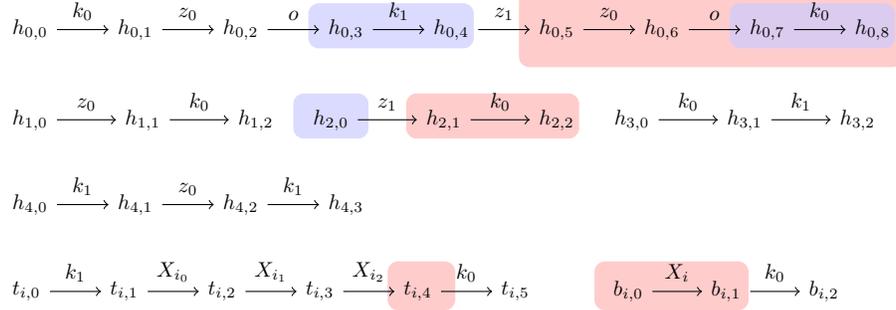
\begin{flushleft}
\begin{tikzpicture}[new set = import nodes]
\begin{scope}[nodes={set=import nodes}]%
		\foreach \i in {0,...,8} { \coordinate (\i) at (\i*1.4cm,0) ;}
		\foreach \i in {5} {\fill[red!20, rounded corners] (\i) +(-0.5,-0.5) rectangle +(4.6,0.5);}
		\foreach \i in {3,7} {\fill[blue!20, rounded corners, opacity=0.7] (\i) +(-0.5,-0.25) rectangle +(1.7,0.35);}
		\foreach \i in {0,...,8} { \node (\i) at (\i) {\nscale{$h_{0,\i}$}};}
\graph {
	(import nodes);
			0 ->["\escale{$k_0$}"]1->["\escale{$z_0$}"]2 ->["\escale{$o$}"]3 ->["\escale{$k_1$}"]4  ->["\escale{$z_1$}"]5->["\escale{$z_0$}"]6->["\escale{$o$}"]7->["\escale{$k_0$}"]8;
};
\end{scope}
\begin{scope}[yshift=-1.2cm,nodes={set=import nodes}]%
		\foreach \i in {0,...,2} { \coordinate (\i) at (\i*1.5cm,0) ;}
		\foreach \i in {0,...,2} { \node (\i) at (\i) {\nscale{$h_{1,\i}$}};}
\graph {
	(import nodes);
			0 ->["\escale{$z_0$}"]1->["\escale{$k_0$}"]2;
			};
\end{scope}
\begin{scope}[xshift=4cm, yshift=-1.2cm,nodes={set=import nodes}]%
		\foreach \i in {0,...,2} { \coordinate (\i) at (\i*1.5cm,0) ;}
		\foreach \i in {1} {\fill[red!20, rounded corners] (\i) +(-0.5,-0.25) rectangle +(1.8,0.35);}
		\foreach \i in {0} {\fill[blue!20, rounded corners, opacity=0.7] (\i) +(-0.5,-0.25) rectangle +(0.5,0.35);}
		\foreach \i in {0,...,2} { \node (\i) at (\i) {\nscale{$h_{2,\i}$}};}
\graph {
	(import nodes);
			0 ->["\escale{$z_1$}"]1->["\escale{$k_0$}"]2;
			};
\end{scope}
\begin{scope}[xshift=8cm,yshift=-1.2cm,nodes={set=import nodes}]%
		\foreach \i in {0,...,2} { \coordinate (\i) at (\i*1.5cm,0) ;}
		\foreach \i in {0,...,2} { \node (\i) at (\i) {\nscale{$h_{3,\i}$}};}
\graph {
	(import nodes);
			0 ->["\escale{$k_0$}"]1->["\escale{$k_1$}"]2;
			};
\end{scope}
\end{tikzpicture}
\end{flushleft}
\begin{flushleft}
\begin{tikzpicture}[new set = import nodes]
\begin{scope}[nodes={set=import nodes}]%
		\foreach \i in {0,...,3} { \coordinate (\i) at (\i*1.4cm,0) ;}
		\foreach \i in {0,...,3} { \node (\i) at (\i) {\nscale{$h_{4,\i}$}};}
\graph {
	(import nodes);
			0 ->["\escale{$k_1$}"]1->["\escale{$z_0$}"]2 ->["\escale{$k_1$}"]3;
			};
\end{scope}
\end{tikzpicture}
\end{flushleft}
\begin{flushleft}
\begin{tikzpicture}[new set = import nodes]
\begin{scope}[nodes={set=import nodes}]%
		\foreach \i in {0,...,5} { \coordinate (\i) at (\i*1.3cm,0) ;}
		\foreach \i in {4} {\fill[red!20, rounded corners] (\i) +(-0.4,-0.25) rectangle +(0.5,0.4);}
		\foreach \i in {0,...,5} { \node (\i) at (\i) {\nscale{$t_{i,\i}$}};}
\graph {
	(import nodes);
			0 ->["\escale{$k_1$}"]1->["\escale{$X_{i_0}$}"]2 ->["\escale{$X_{i_1}$}"]3->["\escale{$X_{i_2}$}"]4->["\escale{$k_0$}"]5; 
			};
\end{scope}
\begin{scope}[xshift=8cm,nodes={set=import nodes}]%
		\foreach \i in {0,...,2} { \coordinate (\i) at (\i*1.3cm,0) ;}
		\foreach \i in {0} {\fill[red!20, rounded corners] (\i) +(-0.45,-0.25) rectangle +(1.6,0.4);}
		\foreach \i in {0,...,2} { \node (\i) at (\i) {\nscale{$b_{i,\i}$}};}
\graph {
	(import nodes);
			0 ->["\escale{$X_i$}"]1->["\escale{$k_0$}"]2;  
			};
\end{scope}
\end{tikzpicture}
\end{flushleft}
The red colored area of the following figure sketches a region $(sup, sig)$ such that $sup(h_{0,7})$ and $sup(h_{0,8})=1$.
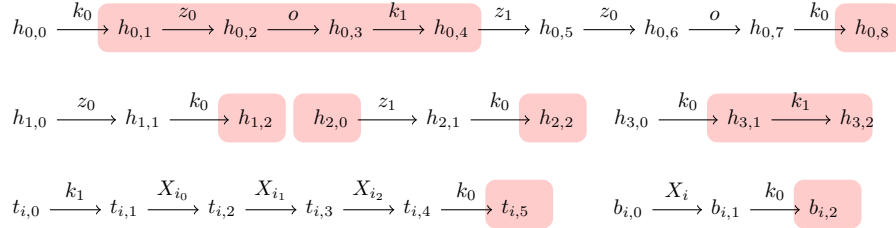
\begin{flushleft}
\begin{tikzpicture}[new set = import nodes]
\begin{scope}[nodes={set=import nodes}]%
		\foreach \i in {0,...,8} { \coordinate (\i) at (\i*1.4cm,0) ;}
		\foreach \i in {8} {\fill[red!20, rounded corners] (\i) +(-0.5,-0.3) rectangle +(0.4,0.35);}
		\foreach \i in {1} {\fill[red!20, rounded corners] (\i) +(-0.5,-0.3) rectangle +(4.6,0.35);}
		\foreach \i in {0,...,8} { \node (\i) at (\i) {\nscale{$h_{0,\i}$}};}
\graph {
	(import nodes);
			0 ->["\escale{$k_0$}"]1->["\escale{$z_0$}"]2 ->["\escale{$o$}"]3 ->["\escale{$k_1$}"]4  ->["\escale{$z_1$}"]5->["\escale{$z_0$}"]6->["\escale{$o$}"]7->["\escale{$k_0$}"]8;
};
\end{scope}
\begin{scope}[yshift=-1.2cm,nodes={set=import nodes}]%
		\foreach \i in {0,...,2} { \coordinate (\i) at (\i*1.5cm,0) ;}
		\foreach \i in {2} {\fill[red!20, rounded corners] (\i) +(-0.5,-0.3) rectangle +(0.4,0.35);}
		\foreach \i in {0,...,2} { \node (\i) at (\i) {\nscale{$h_{1,\i}$}};}
\graph {
	(import nodes);
			0 ->["\escale{$z_0$}"]1->["\escale{$k_0$}"]2;
			};
\end{scope}
\begin{scope}[xshift=4cm, yshift=-1.2cm,nodes={set=import nodes}]%
		\foreach \i in {0,...,2} { \coordinate (\i) at (\i*1.5cm,0) ;}
		\foreach \i in {0,2} {\fill[red!20, rounded corners] (\i) +(-0.5,-0.3) rectangle +(0.4,0.35);}
		\foreach \i in {0,...,2} { \node (\i) at (\i) {\nscale{$h_{2,\i}$}};}
\graph {
	(import nodes);
			0 ->["\escale{$z_1$}"]1->["\escale{$k_0$}"]2;
			};
\end{scope}
\begin{scope}[xshift=8cm,yshift=-1.2cm,nodes={set=import nodes}]%
		\foreach \i in {0,...,2} { \coordinate (\i) at (\i*1.5cm,0) ;}
		\foreach \i in {1} {\fill[red!20, rounded corners] (\i) +(-0.5,-0.3) rectangle +(1.7,0.35);}
		\foreach \i in {0,...,2} { \node (\i) at (\i) {\nscale{$h_{3,\i}$}};}
\graph {
	(import nodes);
			0 ->["\escale{$k_0$}"]1->["\escale{$k_1$}"]2;
			};
\end{scope}
\end{tikzpicture}
\end{flushleft}
\begin{flushleft}
\begin{tikzpicture}[new set = import nodes]
\begin{scope}[nodes={set=import nodes}]%
		\foreach \i in {0,...,5} { \coordinate (\i) at (\i*1.3cm,0) ;}
		\foreach \i in {5} {\fill[red!20, rounded corners] (\i) +(-0.4,-0.25) rectangle +(0.5,0.4);}
		\foreach \i in {0,...,5} { \node (\i) at (\i) {\nscale{$t_{i,\i}$}};}
\graph {
	(import nodes);
			0 ->["\escale{$k_1$}"]1->["\escale{$X_{i_0}$}"]2 ->["\escale{$X_{i_1}$}"]3->["\escale{$X_{i_2}$}"]4->["\escale{$k_0$}"]5; 
			};
\end{scope}
\begin{scope}[xshift=8cm,nodes={set=import nodes}]%
		\foreach \i in {0,...,2} { \coordinate (\i) at (\i*1.3cm,0) ;}
		\foreach \i in {2} {\fill[red!20, rounded corners] (\i) +(-0.4,-0.25) rectangle +(0.5,0.4);}
		\foreach \i in {0,...,2} { \node (\i) at (\i) {\nscale{$b_{i,\i}$}};}
\graph {
	(import nodes);
			0 ->["\escale{$X_i$}"]1->["\escale{$k_0$}"]2;  
			};
\end{scope}
\end{tikzpicture}
\end{flushleft}
This proves the ESSP of $A^\tau_\varphi$.
Moreover, it is easy to see that $A^\tau_\varphi$ has the SSP, too.

\subsection{Continuation of the Proof of Theorem~\ref{the:nop_inp_set_free+used}}
Let $e\in E(A^\tau_\varphi)\setminus (\ominus\cup\oplus)$. 
Let $A_{i_1},\dots,A_{i_n}$ be $A^\tau_\varphi$'s gadgets such that $e\in \bigcap_{j=1}^n E(A_{i_j})$.
Since $\free\in \tau$, we can simply define a region $(sup, sig)$ that solves $e$ in $A^\tau_\varphi \setminus (A_{i_1}\cup \dots \cup A_{i_n})$ as follows.
For all $s\in S(A^\tau_\varphi)$ we let $sup(s)=0$ if $s \in \bigcap_{j=1}^n S(A_{i_j})$ and else $sup(s)=1$.
Moreover, we let $sig(e)=\free$ and $sig(e')=\nop$ for all $e'\in E(A^\tau_\varphi)\setminus (\{e\}\cup\ominus\cup\oplus)$.
Finally, by the uniqueness of the events of $\ominus\cup\oplus$, it is easy to see that we can define their signatures appropriately to obtain a proper region $(sup, sig)$.
Thus, in what follows, if $e\in E(A^\tau_\varphi)\setminus (\ominus\cup\oplus)$, then we restrict ourselves to solve the ESSP atoms $(e,s)$ such that $e$ and $s$ appear in the same gadget.

$(k_0)$:
The region that solves $\alpha$, also solves all atoms $(k_0,s)$ where $s\in \{h_{0,2}, h_{0,3}, h_{0,4}\}$.
The red area of the following figure sketches the solvability of all $(k_0,s)$ where $s\in \{h_{0,1}, h_{0,6},b_{i_2,0}\}\cup (S(T_{i})\setminus \{t_{i,4}\})$ and $i\in \{0,\dots, m-1\}$ is arbitrary but fixed.
This proves the solvability of $k_0$.
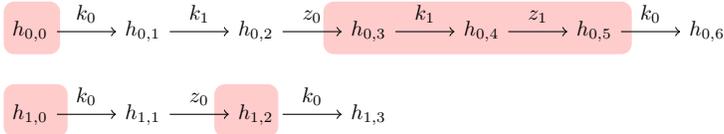
\begin{flushleft}
\begin{tikzpicture}[new set = import nodes]
\begin{scope}[nodes={set=import nodes}]%
		\foreach \i in {0,...,6} { \coordinate (\i) at (\i*1.4cm,0) ;}
		\foreach \i in {0} {\fill[red!20, rounded corners] (\i) +(-0.35,-0.3) rectangle +(0.4,0.4);}
		\foreach \i in {3} {\fill[red!20, rounded corners] (\i) +(-0.3,-0.3) rectangle +(3.8,0.4);}
		\foreach \i in {0,...,6} { \node (\i) at (\i*1.5cm,0) {\nscale{$h_{0,\i}$}};}
\graph {
	(import nodes);
			0 ->["\escale{$k_0$}"]1->["\escale{$k_1$}"]2 ->["\escale{$z_0$}"]3 ->["\escale{$k_1$}"]4  ->["\escale{$z_1$}"]5->["\escale{$k_0$}"]6;

			};
\end{scope}
\begin{scope}[yshift=-1.1cm,nodes={set=import nodes}]%
		\foreach \i in {0,...,3} { \coordinate (\i) at (\i*1.4cm,0) ;}
		\foreach \i in {0,2} {\fill[red!20, rounded corners] (\i) +(-0.35,-0.3) rectangle +(0.5,0.4);}
		\foreach \i in {0,...,3} { \node (\i) at (\i*1.5cm,0) {\nscale{$h_{1,\i}$}};}
\graph {
	(import nodes);
			0 ->["\escale{$k_0$}"]1->["\escale{$z_0$}"]2 ->["\escale{$k_0$}"]3;
			};
\end{scope}
\end{tikzpicture}
\end{flushleft}
\begin{flushleft}
\begin{tikzpicture}[new set = import nodes]
\begin{scope}[nodes={set=import nodes}]%
		\foreach \i in {0,...,5} { \coordinate (\i) at (\i*1.4cm,0) ;}
		\foreach \i in {4} {\fill[red!20, rounded corners] (\i) +(-0.35,-0.3) rectangle +(0.4,0.4);}
		\foreach \i in {0,...,5} { \node (\i) at (\i) {\nscale{$t_{i,\i}$}};}
\graph {
	(import nodes);
			0 ->["\escale{$k_1$}"]1->["\escale{$X_{i_0}$}"]2 ->["\escale{$X_{i_1}$}"]3->["\escale{$X_{i_2}$}"]4->["\escale{$k_0$}"]5; 
			};
\end{scope}
\begin{scope}[xshift=7.75cm,nodes={set=import nodes}]%
		\foreach \i in {0,...,2} { \coordinate (\i) at (\i*1.3cm,0) ;}
		\foreach \i in {1} {\fill[red!20, rounded corners] (\i) +(-0.35,-0.3) rectangle +(0.4,0.4);}
		\foreach \i in {0,...,2} { \node (\i) at (\i) {\nscale{$b_{i,\i}$}};}
\graph {
	(import nodes);
			0 ->["\escale{$X_i$}"]1->["\escale{$k_0$}"]2;  
			};
\end{scope}
\end{tikzpicture}
\end{flushleft}

($k_1$):
The red colored area of the following figure sketches a region that solves $k_1$ at all relevant states but $h_{0,6}$ and $t_{i,5}$.
The latter is done by the region which is sketched by the blue colored area.
For both regions, $sig(k_1)=\inp$.
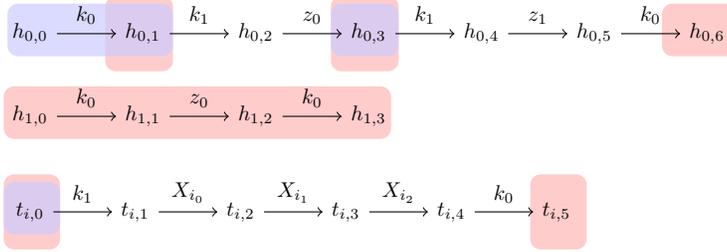
\begin{flushleft}
\begin{tikzpicture}[new set = import nodes]
\begin{scope}[nodes={set=import nodes}]%
		\foreach \i in {0,...,6} { \coordinate (\i) at (\i*1.4cm,0) ;}
		\foreach \i in {1} {\fill[red!20, rounded corners] (\i) +(-0.4,-0.5) rectangle +(0.5,0.5);}
		\foreach \i in {3} {\fill[red!20, rounded corners] (\i) +(-0.2,-0.5) rectangle +(0.7,0.5);}
		\foreach \i in {6} {\fill[red!20, rounded corners] (\i) +(0,-0.3) rectangle +(1,0.4);}
		\foreach \i in {0} {\fill[blue!20, rounded corners, opacity=0.7] (\i) +(-0.3,-0.3) rectangle +(1.9,0.4);}
		\foreach \i in {3} {\fill[blue!20, rounded corners, opacity=0.7] (\i) +(-0.2,-0.3) rectangle +(0.7,0.4);}
		\foreach \i in {0,...,6} { \node (\i) at (\i*1.5cm,0) {\nscale{$h_{0,\i}$}};}
\graph {
	(import nodes);
			0 ->["\escale{$k_0$}"]1->["\escale{$k_1$}"]2 ->["\escale{$z_0$}"]3 ->["\escale{$k_1$}"]4  ->["\escale{$z_1$}"]5->["\escale{$k_0$}"]6;

			};
\end{scope}
\begin{scope}[yshift=-1.1cm,nodes={set=import nodes}]%
		\foreach \i in {0,...,3} { \coordinate (\i) at (\i*1.4cm,0) ;}
		\foreach \i in {0} {\fill[red!20, rounded corners] (\i) +(-0.35,-0.3) rectangle +(4.8,0.4);}
		\foreach \i in {0,...,3} { \node (\i) at (\i*1.5cm,0) {\nscale{$h_{1,\i}$}};}
\graph {
	(import nodes);
			0 ->["\escale{$k_0$}"]1->["\escale{$z_0$}"]2 ->["\escale{$k_0$}"]3;
			};
\end{scope}
\end{tikzpicture}
\end{flushleft}
\begin{flushleft}
\begin{tikzpicture}[new set = import nodes]
\begin{scope}[nodes={set=import nodes}]%
		\foreach \i in {0,...,5} { \coordinate (\i) at (\i*1.4cm,0) ;}
		\foreach \i in {0,5} {\fill[red!20, rounded corners] (\i) +(-0.35,-0.5) rectangle +(0.4,0.5);}
		\foreach \i in {0} {\fill[blue!20, rounded corners, opacity=0.7] (\i) +(-0.35,-0.3) rectangle +(0.4,0.4);}
		\foreach \i in {0,...,5} { \node (\i) at (\i) {\nscale{$t_{i,\i}$}};}
\graph {
	(import nodes);
			0 ->["\escale{$k_1$}"]1->["\escale{$X_{i_0}$}"]2 ->["\escale{$X_{i_1}$}"]3->["\escale{$X_{i_2}$}"]4->["\escale{$k_0$}"]5; 
			};
\end{scope}
\end{tikzpicture}
\end{flushleft}

($z_0$):
The red colored area of the following figure sketches the solvability of $(z_0,s)$ for all relevant $s\in (S(H_0)\cup S(H_1))\setminus \{h_{0,4}, h_{0,5}, h_{0,6}, h_{1,0},\}$ and the blue colored area for $s\in \{h_{0,4}, h_{0,5},h_{1,0}\}$.
For both regions, $sig(z_0)=\inp$.
It is easy to see that $(z_0,h_{0,6})$ is solvable.
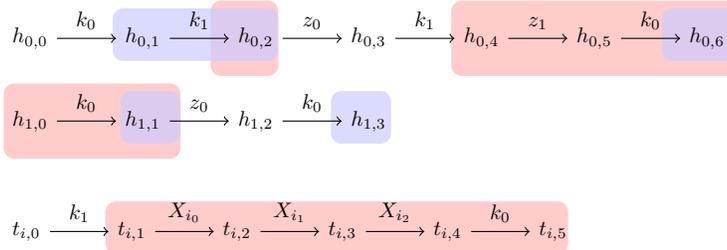
\begin{flushleft}
\begin{tikzpicture}[new set = import nodes]
\begin{scope}[nodes={set=import nodes}]%
		\foreach \i in {0,...,6} { \coordinate (\i) at (\i*1.4cm,0) ;}
		\foreach \i in {2} {\fill[red!20, rounded corners] (\i) +(-0.4,-0.5) rectangle +(0.5,0.5);}
		\foreach \i in {4} {\fill[red!20, rounded corners] (\i) +(0,-0.5) rectangle +(3.8,0.5);}
		\foreach \i in {1} {\fill[blue!20, rounded corners, opacity=0.7] (\i) +(-0.3,-0.3) rectangle +(1.9,0.4);}
		\foreach \i in {6} {\fill[blue!20, rounded corners, opacity=0.7] (\i) +(0,-0.3) rectangle +(1.,0.4);}
		\foreach \i in {0,...,6} { \node (\i) at (\i*1.5cm,0) {\nscale{$h_{0,\i}$}};}
\graph {
	(import nodes);
			0 ->["\escale{$k_0$}"]1->["\escale{$k_1$}"]2 ->["\escale{$z_0$}"]3 ->["\escale{$k_1$}"]4  ->["\escale{$z_1$}"]5->["\escale{$k_0$}"]6;

			};
\end{scope}
\begin{scope}[yshift=-1.1cm,nodes={set=import nodes}]%
		\foreach \i in {0,...,3} { \coordinate (\i) at (\i*1.4cm,0) ;}
		\foreach \i in {0} {\fill[red!20, rounded corners] (\i) +(-0.35,-0.5) rectangle +(2,0.5);}
		\foreach \i in {1,3} {\fill[blue!20, rounded corners, opacity=0.7] (\i) +(-0.2,-0.3) rectangle +(0.6,0.4);}
		\foreach \i in {0,...,3} { \node (\i) at (\i*1.5cm,0) {\nscale{$h_{1,\i}$}};}
\graph {
	(import nodes);
			0 ->["\escale{$k_0$}"]1->["\escale{$z_0$}"]2 ->["\escale{$k_0$}"]3;
			};
\end{scope}
\end{tikzpicture}
\end{flushleft}
\begin{flushleft}
\begin{tikzpicture}[new set = import nodes]
\begin{scope}[nodes={set=import nodes}]%
		\foreach \i in {0,...,5} { \coordinate (\i) at (\i*1.4cm,0) ;}
		\foreach \i in {1} {\fill[red!20, rounded corners] (\i) +(-0.35,-0.3) rectangle +(5.8,0.4);}
		\foreach \i in {0,...,5} { \node (\i) at (\i) {\nscale{$t_{i,\i}$}};}
\graph {
	(import nodes);
			0 ->["\escale{$k_1$}"]1->["\escale{$X_{i_0}$}"]2 ->["\escale{$X_{i_1}$}"]3->["\escale{$X_{i_2}$}"]4->["\escale{$k_0$}"]5; 
			};
\end{scope}
\end{tikzpicture}
\end{flushleft}

($z_1$):
The red colored area of the following figure sketches the solvability of $(z_1,s)$ for all relevant $s\in S(H_0)\setminus \{h_{0,2}\}$ and the blue colored area for $(z_1,h_{0,2})$.
For both regions, $sig(z_1)=\inp$.
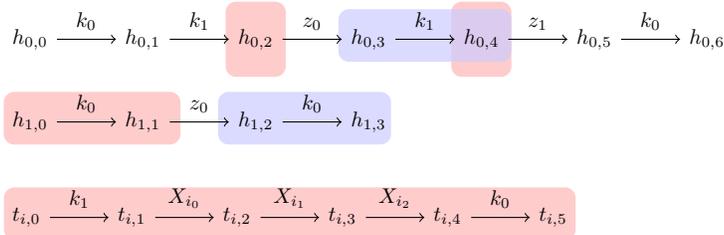
\begin{flushleft}
\begin{tikzpicture}[new set = import nodes]
\begin{scope}[nodes={set=import nodes}]%
		\foreach \i in {0,...,6} { \coordinate (\i) at (\i*1.4cm,0) ;}
		\foreach \i in {2} {\fill[red!20, rounded corners] (\i) +(-0.2,-0.5) rectangle +(0.6,0.5);}
		\foreach \i in {4} {\fill[red!20, rounded corners] (\i) +(0,-0.5) rectangle +(0.8,0.5);}
		\foreach \i in {3} {\fill[blue!20, rounded corners, opacity=0.7] (\i) +(-0.1,-0.3) rectangle +(2.2,0.4);}
		%
		\foreach \i in {0,...,6} { \node (\i) at (\i*1.5cm,0) {\nscale{$h_{0,\i}$}};}
\graph {
	(import nodes);
			0 ->["\escale{$k_0$}"]1->["\escale{$k_1$}"]2 ->["\escale{$z_0$}"]3 ->["\escale{$k_1$}"]4  ->["\escale{$z_1$}"]5->["\escale{$k_0$}"]6;

			};
\end{scope}
\begin{scope}[yshift=-1.1cm,nodes={set=import nodes}]%
		\foreach \i in {0,...,3} { \coordinate (\i) at (\i*1.4cm,0) ;}
		\foreach \i in {0} {\fill[red!20, rounded corners] (\i) +(-0.35,-0.3) rectangle +(2,0.4);}
		\foreach \i in {2} {\fill[blue!20, rounded corners, opacity=0.7] (\i) +(-0.3,-0.3) rectangle +(2.0,0.4);}
		\foreach \i in {0,...,3} { \node (\i) at (\i*1.5cm,0) {\nscale{$h_{1,\i}$}};}
\graph {
	(import nodes);
			0 ->["\escale{$k_0$}"]1->["\escale{$z_0$}"]2 ->["\escale{$k_0$}"]3;
			};
\end{scope}
\end{tikzpicture}
\end{flushleft}
\begin{flushleft}
\begin{tikzpicture}[new set = import nodes]
\begin{scope}[nodes={set=import nodes}]%
		\foreach \i in {0,...,5} { \coordinate (\i) at (\i*1.4cm,0) ;}
		\foreach \i in {0} {\fill[red!20, rounded corners] (\i) +(-0.3,-0.3) rectangle +(7.3,0.4);}
		\foreach \i in {0,...,5} { \node (\i) at (\i) {\nscale{$t_{i,\i}$}};}
\graph {
	(import nodes);
			0 ->["\escale{$k_1$}"]1->["\escale{$X_{i_0}$}"]2 ->["\escale{$X_{i_1}$}"]3->["\escale{$X_{i_2}$}"]4->["\escale{$k_0$}"]5; 
			};
\end{scope}
\end{tikzpicture}
\end{flushleft}

($X_i$):
It is easy to see, that the variable events are solvable.

It remains to prove that all $e\in \ominus\cup\oplus$ are solvable and that $A^\tau_\varphi$ has the $\tau$-SSP.
We can do this similarly to the proof of Theorem~\ref{the:nop_inp_out_set+used_free}.

%
\subsection{Continuation of the Proof of Theorem~\ref{the:nop_inp_res_swap+used_free}}
($k$):
If $s\in (\bigcup_{i=0}^3S(H_i)\cup S(F_0)\cup S(F_1)\cup \bigcup_{i=0}^{10}S(G_i))$ and $s\not\in \{h_{1,2}, h_{2,3}, h_{3,0}, g_{j,3}\mid 0\leq j\leq 10\}$, then the region that solves $\alpha$ also solves $(k,s)$.
The red colored area of the following figure sketches how $k$ can be solved at the remaining states of $A^\tau_\varphi$'s gadgets.
For this (these) region(s) holds $sig(k)=\inp$, and it is easy to see, that the situation sketched for $T_{i,0}$ can be transferred to all $T_{i,1},\dots, T_{i,6}$ and $B_i$, $i\in \{0,\dots, m-1\}$.
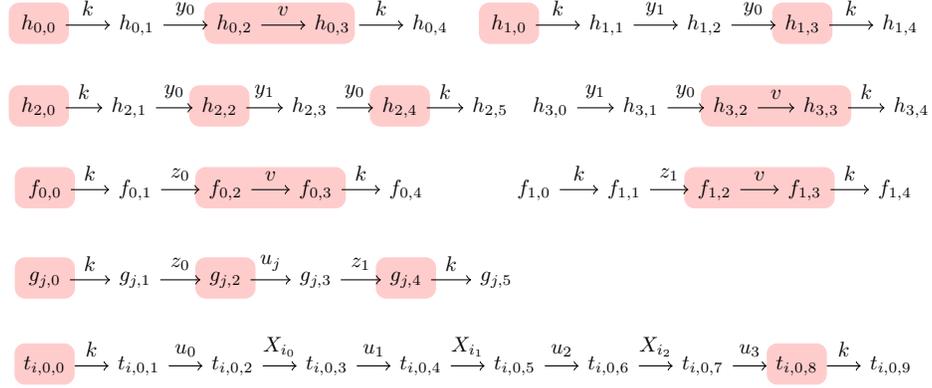
\begin{center}
\begin{tikzpicture}[new set = import nodes]
\begin{scope}[nodes={set=import nodes}]

		\foreach \i in {0,...,4} { \coordinate (\i) at (\i*1.3cm,0) ;}
		\foreach \i in {0} {\fill[red!20, rounded corners] (\i) +(-0.4,-0.25) rectangle +(0.4,0.3);}
		\foreach \i in {2} {\fill[red!20, rounded corners] (\i) +(-0.4,-0.25) rectangle +(1.6,0.3);}
		\foreach \i in {0,...,4} { \node (\i) at (\i) {\nscale{$h_{0,\i}$}};}
\graph {
	(import nodes);
			0 ->["\escale{$k$}"]1->["\escale{$y_0$}"]2->["\escale{$v$}"]3->["\escale{$k$}"]4;  
			};
\end{scope}
\begin{scope}[xshift=6.25cm,nodes={set=import nodes}]%

		\foreach \i in {0,...,4} { \coordinate (\i) at (\i*1.3cm,0) ;}
		\foreach \i in {0,3} {\fill[red!20, rounded corners] (\i) +(-0.4,-0.25) rectangle +(0.4,0.3);}
		\foreach \i in {0,...,4} { \node (\i) at (\i) {\nscale{$h_{1,\i}$}};}
\graph {
	(import nodes);
			0 ->["\escale{$k$}"]1->["\escale{$y_1$}"]2->["\escale{$y_0$}"]3->["\escale{$k$}"]4;  
			};
\end{scope}
\begin{scope}[yshift=-1.1cm,nodes={set=import nodes}]%

		\foreach \i in {0,...,5} { \coordinate (\i) at (\i*1.2cm,0) ;}
		\foreach \i in {0,2,4} {\fill[red!20, rounded corners] (\i) +(-0.4,-0.25) rectangle +(0.4,0.3);}
		\foreach \i in {0,...,5} { \node (\i) at (\i) {\nscale{$h_{2,\i}$}};}
\graph {
	(import nodes);
			0 ->["\escale{$k$}"]1->["\escale{$y_0$}"]2->["\escale{$y_1$}"]3->["\escale{$y_0$}"]4->["\escale{$k$}"]5;  
			};
\end{scope}
\begin{scope}[xshift=6.8cm, yshift=-1.1cm,nodes={set=import nodes}]%

		\foreach \i in {0,...,4} { \coordinate (\i) at (\i*1.2cm,0) ;}
		\foreach \i in {2} {\fill[red!20, rounded corners] (\i) +(-0.4,-0.25) rectangle +(1.6,0.3);}		
		\foreach \i in {0,...,4} { \node (\i) at (\i) {\nscale{$h_{3,\i}$}};}
\graph {
	(import nodes);
			0 ->["\escale{$y_1$}"]1->["\escale{$y_0$}"]2->["\escale{$v$}"]3->["\escale{$k$}"]4;  
			};
\end{scope}
\end{tikzpicture}
\end{center}
\begin{center}
\begin{tikzpicture}[new set = import nodes]
\begin{scope}[nodes={set=import nodes}]
		
		\foreach \i in {0,...,4} { \coordinate (\i) at (\i*1.2cm,0) ;}
		\foreach \i in {2} {\fill[red!20, rounded corners] (\i) +(-0.4,-0.25) rectangle +(1.6,0.3);}
		\foreach \i in {0} {\fill[red!20, rounded corners] (\i) +(-0.4,-0.25) rectangle +(0.4,0.3);}
		\foreach \i in {0,...,4} { \node (\i) at (\i) {\nscale{$f_{0,\i}$}};}
\graph {
	(import nodes);
			0 ->["\escale{$k$}"]1->["\escale{$z_0$}"]2->["\escale{$v$}"]3->["\escale{$k$}"]4;  
			};
\end{scope}
\begin{scope}[xshift=6.5cm,nodes={set=import nodes}]
		
		\foreach \i in {0,...,4} { \coordinate (\i) at (\i*1.2cm,0) ;}		
		\foreach \i in {2} {\fill[red!20, rounded corners] (\i) +(-0.4,-0.25) rectangle +(1.6,0.3);}
		\foreach \i in {0,...,4} { \node (\i) at (\i) {\nscale{$f_{1,\i}$}};}
\graph {
	(import nodes);
			0 ->["\escale{$k$}"]1->["\escale{$z_1$}"]2->["\escale{$v$}"]3->["\escale{$k$}"]4;  
			};
\end{scope}
\begin{scope}[yshift=-1.2cm,nodes={set=import nodes}]
		
		\foreach \i in {0,...,5} { \coordinate (\i) at (\i*1.2cm,0) ;}
		\foreach \i in {0,2,4} {\fill[red!20, rounded corners] (\i) +(-0.4,-0.25) rectangle +(0.4,0.3);}
		\foreach \i in {0,...,5} { \node (\i) at (\i) {\nscale{$g_{j,\i}$}};}
\graph {
	(import nodes);
			0 ->["\escale{$k$}"]1->["\escale{$z_0$}"]2->["\escale{$u_j$}"]3->["\escale{$z_1$}"]4->["\escale{$k$}"]5;  
			};
\end{scope}
\end{tikzpicture}
\end{center}
\begin{center}
\begin{tikzpicture}[new set = import nodes]
\begin{scope}[nodes={set=import nodes}]
		
		\foreach \i in {0,...,9} { \coordinate (\i) at (\i*1.25cm,0) ;}
		\foreach \i in {0,8} {\fill[red!20, rounded corners] (\i) +(-0.4,-0.25) rectangle +(0.4,0.3);}
		\foreach \i in {0,...,9} { \node (\i) at (\i) {\nscale{$t_{i,0,\i}$}};}
\graph {
	(import nodes);
			0 ->["\escale{$k$}"]1->["\escale{$u_{0}$}"]2->["\escale{$X_{i_0}$}"]3->["\escale{$u_{1}$}"]4->["\escale{$X_{i_1}$}"]5->["\escale{$u_{2}$}"]6->["\escale{$X_{i_2}$}"]7->["\escale{$u_{3}$}"]8->["\escale{$k$}"]9;
			};
\end{scope}
\end{tikzpicture}
\end{center}

($y_0$):
The red colored area of the following figure sketches the solvability of $(y_0,s)$ for all $s\in \bigcup_{i=0}^3S(H_i)\setminus\{h_{0,4}, h_{1,0}, h_{1,4}, h_{3,4}\}$.
For this region, $sig(y_0)=\inp$.
Since we can define $sig(k)=\swap$, it is easy to see that this region is extendable to $A^\tau_\varphi$.
Moreover, the blue colored area sketches the solvability of $(y_0,s)$ for all $s\in \{h_{0,4}, h_{1,0}, h_{1,4}, h_{3,4}\}$ and all remaining relevant states $s$ of $A^\tau_\varphi$'s gadgets.
Again, $sig(y_0)=\inp$.
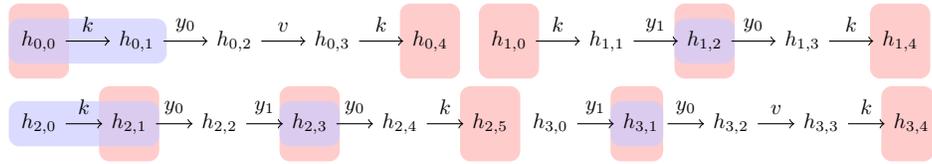
\begin{center}
\begin{tikzpicture}[new set = import nodes]
\begin{scope}[nodes={set=import nodes}]

		\foreach \i in {0,...,4} { \coordinate (\i) at (\i*1.3cm,0) ;}
		\foreach \i in {0,4} {\fill[red!20, rounded corners] (\i) +(-0.4,-0.5) rectangle +(0.4,0.5);}
		\foreach \i in {0} {\fill[blue!20, rounded corners, opacity=0.7] (\i) +(-0.4,-0.3) rectangle +(1.7,0.3);}
		\foreach \i in {0,...,4} { \node (\i) at (\i) {\nscale{$h_{0,\i}$}};}
\graph {
	(import nodes);
			0 ->["\escale{$k$}"]1->["\escale{$y_0$}"]2->["\escale{$v$}"]3->["\escale{$k$}"]4;  
			};
\end{scope}
\begin{scope}[xshift=6.25cm,nodes={set=import nodes}]%

		\foreach \i in {0,...,4} { \coordinate (\i) at (\i*1.3cm,0) ;}
		\foreach \i in {0,2,4} {\fill[red!20, rounded corners] (\i) +(-0.4,-0.5) rectangle +(0.4,0.5);}
		\foreach \i in {2} {\fill[blue!20, rounded corners, opacity=0.7] (\i) +(-0.4,-0.3) rectangle +(0.4,0.3);}
		\foreach \i in {0,...,4} { \node (\i) at (\i) {\nscale{$h_{1,\i}$}};}
\graph {
	(import nodes);
			0 ->["\escale{$k$}"]1->["\escale{$y_1$}"]2->["\escale{$y_0$}"]3->["\escale{$k$}"]4;  
			};
\end{scope}
\begin{scope}[yshift=-1.1cm,nodes={set=import nodes}]%

		\foreach \i in {0,...,5} { \coordinate (\i) at (\i*1.2cm,0) ;}
		\foreach \i in {1,3,5} {\fill[red!20, rounded corners] (\i) +(-0.4,-0.5) rectangle +(0.4,0.5);}
		\foreach \i in {0} {\fill[blue!20, rounded corners, opacity=0.7] (\i) +(-0.4,-0.3) rectangle +(1.6,0.3);}
		\foreach \i in {3} {\fill[blue!20, rounded corners, opacity=0.7] (\i) +(-0.4,-0.3) rectangle +(0.4,0.3);}
		\foreach \i in {0,...,5} { \node (\i) at (\i) {\nscale{$h_{2,\i}$}};}
\graph {
	(import nodes);
			0 ->["\escale{$k$}"]1->["\escale{$y_0$}"]2->["\escale{$y_1$}"]3->["\escale{$y_0$}"]4->["\escale{$k$}"]5;  
			};
\end{scope}
\begin{scope}[xshift=6.8cm, yshift=-1.1cm,nodes={set=import nodes}]%

		\foreach \i in {0,...,4} { \coordinate (\i) at (\i*1.2cm,0) ;}
		\foreach \i in {1,4} {\fill[red!20, rounded corners] (\i) +(-0.4,-0.5) rectangle +(0.3,0.5);}
		\foreach \i in {1} {\fill[blue!20, rounded corners, opacity=0.7] (\i) +(-0.4,-0.3) rectangle +(0.3,0.3);}
		\foreach \i in {0,...,4} { \node (\i) at (\i) {\nscale{$h_{3,\i}$}};}
\graph {
	(import nodes);
			0 ->["\escale{$y_1$}"]1->["\escale{$y_0$}"]2->["\escale{$v$}"]3->["\escale{$k$}"]4;  
			};
\end{scope}
\end{tikzpicture}
\end{center}

($y_1$):
The red colored area of the following figure sketches the solvability of $(y_1,s)$ for all $s\in \{h_{0,0}, h_{0,1}, h_{1,2}, h_{2,0}, h_{2,1}, h_{3,1}\}$.
The blue  colored area does this for $(y_1,h_{1,0})$.
It is easy to see that the remaining atoms $(y_1,s)$ are solvable, too.
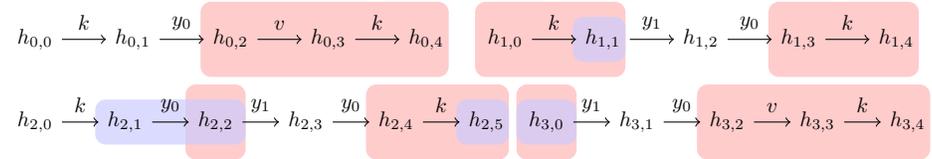
\begin{center}
\begin{tikzpicture}[new set = import nodes]
\begin{scope}[nodes={set=import nodes}]

		\foreach \i in {0,...,4} { \coordinate (\i) at (\i*1.3cm,0) ;}
		\foreach \i in {2} {\fill[red!20, rounded corners] (\i) +(-0.4,-0.5) rectangle +(2.9,0.5);}
		\foreach \i in {0,...,4} { \node (\i) at (\i) {\nscale{$h_{0,\i}$}};}
\graph {
	(import nodes);
			0 ->["\escale{$k$}"]1->["\escale{$y_0$}"]2->["\escale{$v$}"]3->["\escale{$k$}"]4;  
			};
\end{scope}
\begin{scope}[xshift=6.25cm,nodes={set=import nodes}]%

		\foreach \i in {0,...,4} { \coordinate (\i) at (\i*1.3cm,0) ;}
		\foreach \i in {0,3} {\fill[red!20, rounded corners] (\i) +(-0.4,-0.5) rectangle +(1.6,0.5);}
		\foreach \i in {1} {\fill[blue!20, rounded corners, opacity=0.7] (\i) +(-0.4,-0.3) rectangle +(0.3,0.3);}
		\foreach \i in {0,...,4} { \node (\i) at (\i) {\nscale{$h_{1,\i}$}};}
\graph {
	(import nodes);
			0 ->["\escale{$k$}"]1->["\escale{$y_1$}"]2->["\escale{$y_0$}"]3->["\escale{$k$}"]4;  
			};
\end{scope}
\begin{scope}[yshift=-1.1cm,nodes={set=import nodes}]%

		\foreach \i in {0,...,5} { \coordinate (\i) at (\i*1.2cm,0) ;}
		\foreach \i in {2} {\fill[red!20, rounded corners] (\i) +(-0.4,-0.5) rectangle +(0.4,0.5);}
		\foreach \i in {4} {\fill[red!20, rounded corners] (\i) +(-0.4,-0.5) rectangle +(1.5,0.5);}
		\foreach \i in {5} {\fill[blue!20, rounded corners, opacity=0.7] (\i) +(-0.4,-0.3) rectangle +(0.3,0.3);}
		\foreach \i in {1} {\fill[blue!20, rounded corners, opacity=0.7] (\i) +(-0.4,-0.3) rectangle +(1.6,0.3);}
		\foreach \i in {0,...,5} { \node (\i) at (\i) {\nscale{$h_{2,\i}$}};}
\graph {
	(import nodes);
			0 ->["\escale{$k$}"]1->["\escale{$y_0$}"]2->["\escale{$y_1$}"]3->["\escale{$y_0$}"]4->["\escale{$k$}"]5;  
			};
\end{scope}
\begin{scope}[xshift=6.8cm, yshift=-1.1cm,nodes={set=import nodes}]%

		\foreach \i in {0,...,4} { \coordinate (\i) at (\i*1.2cm,0) ;}
		\foreach \i in {0} {\fill[red!20, rounded corners] (\i) +(-0.4,-0.5) rectangle +(0.4,0.5);}
		\foreach \i in {2} {\fill[red!20, rounded corners] (\i) +(-0.4,-0.5) rectangle +(2.7,0.5);}
		\foreach \i in {0} {\fill[blue!20, rounded corners, opacity=0.7] (\i) +(-0.4,-0.3) rectangle +(0.4,0.3);}
		\foreach \i in {0,...,4} { \node (\i) at (\i) {\nscale{$h_{3,\i}$}};}
\graph {
	(import nodes);
			0 ->["\escale{$y_1$}"]1->["\escale{$y_0$}"]2->["\escale{$v$}"]3->["\escale{$k$}"]4;  
			};
\end{scope}
\end{tikzpicture}
\end{center}

($v$):
The red colored area of the following figure sketches the solvability of $(v,s)$ for all relevant $s\in \bigcup_{i=0}^3S(H_i)\cup S(F_0)\cup S(F_1)$ except for $s\in \{h_{1,3}, h_{2,2}, h_{2,3}, h_{2,4}, h_{2,5}\}$.
The blue colored area sketches the solvability of the remaining atoms.
For both regions, $sig(v)=\inp$, and it is easy to see that they can be extended to $A^\tau_\varphi$.
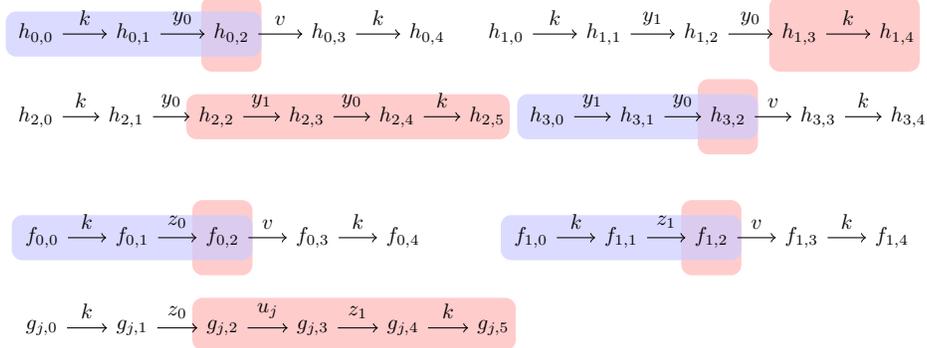
\begin{center}
\begin{tikzpicture}[new set = import nodes]
\begin{scope}[nodes={set=import nodes}]

		\foreach \i in {0,...,4} { \coordinate (\i) at (\i*1.3cm,0) ;}
		\foreach \i in {2} {\fill[red!20, rounded corners] (\i) +(-0.4,-0.5) rectangle +(0.4,0.5);}
		\foreach \i in {0} {\fill[blue!20, rounded corners, opacity=0.7] (\i) +(-0.4,-0.3) rectangle +(3,0.3);}
		\foreach \i in {0,...,4} { \node (\i) at (\i) {\nscale{$h_{0,\i}$}};}
\graph {(import nodes);
	0 ->["\escale{$k$}"]1->["\escale{$y_0$}"]2->["\escale{$v$}"]3->["\escale{$k$}"]4;  
			};
\end{scope}
\begin{scope}[xshift=6.25cm,nodes={set=import nodes}]%

		\foreach \i in {0,...,4} { \coordinate (\i) at (\i*1.3cm,0) ;}
		\foreach \i in {3} {\fill[red!20, rounded corners] (\i) +(-0.4,-0.5) rectangle +(1.6,0.5);}
		\foreach \i in {0,...,4} { \node (\i) at (\i) {\nscale{$h_{1,\i}$}};}
\graph {
	(import nodes);
			0 ->["\escale{$k$}"]1->["\escale{$y_1$}"]2->["\escale{$y_0$}"]3->["\escale{$k$}"]4;  
			};
\end{scope}
\begin{scope}[yshift=-1.1cm,nodes={set=import nodes}]%

		\foreach \i in {0,...,5} { \coordinate (\i) at (\i*1.2cm,0) ;}
		\foreach \i in {2} {\fill[red!20, rounded corners] (\i) +(-0.4,-0.3) rectangle +(3.9,0.3);}
		\foreach \i in {0,...,5} { \node (\i) at (\i) {\nscale{$h_{2,\i}$}};}
\graph {
	(import nodes);
			0 ->["\escale{$k$}"]1->["\escale{$y_0$}"]2->["\escale{$y_1$}"]3->["\escale{$y_0$}"]4->["\escale{$k$}"]5;  
			};
\end{scope}
\begin{scope}[xshift=6.8cm, yshift=-1.1cm,nodes={set=import nodes}]%

		\foreach \i in {0,...,4} { \coordinate (\i) at (\i*1.2cm,0) ;}
		\foreach \i in {2} {\fill[red!20, rounded corners] (\i) +(-0.4,-0.5) rectangle +(0.4,0.5);}
		\foreach \i in {0} {\fill[blue!20, rounded corners, opacity=0.7] (\i) +(-0.4,-0.3) rectangle +(2.8,0.3);}
		\foreach \i in {0,...,4} { \node (\i) at (\i) {\nscale{$h_{3,\i}$}};}
\graph {
	(import nodes);
			0 ->["\escale{$y_1$}"]1->["\escale{$y_0$}"]2->["\escale{$v$}"]3->["\escale{$k$}"]4;  
			};
\end{scope}
\end{tikzpicture}
\end{center}
\begin{center}
\begin{tikzpicture}[new set = import nodes]
\begin{scope}[nodes={set=import nodes}]
		
		\foreach \i in {0,...,4} { \coordinate (\i) at (\i*1.2cm,0) ;}
		\foreach \i in {2} {\fill[red!20, rounded corners] (\i) +(-0.4,-0.5) rectangle +(0.4,0.5);}
		\foreach \i in {0} {\fill[blue!20, rounded corners, opacity=0.7] (\i) +(-0.4,-0.3) rectangle +(2.8,0.3);}
		\foreach \i in {0,...,4} { \node (\i) at (\i) {\nscale{$f_{0,\i}$}};}
\graph {
	(import nodes);
			0 ->["\escale{$k$}"]1->["\escale{$z_0$}"]2->["\escale{$v$}"]3->["\escale{$k$}"]4;  
			};
\end{scope}
\begin{scope}[xshift=6.5cm,nodes={set=import nodes}]
		
		\foreach \i in {0,...,4} { \coordinate (\i) at (\i*1.2cm,0) ;}		
		\foreach \i in {2} {\fill[red!20, rounded corners] (\i) +(-0.4,-0.5) rectangle +(0.4,0.5);}
		\foreach \i in {0} {\fill[blue!20, rounded corners, opacity=0.7] (\i) +(-0.4,-0.3) rectangle +(2.8,0.3);}
		\foreach \i in {0,...,4} { \node (\i) at (\i) {\nscale{$f_{1,\i}$}};}
\graph {
	(import nodes);
			0 ->["\escale{$k$}"]1->["\escale{$z_1$}"]2->["\escale{$v$}"]3->["\escale{$k$}"]4;  
			};
\end{scope}
\begin{scope}[yshift=-1.2cm,nodes={set=import nodes}]
		
		\foreach \i in {0,...,5} { \coordinate (\i) at (\i*1.2cm,0) ;}
		\foreach \i in {2} {\fill[red!20, rounded corners] (\i) +(-0.4,-0.3) rectangle +(3.9,0.4);}
		\foreach \i in {0,...,5} { \node (\i) at (\i) {\nscale{$g_{j,\i}$}};}
\graph {
	(import nodes);
			0 ->["\escale{$k$}"]1->["\escale{$z_0$}"]2->["\escale{$u_j$}"]3->["\escale{$z_1$}"]4->["\escale{$k$}"]5;  
			};
\end{scope}
\end{tikzpicture}
\end{center}

It is easy to see that the remaining ESSP atoms $(e,s)$ for all $e\in \{z_0,z_1\}\cup V(\varphi)\cup\{u_0,\dots, u_{10}\}\cup\{w_0,\dots, w_{3m-1}\}$ and $s\in S(A^\tau_\varphi)$ are solvable, too.
It remains to prove that all $e\in \ominus\cup\oplus$ are solvable and that $A^\tau_\varphi$ has the $\tau$-SSP.
We can do this similarly to the proof of Theorem~\ref{the:nop_inp_out_set+used_free}.

\subsection{Continuation of the Proof of Theorem~\ref{the:nop_inp_set_swap+out_res_used_free}}
$(k)$:
The region $(sup, sig)$ that solves $\alpha$ also solves $(k,s)$ for all $s\in \{s'\in S(A^\tau_\varphi) \mid s'\edge{k}\}\cup\{h_{1,0}, h_{3,3}, h_{4,2},h_{6,2},h_{8,2},h_{10,2}\}\cup\{f_{0,2},f_{1,2},f_{2,2},f_{2,3}\}\cup\{g_{i,2}, n_{i,3}\mid i\in \{0,\dots, 13\}\}$.
This region can simply modified to a region $(sup', sig')$ that solves $(k,h_{3,2})$.
To do we define $(sup', sig')$ for all $e\in E(A^\tau_\varphi)$ and $s\in S(A^\tau_\varphi)$ as follows:
If $s=h_{3,2}$, then $sup'(s)=0$; if $s=h_{3,3}$, then $sup'(s)=1$; if $s\not\in \{h_{3,2},h_{3,3}$, then $sup'(s)=sup(s)$.
If $e=v_1$, then $sig'(e)=7nop$; if $e\not=v_1$, then $sig'(e)=sig(e)$.

The red colored area of the following figure sketches the solvability of $(k,s)$ for all remaining relevant states $s$ of $A^\tau_\varphi$ except for $s =h_{11,2}$.
It is easy to see that $(sup, sig)$, as sketched for $T_{i,0}$, can be transferred to $T_{i,0},\dots, T_{i,6}$.
Moreover, the region $(sup, sig)$ can be easily modified to a region $(sup', sig')$ that solves $(k,h_{11,2})$.
This modified region maps $h_{11,2}$ to $0$ and $h_{9,2}$ and $h_{12,3}$ to $1$, cf. the blue area.
For both regions, $sig(k)=\inp$.
Altogether, this proves the solvability of $k$.
\begin{center}
\begin{tikzpicture}[new set = import nodes]
\begin{scope}[nodes={set=import nodes}]%

		\foreach \i in {0,...,2} { \coordinate (\i) at (\i*1.4cm,0) ;}
		\foreach \i in {0} {\fill[red!20, rounded corners] (\i) +(-0.4,-0.25) rectangle +(0.4,0.35);}
		\foreach \i in {0,...,2} { \node (\i) at (\i) {\nscale{$h_{0,\i}$}};}
\graph {(import nodes);
			0 ->["\escale{$k$}"]1->["\escale{$v_0$}"]2;  
			};
\end{scope}
\begin{scope}[xshift=5cm,nodes={set=import nodes}]%

		\foreach \i in {0,...,2} { \coordinate (\i) at (\i*1.4cm,0) ;}
		\foreach \i in {0} {\fill[red!20, rounded corners] (\i) +(-0.4,-0.25) rectangle +(1.7,0.35);}
		\foreach \i in {0,...,2} { \node (\i) at (\i) {\nscale{$h_{1,\i}$}};}
\graph {(import nodes);
			0 ->["\escale{$v_0$}"]1->["\escale{$k$}"]2;  
			};
\end{scope}
\begin{scope}[yshift=-1.1cm, nodes={set=import nodes}]%
		
		\foreach \i in {0,...,4} { \coordinate (\i) at (\i*1.4cm,0) ;}
		\foreach \i in {0,3} {\fill[red!20, rounded corners] (\i) +(-0.4,-0.25) rectangle +(0.4,0.35);}
		\foreach \i in {0,...,4} { \node (\i) at (\i) {\nscale{$h_{2,\i}$}};}
\graph {(import nodes);
			0 ->["\escale{$k$}"]1->["\escale{$v_0$}"]2->["\escale{$v_1$}"]3->["\escale{$k$}"]4;  
			};
\end{scope}

\begin{scope}[xshift=6.9cm, yshift=-1.1cm,nodes={set=import nodes}]%
		
		\foreach \i in {0,...,3} { \coordinate (\i) at (\i*1.4cm,0) ;}
		\foreach \i in {0} {\fill[red!20, rounded corners] (\i) +(-0.4,-0.25) rectangle +(0.4,0.35);}
		\foreach \i in {2} {\fill[red!20, rounded corners] (\i) +(-0.4,-0.25) rectangle +(1.7,0.35);}
		\foreach \i in {0,...,3} { \node (\i) at (\i) {\nscale{$h_{3,\i}$}};}
\graph {
	(import nodes);
			0 ->["\escale{$k$}"]1->["\escale{$v_1$}"]2->["\escale{$v_0$}"]3;  
			};
\end{scope}
\end{tikzpicture}
\end{center}
\begin{center}
\begin{tikzpicture}[new set = import nodes]
\begin{scope}[nodes={set=import nodes}]%
		\foreach \i in {0,...,4} { \coordinate (\i) at (\i*1.2cm,0) ;}
		\foreach \i in {0} {\fill[red!20, rounded corners] (\i) +(-0.4,-0.25) rectangle +(0.4,0.35);}
		\foreach \i in {2} {\fill[red!20, rounded corners] (\i) +(-0.4,-0.25) rectangle +(1.7,0.35);}
		\foreach \i in {0,...,4} { \node (\i) at (\i) {\nscale{$h_{4,\i}$}};}
\graph {
	(import nodes);
			0 ->["\escale{$k$}"]1->["\escale{$x$}"]2->["\escale{$v_0$}"]3->["\escale{$k$}"]4;  
			};
\end{scope}
\begin{scope}[xshift=6cm,nodes={set=import nodes}]%
		\foreach \i in {0,...,4} { \coordinate (\i) at (\i*1.2cm,0) ;}
		\foreach \i in {0,3} {\fill[red!20, rounded corners] (\i) +(-0.4,-0.25) rectangle +(0.4,0.35);}
		\foreach \i in {0,...,4} { \node (\i) at (\i) {\nscale{$h_{5,\i}$}};}
\graph {
	(import nodes);
			0 ->["\escale{$k$}"]1->["\escale{$v_0$}"]2->["\escale{$x$}"]3->["\escale{$k$}"]4;  
			};
\end{scope}
\begin{scope}[yshift=-1.1cm,nodes={set=import nodes}]%
		\foreach \i in {0,...,4} { \coordinate (\i) at (\i*1.2cm,0) ;}
		\foreach \i in {0} {\fill[red!20, rounded corners] (\i) +(-0.4,-0.25) rectangle +(0.4,0.35);}
		\foreach \i in {2} {\fill[red!20, rounded corners] (\i) +(-0.4,-0.25) rectangle +(1.7,0.35);}
		\foreach \i in {0,...,4} { \node (\i) at (\i) {\nscale{$h_{6,\i}$}};}
\graph {
	(import nodes);
			0 ->["\escale{$k$}"]1->["\escale{$x$}"]2->["\escale{$y_0$}"]3->["\escale{$k$}"]4;  
			};
\end{scope}
\begin{scope}[xshift=6cm,yshift=-1.1cm,nodes={set=import nodes}]%
		\foreach \i in {0,...,4} { \coordinate (\i) at (\i*1.2cm,0) ;}
		\foreach \i in {0,3} {\fill[red!20, rounded corners] (\i) +(-0.4,-0.25) rectangle +(0.4,0.35);}
		\foreach \i in {0,...,4} { \node (\i) at (\i) {\nscale{$h_{7,\i}$}};}
\graph {
	(import nodes);
			0 ->["\escale{$k$}"]1->["\escale{$y_0$}"]2->["\escale{$x$}"]3->["\escale{$k$}"]4;  
			};
\end{scope}
\begin{scope}[yshift=-2.2cm,nodes={set=import nodes}]%
		\foreach \i in {0,...,4} { \coordinate (\i) at (\i*1.2cm,0) ;}
		\foreach \i in {0} {\fill[red!20, rounded corners] (\i) +(-0.4,-0.25) rectangle +(0.4,0.35);}
		\foreach \i in {2} {\fill[red!20, rounded corners] (\i) +(-0.4,-0.25) rectangle +(1.7,0.35);}
		\foreach \i in {0,...,4} { \node (\i) at (\i) {\nscale{$h_{8,\i}$}};}
\graph {
	(import nodes);
			0 ->["\escale{$k$}"]1->["\escale{$x$}"]2->["\escale{$y_1$}"]3->["\escale{$k$}"]4;  
			};
\end{scope}
\begin{scope}[xshift=6cm,yshift=-2.2cm,nodes={set=import nodes}]%
		\foreach \i in {0,...,4} { \coordinate (\i) at (\i*1.2cm,0) ;}
		\foreach \i in {0,3} {\fill[red!20, rounded corners] (\i) +(-0.4,-0.5) rectangle +(0.4,0.5);}
		\foreach \i in {2} {\fill[blue!20, rounded corners, opacity=0.7] (\i) +(-0.4,-0.25) rectangle +(1.6,0.35);}
		\foreach \i in {0,...,4} { \node (\i) at (\i) {\nscale{$h_{9,\i}$}};}
\graph {
	(import nodes);
			0 ->["\escale{$k$}"]1->["\escale{$y_1$}"]2->["\escale{$x$}"]3->["\escale{$k$}"]4;  
			};
\end{scope}
\begin{scope}[yshift=-3.3cm,nodes={set=import nodes}]%
		\foreach \i in {0,...,4} { \coordinate (\i) at (\i*1.2cm,0) ;}
		\foreach \i in {0} {\fill[red!20, rounded corners] (\i) +(-0.4,-0.5) rectangle +(0.4,0.5);}
		\foreach \i in {2} {\fill[red!20, rounded corners] (\i) +(-0.4,-0.5) rectangle +(1.5,0.5);}
		\foreach \i in {0,...,4} { \node (\i) at (\i) {\nscale{$h_{10,\i}$}};}
\graph {
	(import nodes);
			0 ->["\escale{$k$}"]1->["\escale{$x$}"]2->["\escale{$y_2$}"]3->["\escale{$k$}"]4;  
			};
\end{scope}
\begin{scope}[xshift=6cm,yshift=-3.3cm,nodes={set=import nodes}]%
		\foreach \i in {0,...,4} { \coordinate (\i) at (\i*1.2cm,0) ;}
		\foreach \i in {0} {\fill[red!20, rounded corners] (\i) +(-0.4,-0.5) rectangle +(0.4,0.5);}
		\foreach \i in {2} {\fill[red!20, rounded corners] (\i) +(-0.4,-0.5) rectangle +(1.5,0.5);}
		\foreach \i in {3} {\fill[blue!20, rounded corners, opacity=0.7] (\i) +(-0.4,-0.25) rectangle +(0.3,0.35);}
		\foreach \i in {0,...,4} { \node (\i) at (\i) {\nscale{$h_{11,\i}$}};}
\graph {
	(import nodes);
			0 ->["\escale{$k$}"]1->["\escale{$y_2$}"]2->["\escale{$x$}"]3->["\escale{$k$}"]4;  
			};
\end{scope}
\begin{scope}[yshift=-4.4cm,nodes={set=import nodes}]%
		\foreach \i in {0,...,5} { \coordinate (\i) at (\i*1.4cm,0) ;}
		\foreach \i in {0,4} {\fill[red!20, rounded corners] (\i) +(-0.4,-0.5) rectangle +(0.4,0.5);}
		\foreach \i in {3} {\fill[blue!20, rounded corners, opacity=0.7] (\i) +(-0.4,-0.25) rectangle +(1.8,0.35);}
		\foreach \i in {0,...,5} { \node (\i) at (\i) {\nscale{$h_{12,\i}$}};}
\graph {
	(import nodes);
			0 ->["\escale{$k$}"]1->["\escale{$y_0$}"]2->["\escale{$y_1$}"]3->["\escale{$y_2$}"]4->["\escale{$k$}"]5;  
			};
\end{scope}
\end{tikzpicture}
\end{center}
\begin{center}
\begin{tikzpicture}[new set = import nodes]
\begin{scope}[nodes={set=import nodes}]%
		\foreach \i in {0,...,4} { \coordinate (\i) at (\i*1.3cm,0) ;}
		\foreach \i in {0} {\fill[red!20, rounded corners] (\i) +(-0.4,-0.25) rectangle +(0.4,0.35);}
		\foreach \i in {2} {\fill[red!20, rounded corners] (\i) +(-0.4,-0.25) rectangle +(1.5,0.35);}
		\foreach \i in {0,...,4} { \node (\i) at (\i) {\nscale{$f_{0,\i}$}};}
\graph {
	(import nodes);
			0 ->["\escale{$k$}"]1->["\escale{$z_0$}"]2->["\escale{$v_0$}"]3->["\escale{$k$}"]4;
			};
\end{scope}
\begin{scope}[xshift=6.3cm, nodes={set=import nodes}]%
		\foreach \i in {0,...,4} { \coordinate (\i) at (\i*1.3cm,0) ;}
		\foreach \i in {0} {\fill[red!20, rounded corners] (\i) +(-0.4,-0.25) rectangle +(0.4,0.35);}
		\foreach \i in {2} {\fill[red!20, rounded corners] (\i) +(-0.4,-0.25) rectangle +(1.5,0.35);}
		\foreach \i in {0,...,4} { \node (\i) at (\i) {\nscale{$f_{1,\i}$}};}
\graph {
	(import nodes);
			0 ->["\escale{$k$}"]1->["\escale{$z_1$}"]2->["\escale{$v_0$}"]3->["\escale{$k$}"]4;
			};
\end{scope}
\begin{scope}[yshift=-1.2cm, nodes={set=import nodes}]
		
		\foreach \i in {0,...,5} { \coordinate (\i) at (\i*1.3cm,0) ;}
		\foreach \i in {0} {\fill[red!20, rounded corners] (\i) +(-0.4,-0.25) rectangle +(0.4,0.35);}
		\foreach \i in {2} {\fill[red!20, rounded corners] (\i) +(-0.4,-0.25) rectangle +(3,0.35);}
		\foreach \i in {0,...,5} { \node (\i) at (\i) {\nscale{$f_{2,\i}$}};}
\graph {
	(import nodes);
			0 ->["\escale{$k$}"]1->["\escale{$z_0$}"]2->["\escale{$z_1$}"]3->["\escale{$z_2$}"]4->["\escale{$k$}"]5;
			};
\end{scope}
\begin{scope}[yshift=-2.4cm,nodes={set=import nodes}]
		
		\foreach \i in {0,...,5} { \coordinate (\i) at (\i*1.05cm,0) ;}
		\foreach \i in {0} {\fill[red!20, rounded corners] (\i) +(-0.4,-0.25) rectangle +(0.4,0.35);}
		\foreach \i in {2} {\fill[red!20, rounded corners] (\i) +(-0.4,-0.25) rectangle +(2.4,0.35);}
		\foreach \i in {0,...,5} { \node (\i) at (\i) {\nscale{$g_{i,\i}$}};}
\graph {
	(import nodes);
			0 ->["\escale{$k$}"]1->["\escale{$z_0$}"]2->["\escale{$u_i$}"]3->["\escale{$z_1$}"]4->["\escale{$k$}"]5;
			};
\end{scope}
\begin{scope}[xshift=6.2cm,yshift=-2.4cm, nodes={set=import nodes}]
		
		\foreach \i in {0,...,5} { \coordinate (\i) at (\i*1.05cm,0) ;}
		\foreach \i in {0} {\fill[red!20, rounded corners] (\i) +(-0.4,-0.25) rectangle +(0.4,0.35);}
		\foreach \i in {2} {\fill[red!20, rounded corners] (\i) +(-0.4,-0.25) rectangle +(2.4,0.35);}
		\foreach \i in {0,...,5} { \node (\i) at (\i) {\nscale{$n_{i,\i}$}};}
\graph {
	(import nodes);
			0 ->["\escale{$k$}"]1->["\escale{$z_2$}"]2->["\escale{$u_i$}"]3->["\escale{$v_0$}"]4->["\escale{$k$}"]5;
			};
\end{scope}
\end{tikzpicture}
\end{center}
\begin{center}
\begin{tikzpicture}[new set = import nodes]
\begin{scope}[nodes={set=import nodes}]%
		
		\foreach \i in {0,...,9} { \coordinate (\i) at (\i*1.25cm,0) ;}
		\foreach \i in {0,8} {\fill[red!20, rounded corners] (\i) +(-0.3,-0.25) rectangle +(0.3,0.4);}
		\foreach \i in {0,...,9} { \node (\i) at (\i) {\nscale{$t_{i,0,\i}$}};}
\graph {
	(import nodes);
			0 ->["\escale{$k$}"]1->["\escale{$u_{0}$}"]2->["\escale{$X_{i_0}$}"]3->["\escale{$u_{1}$}"]4->["\escale{$X_{i_1}$}"]5->["\escale{$u_{2}$}"]6->["\escale{$X_{i_2}$}"]7->["\escale{$u_{3}$}"]8->["\escale{$k$}"]9;
			};
\end{scope}
\end{tikzpicture}
\end{center}
\begin{center}
\begin{tikzpicture}[new set = import nodes]
\begin{scope}[nodes={set=import nodes}]
		
		\foreach \i in {0,...,3} { \coordinate (\i) at (\i*1.3cm,0) ;}
		\foreach \i in {2} {\fill[red!20, rounded corners] (\i) +(-0.4,-0.25) rectangle +(0.4,0.4);}
		\foreach \i in {0,...,3} { \node (\i) at (\i) {\nscale{$b_{i,\i}$}};}
\graph {
	(import nodes);
			0 ->["\escale{$X_i$}"]1->["\escale{$u_{11}$}"]2->["\escale{$k$}"]3;
			};
\end{scope}

\end{tikzpicture}
\end{center}

($v_0$):
The red colored area of the following figure sketches the solvability of $(v_0,s)$ for all $s\in \{h_{0,0}, h_{2,0}, h_{4,0}, h_{4,1}, h_{5,0}\}$.
Notice that, while the red area sketches the case $sup(h_{4,0})=1, sup(h_{4,1})=0$, the case $sup(h_{4,0})=0, sup(h_{4,1})=1$ is also suitable, too.
This ensures the solvability of $(k,h_{4,0})$.
\begin{center}
\begin{tikzpicture}[new set = import nodes]
\begin{scope}[nodes={set=import nodes}]%

		\foreach \i in {0,...,2} { \coordinate (\i) at (\i*1.4cm,0) ;}
		\foreach \i in {1} {\fill[red!20, rounded corners] (\i) +(-0.4,-0.25) rectangle +(0.4,0.35);}
		\foreach \i in {0,...,2} { \node (\i) at (\i) {\nscale{$h_{0,\i}$}};}
\graph {(import nodes);
			0 ->["\escale{$k$}"]1->["\escale{$v_0$}"]2;  
			};
\end{scope}
\begin{scope}[xshift=5cm,nodes={set=import nodes}]%

		\foreach \i in {0,...,2} { \coordinate (\i) at (\i*1.4cm,0) ;}
		\foreach \i in {0,2} {\fill[red!20, rounded corners] (\i) +(-0.4,-0.25) rectangle +(0.4,0.35);}
		\foreach \i in {0,...,2} { \node (\i) at (\i) {\nscale{$h_{1,\i}$}};}
\graph {(import nodes);
			0 ->["\escale{$v_0$}"]1->["\escale{$k$}"]2;  
			};
\end{scope}
\begin{scope}[yshift=-1.1cm, nodes={set=import nodes}]%
		
		\foreach \i in {0,...,4} { \coordinate (\i) at (\i*1.4cm,0) ;}
		\foreach \i in {1,4} {\fill[red!20, rounded corners] (\i) +(-0.4,-0.25) rectangle +(0.4,0.35);}
		\foreach \i in {0,...,4} { \node (\i) at (\i) {\nscale{$h_{2,\i}$}};}
\graph {(import nodes);
			0 ->["\escale{$k$}"]1->["\escale{$v_0$}"]2->["\escale{$v_1$}"]3->["\escale{$k$}"]4;  
			};
\end{scope}

\begin{scope}[xshift=6.9cm, yshift=-1.1cm,nodes={set=import nodes}]%
		
		\foreach \i in {0,...,3} { \coordinate (\i) at (\i*1.4cm,0) ;}
		\foreach \i in {1} {\fill[red!20, rounded corners] (\i) +(-0.4,-0.25) rectangle +(1.7,0.35);}
		\foreach \i in {0,...,3} { \node (\i) at (\i) {\nscale{$h_{3,\i}$}};}
\graph {
	(import nodes);
			0 ->["\escale{$k$}"]1->["\escale{$v_1$}"]2->["\escale{$v_0$}"]3;  
			};
\end{scope}
\end{tikzpicture}
\end{center}
\begin{center}
\begin{tikzpicture}[new set = import nodes]
\begin{scope}[nodes={set=import nodes}]%
		\foreach \i in {0,...,4} { \coordinate (\i) at (\i*1.2cm,0) ;}
		\foreach \i in {0,2,4} {\fill[red!20, rounded corners] (\i) +(-0.4,-0.25) rectangle +(0.4,0.35);}
		\foreach \i in {0,...,4} { \node (\i) at (\i) {\nscale{$h_{4,\i}$}};}
\graph {
	(import nodes);
			0 ->["\escale{$k$}"]1->["\escale{$x$}"]2->["\escale{$v_0$}"]3->["\escale{$k$}"]4;  
			};
\end{scope}
\begin{scope}[xshift=6cm,nodes={set=import nodes}]%
		\foreach \i in {0,...,4} { \coordinate (\i) at (\i*1.2cm,0) ;}
		\foreach \i in {1,3} {\fill[red!20, rounded corners] (\i) +(-0.4,-0.25) rectangle +(0.4,0.35);}
		\foreach \i in {0,...,4} { \node (\i) at (\i) {\nscale{$h_{5,\i}$}};}
\graph {
	(import nodes);
			0 ->["\escale{$k$}"]1->["\escale{$v_0$}"]2->["\escale{$x$}"]3->["\escale{$k$}"]4;  
			};
\end{scope}
\begin{scope}[yshift=-1.1cm,nodes={set=import nodes}]%
		\foreach \i in {0,...,4} { \coordinate (\i) at (\i*1.2cm,0) ;}
		\foreach \i in {0} {\fill[red!20, rounded corners] (\i) +(-0.4,-0.25) rectangle +(0.4,0.35);}
		\foreach \i in {2} {\fill[red!20, rounded corners] (\i) +(-0.4,-0.25) rectangle +(1.7,0.35);}
		\foreach \i in {0,...,4} { \node (\i) at (\i) {\nscale{$h_{6,\i}$}};}
\graph {
	(import nodes);
			0 ->["\escale{$k$}"]1->["\escale{$x$}"]2->["\escale{$y_0$}"]3->["\escale{$k$}"]4;  
			};
\end{scope}
\begin{scope}[xshift=6cm,yshift=-1.1cm,nodes={set=import nodes}]%
		\foreach \i in {0,...,4} { \coordinate (\i) at (\i*1.2cm,0) ;}
		\foreach \i in {0,3} {\fill[red!20, rounded corners] (\i) +(-0.4,-0.25) rectangle +(0.4,0.35);}
		\foreach \i in {0,...,4} { \node (\i) at (\i) {\nscale{$h_{7,\i}$}};}
\graph {
	(import nodes);
			0 ->["\escale{$k$}"]1->["\escale{$y_0$}"]2->["\escale{$x$}"]3->["\escale{$k$}"]4;  
			};
\end{scope}
\begin{scope}[yshift=-2.2cm,nodes={set=import nodes}]%
		\foreach \i in {0,...,4} { \coordinate (\i) at (\i*1.2cm,0) ;}
		\foreach \i in {0} {\fill[red!20, rounded corners] (\i) +(-0.4,-0.25) rectangle +(0.4,0.35);}
		\foreach \i in {2} {\fill[red!20, rounded corners] (\i) +(-0.4,-0.25) rectangle +(1.7,0.35);}
		\foreach \i in {0,...,4} { \node (\i) at (\i) {\nscale{$h_{8,\i}$}};}
\graph {
	(import nodes);
			0 ->["\escale{$k$}"]1->["\escale{$x$}"]2->["\escale{$y_1$}"]3->["\escale{$k$}"]4;  
			};
\end{scope}
\begin{scope}[xshift=6cm,yshift=-2.2cm,nodes={set=import nodes}]%
		\foreach \i in {0,...,4} { \coordinate (\i) at (\i*1.2cm,0) ;}
		\foreach \i in {0,3} {\fill[red!20, rounded corners] (\i) +(-0.4,-0.3) rectangle +(0.4,0.3);}
		\foreach \i in {0,...,4} { \node (\i) at (\i) {\nscale{$h_{9,\i}$}};}
\graph {
	(import nodes);
			0 ->["\escale{$k$}"]1->["\escale{$y_1$}"]2->["\escale{$x$}"]3->["\escale{$k$}"]4;  
			};
\end{scope}
\begin{scope}[yshift=-3.3cm,nodes={set=import nodes}]%
		\foreach \i in {0,...,4} { \coordinate (\i) at (\i*1.2cm,0) ;}
		\foreach \i in {0} {\fill[red!20, rounded corners] (\i) +(-0.4,-0.3) rectangle +(0.4,0.3);}
		\foreach \i in {2} {\fill[red!20, rounded corners] (\i) +(-0.4,-0.3) rectangle +(1.7,0.3);}
		\foreach \i in {0,...,4} { \node (\i) at (\i) {\nscale{$h_{10,\i}$}};}
\graph {
	(import nodes);
			0 ->["\escale{$k$}"]1->["\escale{$x$}"]2->["\escale{$y_2$}"]3->["\escale{$k$}"]4;  
			};
\end{scope}
\begin{scope}[xshift=6cm,yshift=-3.3cm,nodes={set=import nodes}]%
		\foreach \i in {0,...,4} { \coordinate (\i) at (\i*1.2cm,0) ;}
		\foreach \i in {0,3} {\fill[red!20, rounded corners] (\i) +(-0.4,-0.3) rectangle +(0.4,0.3);}
		\foreach \i in {0,...,4} { \node (\i) at (\i) {\nscale{$h_{11,\i}$}};}
\graph {
	(import nodes);
			0 ->["\escale{$k$}"]1->["\escale{$y_2$}"]2->["\escale{$x$}"]3->["\escale{$k$}"]4;  
			};
\end{scope}
\begin{scope}[yshift=-4.4cm,nodes={set=import nodes}]%
		\foreach \i in {0,...,5} { \coordinate (\i) at (\i*1.4cm,0) ;}
		\foreach \i in {0,5} {\fill[red!20, rounded corners] (\i) +(-0.4,-0.3) rectangle +(0.4,0.3);}
		\foreach \i in {0,...,5} { \node (\i) at (\i) {\nscale{$h_{12,\i}$}};}
\graph {
	(import nodes);
			0 ->["\escale{$k$}"]1->["\escale{$y_0$}"]2->["\escale{$y_1$}"]3->["\escale{$y_2$}"]4->["\escale{$k$}"]5;  
			};
\end{scope}
\end{tikzpicture}
\end{center}
\begin{center}
\begin{tikzpicture}[new set = import nodes]
\begin{scope}[nodes={set=import nodes}]%
		\foreach \i in {0,...,4} { \coordinate (\i) at (\i*1.3cm,0) ;}
		\foreach \i in {4} {\fill[red!20, rounded corners] (\i) +(-0.4,-0.25) rectangle +(0.4,0.35);}
		\foreach \i in {1} {\fill[red!20, rounded corners] (\i) +(-0.4,-0.25) rectangle +(1.7,0.35);}
		\foreach \i in {0,...,4} { \node (\i) at (\i) {\nscale{$f_{0,\i}$}};}
\graph {
	(import nodes);
			0 ->["\escale{$k$}"]1->["\escale{$z_0$}"]2->["\escale{$v_0$}"]3->["\escale{$k$}"]4;
			};
\end{scope}
\begin{scope}[xshift=6.3cm, nodes={set=import nodes}]%
		\foreach \i in {0,...,4} { \coordinate (\i) at (\i*1.3cm,0) ;}
		\foreach \i in {4} {\fill[red!20, rounded corners] (\i) +(-0.4,-0.25) rectangle +(0.4,0.35);}
		\foreach \i in {1} {\fill[red!20, rounded corners] (\i) +(-0.4,-0.25) rectangle +(1.7,0.35);}
		\foreach \i in {0,...,4} { \node (\i) at (\i) {\nscale{$f_{1,\i}$}};}
\graph {
	(import nodes);
			0 ->["\escale{$k$}"]1->["\escale{$z_1$}"]2->["\escale{$v_0$}"]3->["\escale{$k$}"]4;
			};
\end{scope}
\begin{scope}[yshift=-1.2cm, nodes={set=import nodes}]
		
		\foreach \i in {0,...,5} { \coordinate (\i) at (\i*1.3cm,0) ;}
		\foreach \i in {0,5} {\fill[red!20, rounded corners] (\i) +(-0.4,-0.25) rectangle +(0.4,0.35);}
		\foreach \i in {0,...,5} { \node (\i) at (\i) {\nscale{$f_{2,\i}$}};}
\graph {
	(import nodes);
			0 ->["\escale{$k$}"]1->["\escale{$z_0$}"]2->["\escale{$z_1$}"]3->["\escale{$z_2$}"]4->["\escale{$k$}"]5;
			};
\end{scope}
\begin{scope}[yshift=-2.4cm,nodes={set=import nodes}]
		
		\foreach \i in {0,...,5} { \coordinate (\i) at (\i*1.05cm,0) ;}
		\foreach \i in {0,5} {\fill[red!20, rounded corners] (\i) +(-0.4,-0.25) rectangle +(0.4,0.35);}
		\foreach \i in {0,...,5} { \node (\i) at (\i) {\nscale{$g_{i,\i}$}};}
\graph {
	(import nodes);
			0 ->["\escale{$k$}"]1->["\escale{$z_0$}"]2->["\escale{$u_i$}"]3->["\escale{$z_1$}"]4->["\escale{$k$}"]5;
			};
\end{scope}
\begin{scope}[xshift=6.2cm,yshift=-2.4cm, nodes={set=import nodes}]
		
		\foreach \i in {0,...,5} { \coordinate (\i) at (\i*1.05cm,0) ;}
		\foreach \i in {0,5} {\fill[red!20, rounded corners] (\i) +(-0.4,-0.25) rectangle +(0.4,0.35);}
		\foreach \i in {0,...,5} { \node (\i) at (\i) {\nscale{$n_{i,\i}$}};}
\graph {
	(import nodes);
			0 ->["\escale{$k$}"]1->["\escale{$z_2$}"]2->["\escale{$u_i$}"]3->["\escale{$v_0$}"]4->["\escale{$k$}"]5;
			};
\end{scope}
\end{tikzpicture}
\end{center}
\begin{center}
\begin{tikzpicture}[new set = import nodes]
\begin{scope}[nodes={set=import nodes}]%
		
		\foreach \i in {0,...,9} { \coordinate (\i) at (\i*1.25cm,0) ;}
		\foreach \i in {0,9} {\fill[red!20, rounded corners] (\i) +(-0.3,-0.25) rectangle +(0.3,0.4);}
		\foreach \i in {0,...,9} { \node (\i) at (\i) {\nscale{$t_{i,0,\i}$}};}
\graph {
	(import nodes);
			0 ->["\escale{$k$}"]1->["\escale{$u_{0}$}"]2->["\escale{$X_{i_0}$}"]3->["\escale{$u_{1}$}"]4->["\escale{$X_{i_1}$}"]5->["\escale{$u_{2}$}"]6->["\escale{$X_{i_2}$}"]7->["\escale{$u_{3}$}"]8->["\escale{$k$}"]9;
			};
\end{scope}
\end{tikzpicture}
\end{center}
\begin{center}
\begin{tikzpicture}[new set = import nodes]
\begin{scope}[nodes={set=import nodes}]
		
		\foreach \i in {0,...,3} { \coordinate (\i) at (\i*1.3cm,0) ;}
		\foreach \i in {3} {\fill[red!20, rounded corners] (\i) +(-0.4,-0.25) rectangle +(0.4,0.4);}
		\foreach \i in {0,...,3} { \node (\i) at (\i) {\nscale{$b_{i,\i}$}};}
\graph {
	(import nodes);
			0 ->["\escale{$X_i$}"]1->["\escale{$u_{11}$}"]2->["\escale{$k$}"]3;
			};
\end{scope}

\end{tikzpicture}
\end{center}
The red area of the following figure sketches a region $(sup, sig)$ that proves the solvability of $(v_0,s)$ for all $s\in \{h_{1,2}, h_{2,4},h_{4,4}\}\cup\bigcup_{i=6}^{12}S(H_i)\cup S(F_0)\cup S(F_1)$. 
If we modify this region by $sup(h_{2,3})=sup(h_{2,4})=1$ and $sup(h_{3,0})=sup(h_{3,1})=0$ (and $sig(v_1)=\set$), then we get a region that solves $(v_0, h_{3,0})$ and $(v_0, h_{3,1})$.
It is easy to see that the remaining atoms $(v_0,s)$ (especially $(v_0,n_{0,2}),\dots, (v_0,n_{13,2})$) are solvable, too.
\begin{center}
\begin{tikzpicture}[new set = import nodes]
\begin{scope}[nodes={set=import nodes}]%

		\foreach \i in {0,...,2} { \coordinate (\i) at (\i*1.4cm,0) ;}
		\foreach \i in {0} {\fill[red!20, rounded corners] (\i) +(-0.4,-0.25) rectangle +(1.7,0.35);}
		\foreach \i in {0,...,2} { \node (\i) at (\i) {\nscale{$h_{0,\i}$}};}
\graph {(import nodes);
			0 ->["\escale{$k$}"]1->["\escale{$v_0$}"]2;  
			};
\end{scope}
\begin{scope}[xshift=5cm,nodes={set=import nodes}]%

		\foreach \i in {0,...,2} { \coordinate (\i) at (\i*1.4cm,0) ;}
		\foreach \i in {0} {\fill[red!20, rounded corners] (\i) +(-0.4,-0.25) rectangle +(1.7,0.35);}
		\foreach \i in {0,...,2} { \node (\i) at (\i) {\nscale{$h_{1,\i}$}};}
\graph {(import nodes);
			0 ->["\escale{$v_0$}"]1->["\escale{$k$}"]2;  
			};
\end{scope}
\begin{scope}[yshift=-1.1cm, nodes={set=import nodes}]%
		
		\foreach \i in {0,...,4} { \coordinate (\i) at (\i*1.4cm,0) ;}
		\foreach \i in {0} {\fill[red!20, rounded corners] (\i) +(-0.4,-0.25) rectangle +(1.8,0.35);}
		\foreach \i in {0,...,4} { \node (\i) at (\i) {\nscale{$h_{2,\i}$}};}
\graph {(import nodes);
			0 ->["\escale{$k$}"]1->["\escale{$v_0$}"]2->["\escale{$v_1$}"]3->["\escale{$k$}"]4;  
			};
\end{scope}

\begin{scope}[xshift=6.9cm, yshift=-1.1cm,nodes={set=import nodes}]%
		
		\foreach \i in {0,...,3} { \coordinate (\i) at (\i*1.4cm,0) ;}
		\foreach \i in {0} {\fill[red!20, rounded corners] (\i) +(-0.4,-0.25) rectangle +(3.2,0.35);}
		\foreach \i in {0,...,3} { \node (\i) at (\i) {\nscale{$h_{3,\i}$}};}
\graph {
	(import nodes);
			0 ->["\escale{$k$}"]1->["\escale{$v_1$}"]2->["\escale{$v_0$}"]3;  
			};
\end{scope}
\end{tikzpicture}
\end{center}
\begin{center}
\begin{tikzpicture}[new set = import nodes]
\begin{scope}[nodes={set=import nodes}]%
		\foreach \i in {0,...,4} { \coordinate (\i) at (\i*1.2cm,0) ;}
		\foreach \i in {0} {\fill[red!20, rounded corners] (\i) +(-0.4,-0.25) rectangle +(2.7,0.35);}
		\foreach \i in {0,...,4} { \node (\i) at (\i) {\nscale{$h_{4,\i}$}};}
\graph {
	(import nodes);
			0 ->["\escale{$k$}"]1->["\escale{$x$}"]2->["\escale{$v_0$}"]3->["\escale{$k$}"]4;  
			};
\end{scope}
\begin{scope}[xshift=6cm,nodes={set=import nodes}]%
		\foreach \i in {0,...,4} { \coordinate (\i) at (\i*1.2cm,0) ;}
		\foreach \i in {0} {\fill[red!20, rounded corners] (\i) +(-0.4,-0.25) rectangle +(1.5,0.35);}
		\foreach \i in {0,...,4} { \node (\i) at (\i) {\nscale{$h_{5,\i}$}};}
\graph {
	(import nodes);
			0 ->["\escale{$k$}"]1->["\escale{$v_0$}"]2->["\escale{$x$}"]3->["\escale{$k$}"]4;  
			};
\end{scope}
\end{tikzpicture}
\end{center}
\begin{center}
\begin{tikzpicture}[new set = import nodes]
\begin{scope}[nodes={set=import nodes}]%
		\foreach \i in {0,...,4} { \coordinate (\i) at (\i*1.3cm,0) ;}
		\foreach \i in {2} {\fill[red!20, rounded corners] (\i) +(-0.4,-0.25) rectangle +(0.4,0.35);}
		\foreach \i in {0,...,4} { \node (\i) at (\i) {\nscale{$f_{0,\i}$}};}
\graph {
	(import nodes);
			0 ->["\escale{$k$}"]1->["\escale{$z_0$}"]2->["\escale{$v_0$}"]3->["\escale{$k$}"]4;
			};
\end{scope}
\begin{scope}[xshift=6.3cm, nodes={set=import nodes}]%
		\foreach \i in {0,...,4} { \coordinate (\i) at (\i*1.3cm,0) ;}
		\foreach \i in {2} {\fill[red!20, rounded corners] (\i) +(-0.4,-0.25) rectangle +(0.4,0.35);}
		\foreach \i in {0,...,4} { \node (\i) at (\i) {\nscale{$f_{1,\i}$}};}
\graph {
	(import nodes);
			0 ->["\escale{$k$}"]1->["\escale{$z_1$}"]2->["\escale{$v_0$}"]3->["\escale{$k$}"]4;
			};
\end{scope}
\begin{scope}[yshift=-1.2cm, nodes={set=import nodes}]
		
		\foreach \i in {0,...,5} { \coordinate (\i) at (\i*1.3cm,0) ;}
		\foreach \i in {2} {\fill[red!20, rounded corners] (\i) +(-0.4,-0.25) rectangle +(4.2,0.35);}
		\foreach \i in {0,...,5} { \node (\i) at (\i) {\nscale{$f_{2,\i}$}};}
\graph {
	(import nodes);
			0 ->["\escale{$k$}"]1->["\escale{$z_0$}"]2->["\escale{$z_1$}"]3->["\escale{$z_2$}"]4->["\escale{$k$}"]5;
			};
\end{scope}
\begin{scope}[yshift=-2.4cm,nodes={set=import nodes}]
		
		\foreach \i in {0,...,5} { \coordinate (\i) at (\i*1.05cm,0) ;}
		\foreach \i in {2} {\fill[red!20, rounded corners] (\i) +(-0.4,-0.25) rectangle +(3.4,0.35);}
		\foreach \i in {0,...,5} { \node (\i) at (\i) {\nscale{$g_{i,\i}$}};}
\graph {
	(import nodes);
			0 ->["\escale{$k$}"]1->["\escale{$z_0$}"]2->["\escale{$u_i$}"]3->["\escale{$z_1$}"]4->["\escale{$k$}"]5;
			};
\end{scope}
\begin{scope}[xshift=6.2cm,yshift=-2.4cm, nodes={set=import nodes}]
		
		\foreach \i in {0,...,5} { \coordinate (\i) at (\i*1.05cm,0) ;}
		\foreach \i in {2} {\fill[red!20, rounded corners] (\i) +(-0.4,-0.25) rectangle +(1.4,0.35);}
		\foreach \i in {0,...,5} { \node (\i) at (\i) {\nscale{$n_{i,\i}$}};}
\graph {
	(import nodes);
			0 ->["\escale{$k$}"]1->["\escale{$z_2$}"]2->["\escale{$u_i$}"]3->["\escale{$v_0$}"]4->["\escale{$k$}"]5;
			};
\end{scope}
\end{tikzpicture}
\end{center}

Notice that all other events can occur in multiple gadgets of $A^\tau_\varphi$, but at most once in a fixed gadget of $A^\tau_\varphi$.
Using this, one finds out that all the other separation atoms are solvable, too.
However, for the sake of simplicity, we refrain from a tedious case analyses and skip the enumeration of the corresponding regions. 
%
\end{appendix}
\end{document}